\documentclass[11pt]{article}
\usepackage{amssymb,amsmath,amsthm,bbm,enumitem,mathtools,dsfont,bm}
\usepackage{graphicx,caption}
\usepackage{subcaption}
\usepackage[dvipsnames]{xcolor}
\usepackage{tikz}
\usepackage[a4paper, total={6.7in, 10in}]{geometry}  
\usepackage{setspace}
\usepackage[hidelinks]{hyperref}
\usepackage{authblk}
\usepackage{textcomp}
\usepackage[title]{appendix}

\usepackage{apacite} 
\AtBeginDocument{%

   \urlstyle{APACsame}%
}
\usepackage[round]{natbib}    

\usetikzlibrary {arrows.meta}

\newtheorem{theorem}{Theorem}[section]
\theoremstyle{definition}
\newtheorem{definition}{Definition}[section]
\theoremstyle{definition}
\newtheorem{example}{Example}[section]
\theoremstyle{definition}
\newtheorem{assumption}{Assumption}[section]
\theoremstyle{definition}
\newtheorem{proposition}{Proposition}[section]
\theoremstyle{definition}
\newtheorem{corollary}{Corollary}[section]
\theoremstyle{remark}
\newtheorem{remark}{Remark}[section]
\theoremstyle{definition}
\newtheorem{lemma}{Lemma}[section]

\newcommand{\I}{\mathbf{I}}
\newcommand{\J}{\mathbf{J}}

\newcommand{\Y}{\mathbf{Y}}
\newcommand{\Nn}{\mathcal{N}_n}

\newcommand{\In}{\mathcal{I}_n}

\DeclareMathOperator{\real}{\mathbb{R}}

\newcommand{\nzp}{\text{nzp}}
\newcommand{\An}{\mathcal{A}_n}

\allowdisplaybreaks

\DeclareMathOperator*{\argmin}{arg\,min}

\setlist[itemize]{leftmargin=0.4cm,labelindent=\parindent}

\title{\bf Pareto Optimal Centralized Risk Sharing with Multiple Agents: Inclusivity and Fairness}
 
\author[$\star$]{Debora Daniela Escobar}
\affil[$\star$]{Department of Actuarial Mathematics and Statistics, Maxwell Institute for Mathematical Sciences, Heriot-Watt University, Edinburgh, United Kingdom. Email: d.escobar@hw.ac.uk.}
\author[$\ddagger$]{Wing Fung Chong}
\affil[$\ddagger$]{Department of Statistics and Actuarial Science, School of Computing and Data Science, The University of Hong Kong, Hong Kong, China. Email: chongwf@hku.hk.}
\date{\today}

\begin{document}

\large

\sloppy
 
\maketitle

\begin{abstract}
This paper studies centralized risk sharing with endogenous prices. Multiple policyholders transfer risks to a central insurer through indemnity decisions, while prices are determined by pricing functionals applied to ceded risks. The resulting problem is multiobjective, with Pareto optimality as the natural efficiency criterion. We show that classical Pareto optimality may fail to reveal whether all agents are represented in a balanced decision process that scalarized objectives may assign zero weight to some agents, and group aggregates may obscure individual risk positions. Motivated by the bilateral Pareto characterization through sequential optimization, we introduce inclusive and fair Pareto optimality, a representation-based refinement requiring every agent to appear exactly once, either individually or as part of a group, in a finite ordered sequence of optimizations. Our main result proves equivalence between this concept and balanced sequential optimization, placing it between Geoffrion-proper Pareto optimality and classical Pareto optimality. An illustrative example demonstrates the framework using the Expected Shortfall.

\end{abstract}

\section{Introduction}
Optimal insurance and reinsurance design is a classical problem in actuarial science, insurance economics, and risk management. Its modern formulation can be traced back to the seminal contributions of \cite{borch1960attempt},  \cite{arrow1963uncertainty, Arrow1974}, and \cite{miller1972}, who derived optimal deductible-type contracts under variance-minimization and expected-utility-maximization criteria. These foundational insights have generated a large literature on optimal insurance and reinsurance, studying how optimal contracts depend on preferences, beliefs, premium principles, risk measures, moral hazard, adverse selection, regulatory constraints, and other market frictions; see, for example, \cite{CAI2008}, \cite{Kaluszka2008}, \cite{Cheung2010}, \cite{ChiTan2011}, \cite{sung2011}, \cite{CUI2013}, \cite{Bernard2015}, \cite{Cheung2015}, \cite{Zhuang2016}, \cite{CAI2017}, \cite{ASIMIT2018}, \cite{Cheung2019}, \cite{lo2019pareto}, \cite{Xu2019}, \cite{chi2020optimal}, \cite{ASIMIT2021587}, \cite{BoonenZhang2021}, \cite{BOONEN2022}, \cite{boonen2023}, \cite{WANG2026261}, and the references therein. Despite the richness of this literature, much of it treats insurance as a bilateral risk-sharing problem between one policyholder and one insurer or reinsurer.

In practice, insurance appears bilateral from the policyholder's perspective, but multilateral from the insurer's perspective. A single insurer pools risks from many policyholders, collects premiums ex-ante, and indemnifies realized losses ex-post. More generally, centralized risk sharing among multiple agents refers to a scheme in which each policyholder transfers their risk to a common central agent, while recognizing that other policyholders are also participating in the arrangement. Insurance can therefore be viewed as a centralized risk-sharing scheme, in which a selected group of agents shares risk through a central agent. This centralized structure has recently attracted renewed attention by \cite{BERNARD2020}. \cite{boonenchongghossub2024}
provide a microeconomic analysis of Pareto-efficient risk sharing in centralized insurance markets, with applications to flood insurance. Further developments include Stackelberg equilibria in \cite{GHOSSOUBSTACK2024}, \cite{boonen2026pareto}, and robustness under dependence uncertainty in \cite{boonen2025pareto}. In these models, the interaction among policyholders and the insurer naturally gives rise to a multiobjective optimization problem, and Pareto optimality, in the tradition of \cite{pareto1906manuale}, serves as the canonical equilibrium notion.

A key feature of much of this recent centralized-insurance literature is that premiums are treated as exogenous decision variables. This formulation makes it possible to trace the Pareto frontier by varying deterministic transfers and to apply scalarization techniques similar in spirit to those developed in \cite{ASIMIT2018} and \cite{ASIMIT2021587} with equal bargaining powers. Such a price-as-decision-variable perspective is natural in market-value or economic equilibrium models. In many insurance applications, however, premiums are not freely chosen decisions. Rather, they are determined by premium principles applied to the indemnity coverage. In that case, premiums are endogenous functions of the ceding decisions. Changing an indemnity therefore changes both the random risk allocation and the associated premium simultaneously. Consequently, the usual deterministic-transfer argument for tracing the Pareto frontier is no longer directly available.

This paper studies multilateral centralized risk sharing under this endogenous pricing structure. We consider a centralized insurance problem with multiple policyholders and one insurer. Each policyholder chooses an indemnity, or ceding function, the insurer charges a premium determined by a given pricing functional, and every agent evaluates the resulting risk position through an objective functional. The premium charged to each policyholder depends on the ceded risk induced by that policyholder’s indemnity. Thus, the insurer’s and policyholders’ objectives are coupled through the indemnity decisions, while premiums themselves are not independent choice variables. This distinction is crucial to this paper. It leads to a multiobjective optimization problem in which Pareto optimality remains the natural starting point, but where its usual scalarization-based interpretations become less informative about the role played by each agent in the decision process, such as whether they still have equal bargaining powers.

The characterization of Pareto optimal solutions has long been a central topic in multiobjective optimization and vector optimization. Since Pareto optimality is defined through non-dominance rather than by a single objective value, much of the literature has sought constructive ways to represent Pareto optimal solutions through scalar or constrained optimization problems. A classical approach is weighted-sum scalarization, where the multiple objectives are aggregated into one weighted objective. Positive weights provide a direct route to Pareto optimality, while convexity of the feasible objective region allows supporting-hyperplane arguments to recover Pareto optimal solutions through suitable non-negative weights. This scalarization viewpoint is foundational in multiobjective mathematical programming; see \cite{evans1984} for a classical overview. However, even in convex settings, the resulting scalarization need not assign positive weight to every objective. Some objectives may receive zero weight, and therefore may be absent from the scalar representation. This is not merely a technical point in multilateral decision problems; when objectives correspond to agents, zero weights mean that some agents' objectives may not be represented in the scalar characterization. Related Pareto-optimality characterization by weighted-sum scalarization for multi-agent investment and comonotonic allocations have been studied; see, for example, \cite{xia2004multi,ravanelli2014comonotone}.

A second major characterization approach is the $\varepsilon$-constraint method, usually traced to \cite{haimes1971bicriterion}. Instead of aggregating all objectives into one weighted objective, this method optimizes one objective while imposing bounds on the remaining objectives. It therefore provides a constrained scalar characterization of Pareto optimality. Its strength is that it does not require all objectives to enter the same scalar aggregate. Its limitation, however, is that the relevant constraint bounds are tied to the Pareto optimal solution being characterized. In other words, to characterize a given Pareto optimal solution by this method, one typically needs to know the other objective values at that very solution in advance. The characterization is therefore implicit; it explains how a Pareto optimum can be certified once the appropriate bounds are known, but it does not by itself provide an explicit representation of how the objectives should enter the decision process. This issue has motivated adaptive and exact $\varepsilon$-constraint schemes for Pareto-frontier generation; see, for example, \cite{LAUMANNS2006}, \cite{BERUBE2009}, and \cite{Cooper2020}.

A third line of work refines Pareto optimality itself. Geoffrion's proper efficiency strengthens Pareto optimality by ruling out efficient points with unbounded or anomalous trade-offs among objectives, thereby producing a notion with sharper scalarization and separation properties; see \cite{GEOFFRION1968}. Later developments on super efficiency by \cite{borwein1993super}, and scalarizations for multiple proper-efficiency notions by \cite{GarciaCastanoFernando2023}, continue this theme by refining efficiency to obtain stronger bounded-trade-off, stability, or scalarization properties. Related characterization tools include nonlinear scalarizations, such as achievement scalarizing functions (\cite{wierzbicki1980}) and Pascoletti-Serafini scalarization (\cite{PascolettiA1984}), Karush-Kuhn-Tucker-type conditions, and outcome-space or approximation methods for nondominated frontiers; see \cite{benson1998outer} and \cite{herzel2021}.

In summary, weighted-sum scalarization may support a Pareto optimal decision, but some agents can receive zero weight; $\varepsilon$-constraint methods can certify a Pareto optimum, but the required bounds depend on the solution itself; and Proper efficiency refines Pareto optimality by controlling the magnitude of trade-offs. This paper contributes to this characterization literature by introducing a different refinement of Pareto optimality, motivated by the multilateral centralized risk sharing problem with endogenous prices. 
We introduce inclusive and fair Pareto optimality, a representation-based refinement of Pareto optimality. It identifies those Pareto optimal decisions that admit a finite sequential optimization representation in which every agent is included exactly once, either individually or as part of a group.

Our first contribution is to revisit the bilateral case from the perspective of sequential optimization. Building on the characterization of Pareto optimal insurance decisions in \cite{CAI2017}, 
we reframe the bilateral Pareto optimality 
through a small family of sequential optimization procedures. Such optimizations are composed as the union of the interior case and boundary cases. In the interior case, one minimizes a strict convex combination of the two agents' objectives. At the boundary, one first optimizes one agent's objective and then, among the corresponding minimizers, optimizes the other agent's objective. This reformulation is not merely notational. It reveals that the bilateral characterization of Pareto optimality is fundamentally sequential. We encode these procedures through sequential matrices and a recursively defined set-valued function, which provide the conceptual and technical starting point for the multilateral analysis.

Our second contribution is to show why the bilateral sequential characterization does not extend directly to the multilateral centralized problem. We introduce a class of sequential matrices in which each column represents one stage of optimization and each row represents one agent. The natural balanced class requires every agent to appear exactly once across the sequence, either individually or as part of a group. We show that every decision generated by such a balanced sequential optimization is Pareto optimal. This extends the sufficient direction of the bilateral theory. The converse, however, fails in general. The reason is structural, that in a multilateral problem, optimizing a convex combination of several agents' objectives identifies only an aggregate group value, not the individual risk positions within that group. Hence an improvement in one agent's objective may be offset by a deterioration in another's without violating optimality of the aggregate. We further show that enlarged classes of sequential procedures may relate to Pareto optimal decisions even when some agents are excluded, or when some agents are optimized more often than others. These observations demonstrate that the classical Pareto optimality alone does not distinguish between balanced and imbalanced sequential representations.

Our third contribution is to introduce and characterize {\it inclusive and fair Pareto optimality}. The concept is designed to capture exactly those Pareto optimal decisions that admit a balanced sequential representation. {\it Inclusivity} requires that every agent's objective be represented in the optimization sequence. {\it Fairness}, in the sense of representation, requires that each agent appears exactly once, either alone or as part of a group. To define the concept, we introduce convex-group risk measures, Pareto optimality between two groups of agents, and ordered set partitions of the agent set. An inclusive and fair Pareto optimal decision is obtained by first optimizing the convex-group risk of an initial group and then recursively incorporating additional groups through Pareto optimality between all previously included agents and the next group.

Our main result proves that inclusive and fair Pareto optimality is exactly characterized by balanced sequential optimization. Equivalently, the set of inclusive and fair Pareto optimal decisions coincides with the union of the recursively generated solution sets associated with the balanced sequential matrices. Moreover, decisions obtained by minimizing a strict convex combination of all agents' objectives in one step are inclusive and fair Pareto optimal. Thus, in the terminology of convex multiobjective optimization, Geoffrion-proper Pareto optimality implies inclusive and fair Pareto optimality, which in turn implies classical Pareto optimality. The new concept is therefore weaker than proper Pareto optimality but stronger than classical Pareto optimality. It preserves the tractable sequential structure of the bilateral theory through groups, while ruling out exclusive or unfair sequential procedures that the classical Pareto optimality alone may not account for.

The remainder of this paper is organized as follows. Section~\ref{sec:problem} introduces the multilateral centralized risk-sharing problem with endogenous prices and defines the classical Pareto optimality. Section~\ref{sec:PO_sequential} revisits the bilateral case, reformulates the known Pareto characterization through sequential optimization, and then extends the sequential framework to the multilateral case. It also explains why the bilateral characterization fails to extend directly and how the classical Pareto optimality relates to exclusive or unfair sequential procedures. Section~\ref{sec:iPO} introduces convex-group risk measures, Pareto optimality between groups, ordered set partitions, and inclusive and fair Pareto optimality. The main characterization theorem establishes the equivalence between inclusive and fair Pareto optimality and balanced sequential optimization. Section~\ref{sec:example_section} provides an illustrative example with two policyholders and one insurer. Section~\ref{sec:conclusion} concludes and discusses directions for future research. Appendices support some details of this paper.

\section{Problem Formulation}\label{sec:problem}
Let $\Omega$ denote the set of future states of the world, and let $\mathcal{F}$ be a $\sigma$-algebra on $\Omega$. Let $P$ be the objective probability measure on $(\Omega,\mathcal{F})$. These together form the probability space $(\Omega,\mathcal{F},P)$, representing uncertainty over a time horizon $[0,T]$, where $0$ denotes the present and $T>0$ is a future time. Throughout this paper, all random variables are defined on $(\Omega,\mathcal{F},P)$.

\begin{definition}[Centralized Risk Sharing]\label{defsetup}
Define a centralized risk sharing problem in $[0,T]$ with $n+1$ agents, denote by $\mathcal{N}_n=\{1,2,\dots, n, n+1\}$ the labels of the agents, where the $(n+1)$-th agent is the central agent, for any $n\in\mathbb{N}=\{1,2,\dots\}$.

At time $t=0$ and prior to the risk sharing arrangement, the first $n$ agents hold their respective risks $X_1,X_2,\dots,X_n$, while the central agent does not hold any risk. With the risk sharing arrangement, the first $n$ agents make an ex-ante decision to transfer part of their respective risk, denoted by $I_i\left(X_i\right)$, for $i=1,2,\dots, n$, to the central agent, and each decision is such that $0\leq I_i(X_i)\leq X_i$. In return, the central agent ex-ante charges each $i$-th agent a price $\pi_i=\Pi_i\left(I_i\left(X_i\right)\right)$, for $i=1,2,\dots,n$. The $i$-th agent, for $i=1,2,\dots,n$, holds the random risk  
\begin{equation}\label{Yi}
Y_i=Y_i\left(I_i\right)=X_i-I_i\left(X_i\right)+\Pi_i\left(I_i\left(X_i\right)\right)\, ,
\end{equation}
while the single $\left(n+1\right)$-th central agent holds a risk 
depending on all decision functions $\mathbf{I}=(I_1,I_2,\dots, I_n)$, given by
\begin{equation}\label{Yn+1}
Y_{n+1}= Y_{n+1} (\mathbf{I})=\sum_{i=1}^{n}I_i\left(X_i\right)-\sum_{i=1}^{n}\Pi_i\left(I_i\left(X_i\right)\right)\,.
\end{equation}

At time $t=T$, if the first $n$ agents ex-post experience their risks $\left(X_1,X_2,\dots,X_n\right)$ as $\left(x_1,x_2,\dots,x_n\right)$, the $i$-th agent will hold an ex-post loss position as $y_i=x_i-I_i\left(x_i\right)+\Pi_i\left(I_i\left(X_i\right)\right)$, for $i=1,2,\dots,n$, while the $\left(n+1\right)$-th agent will hold an ex-post loss position as $y_{n+1}=\sum_{i=1}^{n}I_i\left(x_i\right)-\sum_{i=1}^{n}\Pi_i\left(I_i\left(X_i\right)\right)$.
\qed
\end{definition}

Let $\mathcal{X}$ be the set which contains all possible random variables involved in this risk sharing arrangement. The Centralized Risk Sharing problem requires specifying the set of admissible decision functions $\mathcal{I}_0\ni I_i$, and the pricing rules $\Pi_i: \mathcal{X} \rightarrow \mathbb{R}$, for $i=1,2,\dots,n$. We aim to find an optimal decision, where each agent is endowed with an objective function, denoted by $\rho_i: \mathcal{X}\rightarrow \mathbb{R}$, for $i\in\mathcal{N}_n$.

Classical insurance literature studies a particular case of this setting when $n=1$, with one policyholder and one (re)insurer.
In this context, the policyholder is the buyer of protection and the (re)insurer the seller. More generally, insurance can be tailored for multiple policyholders and one insurer, fitting in the centralized risk sharing. 

In insurance, all of the policyholders' risks are assumed to be non-negative; that is, $X_i\geq 0 $, for $i=1,2,\dots, n$. The decisions $I_i$, for $i=1,2,\dots, n$, are called ceding functions, and their common admissible set, which is convex, is 
\begin{equation*}
\mathcal{I}_0=\left\{I:\mathbb{R}_+= [0, \infty)\to\mathbb{R}_+:I\left(0\right)=0, \ 0\leq I(x)\leq x \text{ for all } x\geq 0\right\}.\footnote{For practical insurance applications, these ceding functions are usually required to be $1$-Lipschitz as well. That is, for any $0\leq x\leq y $, $0\leq I(y)-I(x)\leq y-x$, which implies that, when an ex-post loss increases, both the policyholder and the insurer bear a larger loss, thereby ruling out ex-post moral hazard. Since the results and the example in this paper do not rely on this condition, it is omitted.}
\end{equation*}
Thus, 
at most an ex-post loss $x\geq 0$ is transferred, while $I(0)=0$ means that if the ex-post loss is zero, there will be no transfer.
In the problem, decisions are represented as vectors $\I=(I_1,I_2,\dots,I_n)\in \mathcal{I}_0^n=\mathcal{I}_0\times \overset{(n)}{\cdots} \times \mathcal{I}_0 $, the $n$-fold Cartesian product of $\mathcal{I}_0$; note that the set $\mathcal{I}_0^n$ is convex. In general, consider the Centralized Risk Sharing problem for a subset of decisions $\I\in \mathcal{I}_n \subseteq\mathcal{I}_0^n $, where $\mathcal{I}_n$ is called from here on the feasible set, and the subscript $n$ refers to the setting with $n$ policyholders plus one insurer. An optimal decision vector is denoted by $\I^*=(I_1^*, I_2^*, \dots , I_n^*)\in\In$, and the individual losses for each agent using the optimal decision vector $\I^*$ are  given by $Y^*_i=Y_i\left(I^*_i\right)=X_i-I^*_i\left(X_i\right)+\Pi_i\left(I^*_i\left(X_i\right)\right)$, for   $i=1,2,\dots, n$, and $Y^*_{n+1}=Y_{n+1}(\mathbf{I}^*) = \sum_{i=1}^{n}I^*_i \left(X_i\right)-\sum_{i=1}^{n}\Pi_i \left(I^*_i\left(X_i\right)\right)$, following the same notations as in \eqref{Yi} and \eqref{Yn+1}, respectively.

The pricing rules $\Pi_i$, for $i=1,2,\dots,n$, in insurance are premium principles, which are defined for non-negative random variables, and the objectives $\rho_i$, for $i\in\mathbb{N}$, are chosen to be risk measures. With the arrangement, the time-0 risk objective of each policyholder is given by
\begin{equation}\label{rhoi}
\rho_i\left(Y_i\right)=\rho_i\left(Y_i\left(I_i\right)\right)=\rho_i\left(X_i-I_i\left(X_i\right)+\Pi_i\left(I_i\left(X_i\right)\right)\right)\,,
\end{equation}
for $i=1,2,\dots, n$, while for the insurer is given by
\begin{equation}\label{rhon+1}
\rho_{n+1}\left(Y_{n+1}\right)=\rho_{n+1}(Y_{n+1}(\mathbf{I})) = 
\rho_{n+1}\left(\sum_{i=1}^{n}I_i\left(X_i\right)-\sum_{i=1}^{n}\Pi_i\left(I_i\left(X_i\right)\right)\right)\,,
\end{equation}
for any decision $\I\in \In$. When the risk measures are applied to the risks calculated with an optimal decision vector $\I^*\in\In$, they are denoted as $\rho_i(Y_i^*) = \rho_i(Y_i(I_i^*))$, for $i=1,2,\dots, n$, and $\rho_{n+1}(Y_{n+1}^*) = \rho_{n+1}(Y_{n+1}(\mathbf{I}^*))$, following the same notations as in \eqref{rhoi} and \eqref{rhon+1}.

For the Centralized Risk Sharing problem in insurance, assume that $\mathcal{X}$ is a linear space,
so that the risks in \eqref{rhoi} and \eqref{rhon+1} are well-defined. Note that once the premium principles and risk measures are specified, they may impose other conditions on $\mathcal{X}$, usually in terms of the moments of the random variables, and thus we may choose $\mathcal{X}$ to be a suitable $L^p$ space\footnote{Recall that $L^p = \{Z\in L^0 : \, \int_\Omega |Z|^p \, dP < \infty \}$, for $p\geq 0$, where $L^0$ is the set of all finite real-valued random variables on $(\Omega, \mathcal{F}, P)$.}, for some $p\geq 0$; since, for any $x\geq 0$, $I(x)$ and $x-I(x)$, for $I\in\mathcal{I}_0$, are always bounded above by $x$, the random variables involved in the Centralized Risk Sharing problem would have as many finite moments as the non-negative risks prior to being shared.

For the remainder of this paper, we use the insurance terminology. Nevertheless, our problem setting, results, and discussions can be adapted beyond insurance.
All the elements of the Centralized Risk Sharing problem can be written in the tuple:
\[ (X_1,X_2, \dots, X_n, \rho_1, \rho_2, \dots, \rho_{n+1}, \Pi_1, \Pi_2,\dots,  \Pi_n; \mathcal{I}_n)\, ,\]
that we shorten as $(\Nn;\In)$, with a slight abuse of notation; the labels $\Nn$ of the $n+1$ agents identify the risks prior to sharing, the risk measures for all agents, and the premium principles for the $n$ policyholders, while $\mathcal{I}_n$ remains as the feasible set. The full setting being described by $(\Nn; \In)$, for $n\in\mathbb{N}$, means the Centralized Risk Sharing problem in Definition~\ref{defsetup}. In particular, when $n=1$, the full setting is a bilateral case, since it involves two agents: one policyholder and the insurer.

A standard equilibrium to determine the best decision $\mathbf{I}^* = \left(I^*_1,I^*_2,\dots,I^*_n\right)\in\mathcal{I}_n$, in a setting involving a multiobjective optimization problem, is Pareto optimality.
\begin{definition}\label{def:PO}
Given the full setting $(\Nn; \In)$, for $n\in\mathbb{N}$, a feasible set $\mathcal{C}\subseteq\mathcal{I}_n$, and a group of agents $\mathcal{G}\subseteq\mathcal{N}_n$, we say $\mathbf{I}^* = \left(I^*_1,I^*_2,\dots,I^*_n\right)\in\mathcal{C}$ is Pareto optimal among 
$\mathcal{G}$ within 
$\mathcal{C}$, if there does not exist another $\mathbf{I} = \left(I_1,I_2,\dots,I_n\right)\in\mathcal{C}$ such that, $\rho_i\left(Y_i\right)\leq\rho_i\left(Y^*_i\right)$, for any $i\in\mathcal{G}$, with at least one inequality being strict. In particular, when $\mathcal{G}=\mathcal{N}_n$ and $\mathcal{C}=\mathcal{I}_n$, we say $\mathbf{I}^*\in\mathcal{I}_n$ is Pareto optimal under the full setting.
\qed
\end{definition}

In some results of this paper, we shall require the premium principles, the risk measures, and the feasible set to satisfy certain conditions. We summarize them in the assumption below. When this assumption is required for a result, we specify it with a $(\dagger)$ for simplicity.
\begin{assumption}[$\dagger$]\label{assumption} 
The pricing rules $\Pi_i$ are semi-linear on $\mathcal{X}$; the risk measures $\rho_i$ are convex on $\mathcal{X}$; and the feasible set $\mathcal{I}_n$ is convex, for $i\in \Nn$, for $n\in \mathbb{N}$.\footnote{Let $f:\mathcal{X}\rightarrow \real$. The functional $f$ is semi-linear if $f(wZ_1+Z_2) = wf(Z_1) +  f(Z_2)$, for all $w\geq 0$ and all $Z_1,Z_2\in \mathcal{X}$; $f$ is convex if $f(w Z_1 + (1-w)Z_2) \leq w  f(Z_1) + (1- w)f(Z_2)$, for all $w\in [0,1]$ and all $Z_1, Z_2 \in \mathcal{X}$. Note that the risk measures herein are only required to satisfy the convexity, and thus they might not be ``convex risk measures'' in the usual sense, which satisfy translation invariance, monotonicity, and convexity.}
\end{assumption}
\begin{example}\label{eg:rational}
We outline an example of a convex feasible set. Let
\begin{equation*}
\mathcal{I}_n=\left\{\I\in\mathcal{I}_0^n:\rho_i\left(Y_i\left(I_i\right)\right) \leq \rho_i(X_i), \text{ for } i=1,2,\dots,n, \text{ and }\rho_{n+1}(Y_{n+1}(\mathbf{I})) \leq \rho_{n+1}(0)\right\}.
\end{equation*}
In the insurance context, the conditions in this set are known as individual rationality. If the premium principles are semi-linear, and the risk measures are convex, this feasible set $\mathcal{I}_n$ is indeed convex.
Recall that $\mathcal{I}_0^n$ is convex. Take $\I=(I_1, I_2, \dots, I_n),\;\mathbf{J}= (J_1, J_2, \dots , J_n) \in \In\subseteq\mathcal{I}_0^n$, and thus $\theta_1\I + \theta_2 \mathbf{J}\in\mathcal{I}_0^n$, where $\theta_1\in [0,1]$ and $\theta_2=1-\theta_1\in [0,1]$. We can show that the objectives are convex on the decision set; see Lemma~\ref{rhoconvexI} below. Thus, we have 
\begin{align*}
&\rho_i\left(Y_i(\theta_1 I_i + \theta_2 J_i)\right) \leq \theta_1 \rho_i \left( Y_i\left(I_{i}\right) \right) + \theta_2 \rho_i \left( Y_i\left(J_{i}\right) \right) \,, \text{ for } i=1,2,\dots, n\,,\\
&\rho_{n+1} \left( Y_{n+1}(\theta_1 \mathbf{I} +\theta_2 \mathbf{J})\right) \leq \theta_1\rho_{n+1}\left( Y_{n+1}\left( \mathbf{I} \right) \right) + \theta_2\rho_{n+1}\left( Y_{n+1}\left(\mathbf{J}\right) \right)\;.
\end{align*}

Since $\I,\mathbf{J}\in\mathcal{I}_n$, we have $\rho_i\left(Y_i\left(I_i\right)\right) \leq \rho_i(X_i), \text{ for } i=1,2,\dots,n,\text{ and }\rho_{n+1}(Y_{n+1}(\mathbf{I})) \leq \rho_{n+1}(0)$, as well as $\rho_i\left(Y_i\left(J_i\right)\right) \leq \rho_i(X_i), \text{ for } i=1,2,\dots,n,\text{ and }\rho_{n+1}(Y_{n+1}(\mathbf{J})) \leq \rho_{n+1}(0)$. Therefore, for $i=1,2,\dots,n$, $\rho_i\left(Y_i(\theta_1 I_i + \theta_2 J_i)\right)\leq \rho_i(X_i)$, while $\rho_{n+1}\left( Y_{n+1}(\theta_1 \mathbf{I} +\theta_2 \mathbf{J})\right)\leq \rho_{n+1}(0)$, implying that $\theta_1\I + \theta_2 \mathbf{J}\in\mathcal{I}_n$
\qed
\end{example}

The following lemma is used for the 
Example \ref{eg:rational} above 
and will also be useful when proving later results in this paper. The proof of the following lemma is in Appendix~\ref{ap:rhoconvexI}.
\begin{lemma}\label{rhoconvexI}
Let the full setting $(\Nn;\In)$, for $n\in\mathbb{N}$, satisfy $(\dagger)$. If $\mathcal{C}\subseteq \mathcal{I}_n$ is a convex feasible subset, then, each agent's objective is convex in the feasible decisions $\mathbf{I}\in\mathcal{C}$. That is, for any $m\in\mathbb{N}$, $\mathbf{I}^{(j)}=(I^{(j)}_1,I^{(j)}_2,\dots, I^{(j)}_n)\in\mathcal{C}$, $\theta_j \geq 0$, for $j=1,2,\dots,m$, with $\sum_{j=1}^m\theta_j = 1$, it holds for any $i=1,2,\dots, n$,
\[\begin{aligned}
 &\rho_i\left(Y_i\left(\sum_{j=1}^m\theta_jI^{(j)}_i\right)\right)  \leq \sum_{j=1}^m\theta_j\rho_i\left(Y_i\left(I^{(j)}_i\right)\right)\,,\\
 &\rho_{n+1}\left(Y_{n+1}\left(\sum_{j=1}^m\theta_j\mathbf{I}^{(j)}\right)\right)  \leq \sum_{j=1}^m\theta_j\rho_{n+1}\left(Y_{n+1}\left(\mathbf{I}^{(j)}\right)\right)\,.
\end{aligned}\]
\end{lemma}

\section{Pareto Optimality and Sequential Optimization}\label{sec:PO_sequential}
This section is devoted to studying Pareto optimal decision vectors $\mathbf{I}^* \in\mathcal{I}_n$ for the Centralized Risk Sharing problem in Definition \ref{defsetup}. 
We first review the bilateral case ($n=1$), previously studied in \cite{CAI2017},
as a basis for the discussion in a
multilateral case.

\subsection{Bilateral Risk Sharing Revisited}\label{sec:bilateral}
The bilateral setting ($n=1$) is described by the tuple $\left(X_1, \rho_1, \rho_2, \Pi_1; \mathcal{I}_1\right)$, that is shortened by $(\mathcal{N}_1; \mathcal{I}_1)$, where $\mathcal{I}_1\subseteq \mathcal{I}_0$; each decision is a ceding function $I_1\in \mathcal{I}_1$, where the subscript of $I_1$ shall be dropped in this section for simplicity. Under the bilateral setting, centralization is actually irrelevant, as there is only one policyholder.

Define the set-valued function\footnote{For a set $E$, let $2^E$ denote its power set.} $K:[0,1] \rightarrow 2^{\mathcal I_1}$ by, for $\lambda\in\left[0,1\right]$,
\begin{equation}
K\left(\lambda\right)=\argmin_{I\in\mathcal{I}_1}\left\lbrace\lambda\rho_1\left(Y_1\right)+\left(1-\lambda\right)\rho_2\left(Y_2\right)\right\rbrace\, .
\label{eq:function_K}
\end{equation}
The following proposition is a well-known result; we present it to prepare the study of the multilateral case in the next section.
\begin{proposition}\label{2_min_PO}
Consider the bilateral setting $(\mathcal{N}_1;\mathcal{I}_1)$. If $I^*\in K\left(\lambda\right)$, for some $\lambda\in\left(0,1\right)$, then $I^*\in\mathcal{I}_1$ is Pareto optimal under the bilateral setting.
\end{proposition}
\noindent
The other direction does not hold in general. Theorem 2.1 in \cite{CAI2017} proposed a characterization of Pareto optimal decisions in the bilateral case, which we shall recall it as Corollary \ref{2_result} below.
It essentially modifies the sufficient condition in Proposition \ref{2_min_PO}, more precisely on the function $K$, for the characterization under the assumptions $(\dagger)$.
The key idea involves three results which are recalled below.

\begin{theorem}\label{2_PO_min}
Let the bilateral setting $(\mathcal{N}_1; \mathcal{I}_1)$
satisfy $(\dagger)$. If $I^*\in\mathcal{I}_1$ is Pareto optimal under the bilateral setting, then  $I^*\in K\left(\lambda\right)$, for some $\lambda\in\left[0,1\right]$.
\end{theorem}
Denote $\mathrm{PO}$ as the set of Pareto optimal decisions. Proposition~\ref{2_min_PO} states that $\cup_{\lambda\in (0,1)} K(\lambda)\subseteq \mathrm{PO}$, while Theorem~\ref{2_PO_min} states that $\mathrm{PO}\subseteq \cup_{\lambda\in [0,1]} K(\lambda)$. These\footnote{Under the assumptions $(\dagger)$, the smaller and the larger sets of decisions than the PO set are respectively known as Geoffrion-proper Pareto optimality and weak Pareto optimality.} are illustrated in Figure~\ref{fig:2_results12}.
First, note that the larger set $\cup_{\lambda\in[0,1]} K(\lambda)$ consists of three parts: $\cup_{\lambda\in(0,1)} K(\lambda)$, $K(0)$, and $K(1)$, as shown in Figure~\ref{fig:2_partition}. Secondly, by collating Proposition~\ref{2_min_PO} and Theorem~\ref{2_PO_min} via this composition, observe that, under $(\dagger)$, Pareto optimal decisions are either in $\cup_{\lambda\in (0,1)} K(\lambda)$, in $K(0)$, or in $K(1)$, as shown in Figure~\ref{fig:2_results12partition}. More precisely, since all decisions in $\cup_{\lambda\in (0,1)} K(\lambda)$ are Pareto optimal, what remains to be characterized are the decisions in the left-hand-side circular band and the right-hand-side circular band in Figure~\ref{fig:2_results12partition}, that would fill up the $\mathrm{PO}$ set. Note that the left band is a subset of $K(0)$, while the right band is a subset of $K(1)$. These observations explain the proposal by \cite{CAI2017} to modify the function $K$ in order to characterize the decisions in these bands, and consequently Pareto optimality under $(\dagger)$.

\begin{figure}[h!]
\centering
\begin{subfigure}[C]{0.45\linewidth}
\centering
\begin{tikzpicture}[thick,fill,scale=0.9]
\draw (0,0) circle (1.3cm);
\node at (0,-0.2) {$ \bigcup\limits_{\lambda\in (0,1)} K(\lambda)$};
\draw (0,0) ellipse (2.5 and 2.2);
\node at (1.4,1.1) {$\mathrm{PO}$};

\begin{scope} \clip (0,0) circle (1.3cm); 
\foreach \x in {-0.8,-0.4,...,1.2} { \draw[gray, dotted] (\x,-1.3) -- (\x,1.4); } 
\end{scope}

\end{tikzpicture}
\caption{\scriptsize{Proposition~\ref{2_min_PO}. }}\label{fig:2_min_PO}
\end{subfigure}%
~
\begin{subfigure}[c]{0.45\textwidth}
\centering
\begin{tikzpicture}[thick,fill,scale=0.9]
\draw[rounded corners=15pt] (0,0) rectangle (6,4.4); 
\draw (3,2.2) ellipse (1.34 and 1.16);
\node at (3,2.2) {$\mathrm{PO}$};
\node at (4.8,3.6) {$\bigcup\limits_{\lambda\in [0,1]} K(\lambda)$};
\end{tikzpicture}
\caption{\scriptsize{Theorem~\ref{2_PO_min} under $(\dagger)$.}}\label{fig:2_PO_min}
\end{subfigure}
\caption{\footnotesize{Illustration of Proposition~\ref{2_min_PO} and Theorem~\ref{2_PO_min}.}}\label{fig:2_results12}
\end{figure}
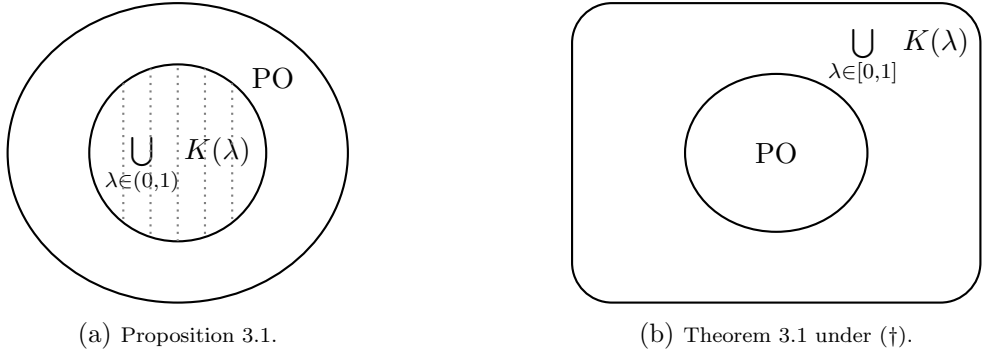
\begin{figure}[h!]
\centering
\begin{subfigure}[c]{0.45\textwidth}
\centering
\begin{tikzpicture}[thick,fill,scale=0.9]
\draw[rounded corners=15pt] (0,0) rectangle (6,4.4);
\draw (3,0) -- (3,1);
\draw (3,3.4) -- (3,4.4);
\begin{scope}
\clip (0,0) rectangle (3,4.4);
\clip (0,0) rectangle (6,4.4) (3,2.2) circle (1.2cm);
\foreach \x in {0.5,1,...,4} {
    \draw[gray, thin, dashed] (\x,0) -- (\x,4.4);
}
\end{scope}
\begin{scope}
\clip (3,0) rectangle (6,4.4);
\clip (0,0) rectangle (6,4.4) (3,2.2) circle (1.2cm);
\foreach \y in {0.5,1,...,4} {
    \draw[gray, thin, dashed] (3,\y) -- (6,\y);
}
\end{scope}
\draw (3,2.2) circle (1.2cm);
\node at (3,2) {$\bigcup\limits_{\lambda\in (0,1)}K(\lambda)$};
\node at (1,3.7) {$K(0)$};
\node at (5.2,3.7) {$K(1)$};
\begin{scope} \clip (3,2.2) circle (1.2cm); 
\foreach \x in {2.2,2.6,3.0,3.4,3.8} { 
\draw[gray, dotted] (\x,0.8) -- (\x,3.6); 
} 
\end{scope}

\end{tikzpicture}
\caption{ }\label{fig:2_partition}
\end{subfigure}%
~
\begin{subfigure}[c]{0.45\textwidth}
\centering
\begin{tikzpicture}[thick,fill,scale=0.9]
\draw[rounded corners=15pt] (0,0) rectangle (6,4.4);
\draw (3,0) -- (3,1);
\draw (3,3.4) -- (3,4.4);
\draw (3,2.2) circle (1.2cm);
\begin{scope}
\clip (0,0) rectangle (3,4.4);
\clip (0,0) rectangle (6,4.4) (3,2.2) circle (1.2cm);
\foreach \x in {0.5,1,...,4} {
    \draw[gray, thin, dashed] (\x,0) -- (\x,4.4);
}
\end{scope}
\begin{scope}
\clip (3,0) rectangle (6,4.4);
\clip (0,0) rectangle (6,4.4) (3,2.2) circle (1.2cm);
\draw[gray, thin, dashed] (3,0.1) -- (5.8,0.1);
\foreach \y in {0.6,1.1,...,4.4} {
    \draw[gray, thin, dashed] (3,\y) -- (6,\y);
}
\end{scope}
\draw[fill=Mulberry!40, fill opacity=0.3] (3,2.2) ellipse (2.3 and 2);
\node at (3,2) {$\bigcup\limits_{\lambda\in (0,1)}K(\lambda)$};
\node at (1,3.8) {$K(0)$};
\node at (5.2,3.8) {$K(1)$};
\node at (4.3,3.35) {$\mathrm{PO}$};
\begin{scope} \clip (3,2.2) circle (1.2cm); 
\foreach \x in {2.2,2.6,3.0,3.4,3.8} { 
\draw[gray, dotted] (\x,0.8) -- (\x,3.6); 
} 
\end{scope}
\end{tikzpicture}
\caption{ }\label{fig:2_results12partition}
\end{subfigure}
\caption{\scriptsize{(a) Illustrates the composition of $\cup_{\lambda\in [0,1]}K(\lambda)$ in three parts: $\cup_{\lambda\in (0,1)} K(\lambda)$ (the circle with vertical dotted lines),  $K(0)$ (the left-hand portion with vertical dashed lines), and $K(1)$ (the right-hand portion with horizontal dashed lines). (b) Illustration of Proposition~\ref{2_min_PO} and Theorem~\ref{2_PO_min}, under $(\dagger)$, via the composition in (a), where the purple shaded circle is the $\mathrm{PO}$ set.}}\label{fig:2_together12}
\end{figure}
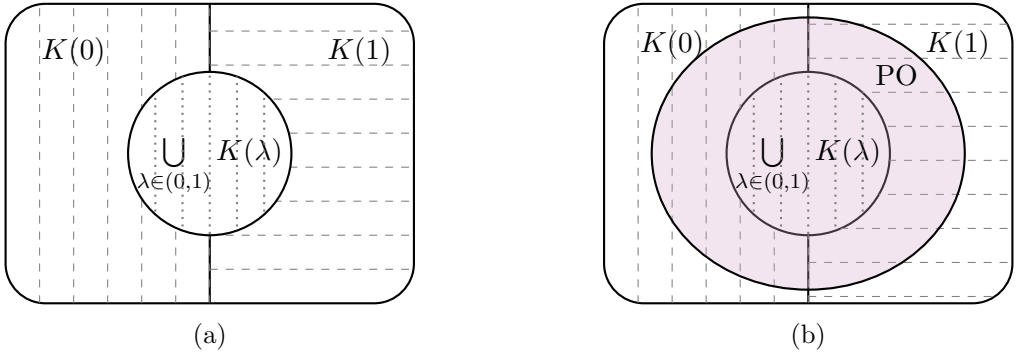

Define the set-valued function $K^*:[0,1] \rightarrow 2^{\mathcal I_1}$ by, for $\lambda\in[0,1]$,
\begin{equation}
\begin{aligned}
K^*\left(\lambda\right)=&\;K\left(\lambda\right),\text{ if }\lambda\in\left(0,1\right),\\
K^*\left(0\right)=\argmin_{I\in K\left(0\right)}\rho_1(Y_1)&,\text{ and }K^*\left(1\right)=\argmin_{I\in K\left(1\right)}\rho_2(Y_2).
\end{aligned}
\label{eq:function_K_star}
\end{equation}

\noindent
While $K^*(\lambda)$ coincides with $K(\lambda)$, for $\lambda\in(0,1)$, at the boundary cases when $\lambda=0$ and $\lambda=1$, $K^*$ is defined through sequential optimization. That is, for $\lambda=0$, the objective of the insurer is first optimized, and then, among the decisions that are optimal for the insurer, the objective of the policyholder is further optimized. For $\lambda=1$, the sequence is opposite, the objective of the policyholder is first optimized, and then, among the decisions that are optimal for the policyholder, the objective of the insurer is further optimized.

\begin{theorem}\label{2_K^*_PO}
Consider the bilateral setting $(\mathcal{N}_1; \mathcal{I}_1)$. If $I^*\in K^*\left(\lambda\right)$, for some $\lambda\in\left[0,1\right]$, then $I^*\in\mathcal{I}_1$ is Pareto optimal under the bilateral setting.
\end{theorem}
\noindent
Theorem \ref{2_K^*_PO} generalizes Proposition~\ref{2_min_PO}, by including the boundary cases when $\lambda=0$ and $\lambda=1$. That is, not only are the decisions that minimize a strict convex combination of the policyholder's and insurer's objectives in one go Pareto optimal, but those obtained through sequential optimization are also Pareto optimal.

The following theorem states the other direction of Theorem \ref{2_K^*_PO}, under the assumptions $(\dagger)$.
\begin{theorem}\label{2_PO_K^*}
Let the bilateral setting $(\mathcal{N}_1; \mathcal{I}_1)$ satisfy $(\dagger)$. If $I^*\in\mathcal{I}_1$ is Pareto optimal under the bilateral setting, then $I^*\in K^*\left(\lambda\right)$, for some $\lambda\in\left[0,1\right]$. 
\end{theorem}

Combining Theorems~\ref{2_K^*_PO} and \ref{2_PO_K^*}, \cite{CAI2017} showed the characterization of Pareto optimal decisions through the function $K^*$ under the bilateral setting.
\begin{corollary}[Theorem 2.1 of \cite{CAI2017}]\label{2_result}
Let the bilateral setting $(\mathcal{N}_1; \mathcal{I}_1)$ satisfy $(\dagger)$. Then, $I^*\in\mathcal{I}_1$ is Pareto optimal under the bilateral setting, if and only if $I^*\in K^*\left(\lambda\right)$, for some $\lambda\in\left[0,1\right]$.
\end{corollary}
\begin{figure}[h!]
\centering
\begin{tikzpicture}[thick,fill,scale=0.9]
\draw (0,0) ellipse (2.3 and 2);
\node at (0,-0.25) {$\mathrm{PO} = \bigcup\limits_{\lambda\in [0,1]} K^*(\lambda)  $};
\end{tikzpicture}
\caption{\footnotesize{Characterization of Pareto optimality, when $n=1$, under $(\dagger)$, as per Corollary~\ref{2_result}.}}
\label{fig:2_PO_equal_K^*}
\end{figure}
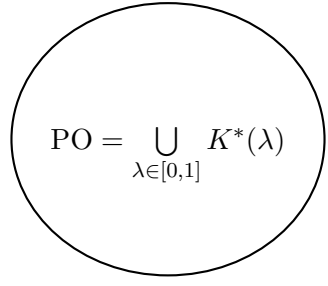
\begin{figure}[h!]
\centering
\begin{subfigure}[c]{0.45\textwidth}
\centering
\begin{tikzpicture}[thick,fill,scale=0.9]
\draw (3,0.2) -- (3,1);
\draw (3,3.4) -- (3,4.2);
\draw (3,2.2) circle (1.2cm);
\draw (3,2.2) ellipse (2.3 and 2);
\begin{scope}
\clip (3,2.2) ellipse (2.3cm and 2cm); 
\clip (-10,-10) rectangle (10,10) (3,2.2) circle (1.2cm); 
\clip (0,0) rectangle (3,4.4); 
\foreach \y in {-4,-3.5,...,8} {
    \draw[gray, dotted] (-1,\y) -- (4,\y+3);
}
\end{scope}
\begin{scope}
\clip (3,2.2) ellipse (2.3cm and 2cm); 
\clip (-10,-10) rectangle (10,10) (3,2.2) circle (1.2cm); 
\clip (3,0) rectangle (6,4.4); 
\foreach \y in {-4,-3.5,...,8} {
    \draw[gray, dotted] (7,\y) -- (2,\y+3);
}
\end{scope}

\begin{scope} \clip (3,2.2) circle (1.2cm); 
\foreach \x in {2.2,2.6,3.0,3.4,3.8} { 
\draw[gray, dotted] (\x,0.8) -- (\x,3.6); 
} 
\end{scope}

 \node at (3,2) {$\bigcup\limits_{\lambda\in (0,1)}K(\lambda)$};
\node at (1.25,2.2) {\footnotesize{$K^*(0)$}};
\node at (4.75,2.2) {\footnotesize{$K^*(1)$}};
\end{tikzpicture}
\caption{ }\label{fig:2_partition*}
\end{subfigure}%
~
\begin{subfigure}[c]{0.45\textwidth}
\centering
\begin{tikzpicture}[thick,fill,scale=0.9]
\draw[rounded corners=15pt] (0,0) rectangle (6,4.4);
\draw (3,0) -- (3,1);
\draw (3,3.4) -- (3,4.4);
\draw (3,2.2) circle (1.2cm);

\begin{scope}
\clip (3,2.2) ellipse (2.3cm and 2cm); 
\clip (-10,-10) rectangle (10,10) (3,2.2) circle (1.2cm); 
\clip (0,0) rectangle (3,4.4); 
\foreach \y in {-4,-3.5,...,8} {
    \draw[gray, dotted] (-1,\y) -- (4,\y+3);
}
\end{scope}
\begin{scope}
\clip (3,2.2) ellipse (2.3cm and 2cm); 
\clip (-10,-10) rectangle (10,10) (3,2.2) circle (1.2cm); 
\clip (3,0) rectangle (6,4.4); 
\foreach \y in {-4,-3.5,...,8} {
    \draw[gray, dotted] (7,\y) -- (2,\y+3);
}
\end{scope}
\begin{scope}
\clip (0,0) rectangle (3,4.4);
\clip (0,0) rectangle (6,4.4) (3,2.2) circle (1.2cm);
\foreach \x in {0.5,1,...,4} {
    \draw[gray, thin, dashed] (\x,0) -- (\x,4.4);
}
\end{scope}
\begin{scope}
\clip (3,0) rectangle (6,4.4);
\clip (0,0) rectangle (6,4.4) (3,2.2) circle (1.2cm);
\draw[gray, thin, dashed] (3,0.1) -- (5.8,0.1);
\foreach \y in {0.6,1.1,...,4.4} {
    \draw[gray, thin, dashed] (3,\y) -- (6,\y);
}
\end{scope}

\begin{scope} \clip (3,2.2) circle (1.2cm); 
\foreach \x in {2.2,2.6,3.0,3.4,3.8} { 
\draw[gray, dotted] (\x,0.8) -- (\x,3.6); 
} 
\end{scope}
\draw[fill=Mulberry!40, fill opacity=0.3] (3,2.2) ellipse (2.3 and 2);
\node at (3,2) {$\bigcup\limits_{\lambda\in (0,1)}K(\lambda)$};
\node at (0.8,3.8) {$K(0)$};
\node at (5.2,3.8) {$K(1)$};
\node at (1.25,2.2) {\footnotesize{$K^*(0)$}};
\node at (4.75,2.2) {\footnotesize{$K^*(1)$}};
\node at (4.3,3.2) {$\mathrm{PO}$};
\end{tikzpicture}
\caption{ }\label{fig:2_results12*}
\end{subfigure}
\caption{\scriptsize{(a) Shows the composition of $\cup_{\lambda\in [0,1]}K^*(\lambda)$ in three parts: $\cup_{\lambda\in (0,1)} K(\lambda)$ (the circle with vertical dotted lines),  $K^*(0)$ (the left-hand-side circular band with upward
dotted lines), and $K^*(1)$ (the right-hand-side circular band with downward diagonal dotted lines). (b) Illustration of Corollary~\ref{2_result}  within the composition of $\cup_{\lambda\in [0,1]}K(\lambda)$, under $(\dagger)$, where the purple shaded circle is the $\mathrm{PO}$ set.}}\label{fig:2_all}
\end{figure}
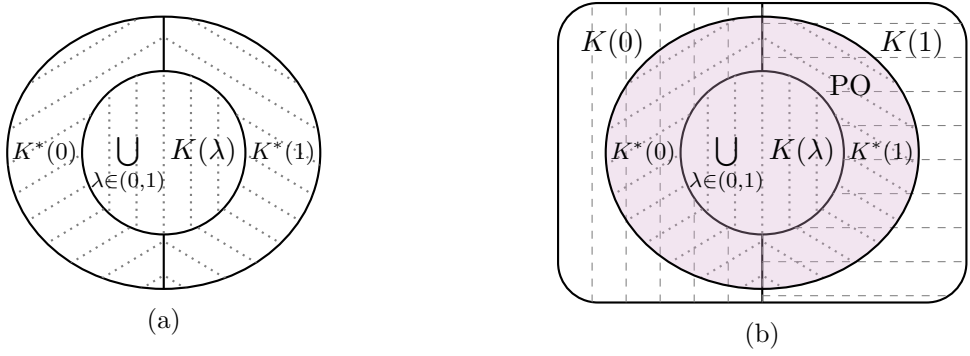

\noindent
This characterization for the bilateral case entails that, under $(\dagger)$, all Pareto optimal decision functions can be solved by either, (i) minimizing the time-$0$ risk positions of both agents in terms of their strict convex combination, or (ii) a bi-level optimization, first minimizing the time-$0$ risk position of only one of the agents, and then, among those decision functions which are optimal for that agent, minimizing the time-$0$ risk position of the other agent. That is, under ($\dagger$), $\mathrm{PO}=\cup_{\lambda\in [0,1]} K^*(\lambda)$; see Figure~\ref{fig:2_PO_equal_K^*} for illustration. Precisely, this characterization can be interpreted by viewing $\mathrm{PO}=\cup_{\lambda\in [0,1]} K^*(\lambda)$ in three parts: $\cup_{\lambda\in (0,1)} K^*(\lambda)=\cup_{\lambda\in (0,1)} K(\lambda)$, $K^*(0)$, and $K^*(1)$, as shown in Figure~\ref{fig:2_partition*}. We show each part separately to illustrate the three different paths to obtain Pareto optimality; however, note these sets may intersect. The first part corresponds to (i) minimizing the strict convex combination, while the second and third parts, $K^*(0)$ and $K^*(1)$, correspond to (ii) the bi-level optimizations. Note that $K^*(0) \subseteq K(0)$ and $K^*(1)\subseteq K(1)$, and thus this characterization identifies the Pareto optimal decisions in the circular bands within $K(0)$ and $K(1)$ of Figure~\ref{fig:2_results12partition} by $K^*(0)$ and $K^*(1)$, respectively. To summarize, we illustrate the characterization of $\mathrm{PO}$ as the composition of these three parts via the function $K^*$ in Figure~\ref{fig:2_results12*}. 

\subsubsection{Encoding Bilateral Notation via Sequential Matrices}
In the following, we propose a new  
notation and redefine the functions to reflect more clearly the role of sequential optimization in characterizing Pareto optimality in the bilateral case ($n=1$). These will, upon extension, be instrumental under the multilateral setting for the remainder of this paper. 
First, for any feasible subset $\mathcal{C}\subseteq \mathcal{I}_1$, define the set-valued function $K\left(\cdot;\mathcal{C}\right):[0,1]^2 \rightarrow 2^{\mathcal I_1}$ by, for $\bm{\lambda} = (\lambda_1, \lambda_2)^T \in [0,1]^2$ with $\lambda_1 + \lambda_2 = 1$,
\begin{equation}\label{eq:K_lambda_C_1}
K\left({\bm\lambda};\mathcal{C}\right)=\argmin_{I\in\mathcal{C}}\left\lbrace\lambda_1\rho_1\left(Y_1\right)+\lambda_2\rho_2\left(Y_2\right)\right\rbrace\, ;
\end{equation}
if $\mathcal{C} = \mathcal{I}_1$, we simply write $K\left({\bm\lambda};\mathcal{I}_1\right)$ as $K\left({\bm\lambda}\right)$, which coincides with the function $K$ in \eqref{eq:function_K}.

Next, define the following set of $2\times 2$ matrices:
\begin{equation}\label{eq:A1}
\mathcal{A}_1 =  \left\{ 
  \begin{pmatrix}
\lambda_1 & 0\\
\lambda_2 & 0 
\end{pmatrix}\text{ such that }\lambda_1, \lambda_2 \in (0,1)\text{ with } \lambda_1+\lambda_2=1 \,, \, 
 \begin{pmatrix}
1 & 0\\
0 & 1 
\end{pmatrix} ,  
  \begin{pmatrix}
0 & 1\\
1 & 0 
\end{pmatrix}     \right\} \, .
\end{equation}
Although this set contains infinitely many matrices, they fall into three distinct classes. More precisely, we write $\mathcal{A}_1 = \mathcal{A}_1^{(1)}\cup \mathcal{A}_1^{(2)}\cup\mathcal{A}_1^{(3)} $, where
\begin{equation}
\mathcal{A}_1^{(1)} = \left\{ \begin{pmatrix}
\lambda_1 & 0\\
\lambda_2 & 0 
\end{pmatrix}:\;\lambda_1, \lambda_2 \in (0,1),\; \lambda_1+\lambda_2=1\, \right\}\;,
\label{eq:A_1_1}
\end{equation}
\begin{equation}
\mathcal{A}_1^{(2)} = \left\{ \begin{pmatrix}
1 & 0\\
0 & 1 
\end{pmatrix}\right\}\,, \text{ and } \, \mathcal{A}_1^{(3)} =  \left\{\begin{pmatrix}
0 & 1\\
1 & 0 
\end{pmatrix}\right\}\, .
\label{eq:A_1}
\end{equation}
Any $\Lambda \in \mathcal{A}_1^{(1)}$ is a representative of this class. Since each of the sets $\mathcal{A}_1^{(2)}$ and $\mathcal{A}_1^{(3)}$ is a singleton, its unique element is the representative of the corresponding class. 
We introduce these matrices so that each class corresponds to a specific sequential optimization, when the matrices in that class are read 
by columns from left to right. For example, a matrix in $\mathcal{A}_1^{(1)}$ indicates that we optimize a strict convex combination of both agents' risks, with the positive weights given in the first column; since the second column is the zero column, there is no further optimization step. The matrix in $\mathcal{A}_1^{(2)}$ indicates that we first optimize the first agent's objective and then the second agent's objective, while the matrix in $\mathcal{A}_1^{(3)}$ reverses this order, so the second agent's objective is optimized first, followed by the first agent's objective.

We now associate these matrices with their corresponding sequential optimizations as follows. 
For each matrix $\Lambda\in \mathcal{A}_1$, define recursively:
\begin{equation}\label{eq:function_K_star_recursion_n_2}
K^{*\left(1\right)}\left(\Lambda[\cdot,1]\right)=K\left(\Lambda[\cdot,1];\mathcal{I}_1\right)=K\left(\Lambda[\cdot,1]\right),
\end{equation}
\begin{equation}\label{eq:function_K_star_recursion_n_1}
K^{*\left(2\right)}\left(\Lambda[\cdot,1:2]\right)=K\left(\Lambda[\cdot, 2];K^{*\left(1\right)}\left(\Lambda[\cdot,1:1]\right)\right),\text{ if $n(\Lambda)=2$},
\end{equation}
where $K$ is defined in \eqref{eq:K_lambda_C_1}, $n(\Lambda)\in\{1,2\}$ denotes the number of non-zero columns of $\Lambda$, $\Lambda[\cdot,j]$ denotes the $j$-th column of $\Lambda $, and $\Lambda[\cdot,1:j]$ denotes the sub-matrix consisting of the first $j$ columns of $\Lambda$, for $j=1,2$; note that, $\Lambda[\cdot, 1:2]=\Lambda$ and $\Lambda[\cdot, 1:1]=\Lambda[\cdot, 1]$. By the definition of $K^{*(1)}$, the objective(s) of the agent(s) are optimized over all feasible decisions according to the weights in the first column of a matrix $\Lambda\in\mathcal{A}_1$, as indicated by the superscript $(1)$. By the definition of $K^{*(2)}$, if the second column of $\Lambda$ is non-zero, then over the decisions that are optimal for an agent from the first optimization step, the objective of another agent is further optimized, 
as indicated by the superscript $(2)$.

Finally, define the set-valued function $K^*:\mathcal{A}_1 \rightarrow 2^{\mathcal I_1}$ by, for $\Lambda\in \mathcal{A}_1$,
\begin{equation}\label{eq:function_K_star_1}
K^*\left(\Lambda\right)=K^{*\left(n\left(\Lambda\right)\right)}\left(\Lambda[\cdot,1:n\left(\Lambda\right)]\right)\, ,
\end{equation}
where $K^{*\left(j\right)}$, for $j=1,2$, is defined recursively by \eqref{eq:function_K_star_recursion_n_2} and \eqref{eq:function_K_star_recursion_n_1}.

The following shows that the set-valued functions $K^*$ defined by \eqref{eq:function_K_star} and \eqref{eq:function_K_star_1} are equivalent by examining each of the three representative matrices in $\mathcal{A}_1$.
\begin{itemize}
\item
For $\Lambda \in \mathcal{A}_1^{(1)}$, $n(\Lambda) = 1$,  $\Lambda[\cdot ,1] = (\lambda_1,\lambda_2)^T$ for some $\lambda_1, \lambda_2\in (0,1)$ such that $\lambda_1+\lambda_2=1$, 
\[
K^{\ast} (\Lambda) = K^{\ast (1)} ((\lambda_1,\lambda_2)^T) = K((\lambda_1, \lambda_2)^T; \mathcal{I}_1)  = \argmin_{I\in \mathcal{I}_1} \left\lbrace\lambda_1\rho_1(Y_1) +  \lambda_2\rho_2(Y_2) \right\rbrace \, ,
\]
which corresponds to $K^*(\lambda)=K(\lambda)$, when $\lambda \in (0,1)$, in \eqref{eq:function_K_star}. 
\item 
For $\Lambda \in \mathcal{A}_1^{(2)}$, $n(\Lambda) = 2$,  $\Lambda[\cdot ,1] = (1,0)^T$, $\Lambda[\cdot ,2] = (0,1)^T$,
\[
K^{\ast(1)} ((1,0)^T) = K((1,0)^T; \mathcal{I}_1) = \argmin_{I\in \mathcal{I}_1}\rho_1(Y_1)\,, \text{ and}
\]
\[
K^{\ast} (\Lambda) = K^{\ast (2)}(\Lambda[\cdot ,1:2]) = K((0,1)^T; K^{\ast(1)} ((1,0)^T)) = \argmin_{I\in K^{\ast(1)} ((1,0)^T)} \rho_2(Y_2)\,,
\]
which corresponds to $K^*(1)$ in \eqref{eq:function_K_star}.
\item 
For $\Lambda \in \mathcal{A}_1^{(3)}$, $n(\Lambda) = 2$, $\Lambda[\cdot ,1] = (0,1)^T$, $\Lambda[\cdot ,2] = (1,0)^T$,
\[
K^{\ast(1)} ((0,1)^T) = K((0,1)^T; \mathcal{I}_1) = \argmin_{I\in \mathcal{I}_1}\rho_2(Y_2)\,, \text{ and}
\] 
\[
K^{\ast} (\Lambda) = K^{\ast (2)}(\Lambda[\cdot ,1:2]) = K((1,0)^T; K^{\ast(1)} ((0,1)^T)) = \argmin_{I\in K^{\ast(1)} ((0,1)^T)} \rho_1(Y_1)\,,
\]
which corresponds to $K^*(0)$ in \eqref{eq:function_K_star}.
\end{itemize}

Our notation and the generalized function $K$ make the sequential optimizations explicit through the functions $K^{*(j)}$, for $j=1,2$, which recursively record the decisions after each stage of the optimization procedure according to a matrix in $\mathcal{A}_1$.
Finally, we note that although $\Lambda[\cdot,1:2]=\Lambda$ and $\Lambda[\cdot,1:1]=\Lambda[\cdot,1]$, we intentionally use the longer notation to emphasize the recursive structure in \eqref{eq:function_K_star_recursion_n_1}; 
this will facilitate the extension to the multilateral case in the next section.


\subsection{Multilateral Centralized Risk Sharing}\label{sec:Multilateral}
This section discusses the multilateral setting ($n\in\mathbb{N}$), described by the tuple $(\mathcal{N}_n;\mathcal{I}_n)$. When $n\in\mathbb{N}\backslash\{1\}$, centralization becomes relevant, since there is more than one policyholder transferring risk to the central agent. We study Pareto optimal decisions
$\mathbf{I}^* \in\mathcal{I}_n$ by extending the results of Section \ref{sec:bilateral} 
as far as possible through out generalization of the functions $K$ and $K^*$ that capture sequential optimization.


Define the following sets of weights for convex and strict convex combinations:
$$\Delta_n[0,1]=\left\{{\bm\lambda}=\left(\lambda_1,\lambda_2,\dots, \lambda_{n+1}\right)^T\in\left[0,1\right]^{n+1}:\lambda_1+\lambda_2+\cdots +\lambda_{n+1}=1\right\}\, ,$$  
$$ \Delta_n(0,1)=\left\{{\bm\lambda}=\left(\lambda_1,\lambda_2,\dots, \lambda_{n+1}\right)^T\in\left(0,1\right)^{n+1}:\lambda_1+\lambda_2+\cdots +\lambda_{n+1}=1\right\}\, .$$
More generally, for $\bm\lambda\in\mathbb{R}^{n+1}$, let $\nzp(\bm\lambda)\subseteq\mathcal{N}_n$ denote the set of indices corresponding to the non-zero entries of $\bm\lambda$. Let $\mathcal{M}_{n+1}$ denote the set of square matrices of order $n+1$; for $\Lambda\in\mathcal{M}_{n+1}$, write $\Lambda=(\lambda_i^{(j)})_{i,j=1}^{n+1}$, where $i$ and $j$ denote the row and column indices, respectively; also, for $i,j=1,2,\dots,n+1$, let $\Lambda[i,\cdot]$, $\Lambda[\cdot,j]$, and $\Lambda[\cdot,1:j]$ denote, respectively, the $i$-th row of $\Lambda$, the $j$-th column of $\Lambda$, and the sub-matrix consisting of the first $j$ columns of $\Lambda$. Let $\mathcal{M}_{n+1}([0,1];0)\subseteq\mathcal{M}_{n+1}$ be the set of all non-zero matrices whose entries lie in $[0,1]$ and whose zero columns, if any, appear only as the rightmost columns. 
For $\Lambda\in\mathcal{M}_{n+1}([0,1];0)$, let $n(\Lambda)\in\{1,2,\dots,n+1\}$ denote the number of non-zero columns of $\Lambda$. Moreover, every
$\Lambda\in \mathcal{M}_{n+1}([0,1]; 0)$ can be written in terms of its columns as, $\Lambda=[\bm\lambda^{(1)};\bm\lambda^{(2)}; \dots ; \bm\lambda^{(n(\Lambda))};\mathbf{0};\dots; \mathbf{0}]$, where $\mathbf{0}$ denotes the zero vector.

Below, we define a set of square matrices that covers all sequential optimizations for the multilateral centralized risk sharing problem.
\begin{definition}
Define the set of sequential matrices $\mathcal{A}_n$ by:
\begin{align}\label{A_n}
\mathcal{A}_n   = \biggr\{   & \nonumber  \Lambda = (\lambda_{i}^{(j)})_{i,j=1}^{n+1}=[\bm\lambda^{(1)}; \bm\lambda^{(2)};\dots ; \bm\lambda^{(n(\Lambda))}; \mathbf{0}; \dots ; \mathbf{0}]    \in \mathcal{M}_{n+1} ([0,1]; 0): \\ 
&(\text{a})\;\Lambda[\cdot, j]=\bm\lambda^{(j)} \in \Delta_{n}[0,1],   \text{ for } j=1,2,\dots, n(\Lambda) \,, \text{ and,} \\
&(\text{b})\;\text{for all }i\in\mathcal{N}_n, \,   \nzp(\Lambda[i,\cdot]^T)  = \{j\}, \text{ for some } j\in\{1,2,\dots, n(\Lambda)\}\biggr\}\,.\nonumber
\end{align}
For any $\Lambda =  (\lambda_{i}^{(j)})_{i,j=1}^{n+1}\in \mathcal{A}_n$, the entry $\lambda_i^{(j)}$
is the weight associated with the agent $i$ at the $j$-th level of the $n(\Lambda)$ sequential optimizations.\qed
\end{definition}
\noindent
The definition of $\mathcal{A}_n$ naturally generalizes the matrices in $\mathcal{A}_1$ of \eqref{eq:A1} covering all possible sequential optimizations when $n\in\mathbb{N}$.
Fix a matrix $\Lambda = (\lambda_{i}^{(j)})_{i,j=1}^{n+1} = [\bm\lambda^{(1)}; \bm\lambda^{(2)}; \dots ; \bm\lambda^{(n(\Lambda))}; \mathbf{0}; \dots ; \mathbf{0}] \in \mathcal{A}_n$. The matrix $\Lambda$ has between one and $n+1$ non-zero columns, all of which are leftmost. It defines a sequential optimization procedure with $n(\Lambda)\in\{1,2,\dots,n+1\}$ stages, in which $\lambda_i^{(j)}$ is the weight assigned to the agent $i$ at stage $j$.
The first condition in the definition of $\mathcal{A}_n$ specifies, through the positive weights, the agents whose risks are optimized at each stage; the positive weights in each non-zero column are normalized to sum to one. The second condition states that each agent is assigned exactly one positive weight across all stages, and this weight may appear at any stage of the sequential optimization procedure. Together, these two conditions in $\mathcal{A}_n$ generate all possible partitions of the agents, as noted below.

\begin{remark}\label{rem:Anpartition}
Given a matrix 
$\Lambda = 
[\bm\lambda^{(1)}; \bm\lambda^{(2)}; \dots ; \bm\lambda^{(n(\Lambda))}; \mathbf{0}; \dots ; \mathbf{0}] \in \mathcal{A}_n$, the sets defined by $\mathcal{G}_j = \nzp(\bm{\lambda}^{(j)})$, for $j=1,2,\dots,n(\Lambda)$, form a partition of $\mathcal{N}_n$;
that is: (i) $\mathcal{G}_j\neq \emptyset $, for all $j=1,2,\dots, n(\Lambda)$, and (ii) $\Nn = \cup_{j=1}^{n(\Lambda)}\mathcal{G}_j $, such that $\mathcal{G}_j\cap \mathcal{G}_k=\emptyset$, for all $j,k=1,2,\dots , n(\Lambda)$ with $j\neq k$. The detailed arguments are in Appendix~\ref{ap:Anpartition}.
%
\end{remark}

To define the sets of optimal decisions arising from sequential optimization, we next generalize the functions $K$, $K^{*(j)}$, and $K^*$ defined in \eqref{eq:K_lambda_C_1}, \eqref{eq:function_K_star_recursion_n_2}, \eqref{eq:function_K_star_recursion_n_1}, and \eqref{eq:function_K_star_1} to any $n\in\mathbb{N}$.
\begin{definition}\label{def:K*_n}
Given the full setting $(\Nn;\In)$ for $n\in\mathbb{N}$ and any feasible subset $\mathcal{C}\subseteq\mathcal{I}_n$, define for any ${\bm\lambda}\in\Delta_n[0,1]$, the set of minimizers
\begin{equation}\label{eq:K_lambda_c}
K\left({\bm\lambda};\mathcal{C}\right)
=\argmin_{\I=\left(I_1,I_2,\dots, I_n\right)\in\mathcal{C}}\left\lbrace\lambda_1\rho_1\left(Y_1\right)+\lambda_2\rho_2\left(Y_2\right)+\cdots + \lambda_{n+1}\rho_{n+1}\left(Y_{n+1}\right)\right\rbrace\,.
\end{equation}
If $\mathcal{C}=\mathcal{I}_n$, we simply write $K\left({\bm\lambda};\mathcal{I}_n\right)$ as $K(\bm\lambda)$. For any $\Lambda\in\mathcal{A}_n$, define recursively:
\begin{equation}\label{eq:K_star_recursion_1}
K^{*\left(1\right)}\left(\Lambda[\cdot,1]\right)
=K\left(\Lambda[\cdot,1];\mathcal{I}_n\right)
=K\left(\Lambda[\cdot,1]\right)\, ,
\end{equation}
\begin{equation}\label{eq:K_star_recursion_n}
K^{*\left(j\right)}\left(\Lambda[\cdot,1:j]\right)=K\left(\Lambda[\cdot, j];K^{*\left(j-1\right)}\left(\Lambda[\cdot,1:j-1]\right)\right),\text{ for }j=2,3,\dots,n\left(\Lambda\right)\,.
\end{equation}
Finally, the set-valued function $K^*:\mathcal{A}_n \rightarrow 2^{\mathcal I_n}$ is defined by, for $\Lambda\in\mathcal{A}_n$,
\begin{equation}\label{eq:function_K_star_n}
K^*\left(\Lambda\right)=K^{*\left(n\left(\Lambda\right)\right)}\left(\Lambda[\cdot,1:n\left(\Lambda\right)]\right)\, .
\end{equation}  
The set $\cup_{\Lambda\in \An} K^*(\Lambda)$ consists of all decision vectors arising from the sequential optimizations encoded by the sequential matrices in $\mathcal{A}_n$ via the function $K^*$.\qed
\end{definition}

The following proposition shows that the convexity of a feasible subset is preserved after each stage of the optimization of a convex combination, provided that the assumptions $(\dagger)$ hold. The proof of the following Proposition is in Appendix~\ref{ap:n_K_lambda_c_convex}.
\begin{proposition}\label{n_K_lambda_c_convex}
Let the full setting $(\Nn;\In)$, for $n\in\mathbb{N}$, satisfy $(\dagger)$, let a feasible subset $\mathcal{C}\subseteq \mathcal{I}_n$, and ${\bm\lambda}\in\Delta_n[0,1]$. If $\mathcal{C}$ is convex, then 
$K\left({\bm\lambda};\mathcal{C}\right)$, as defined in \eqref{eq:K_lambda_c}, is convex.
\end{proposition}

\begin{corollary}
Let the full setting $(\Nn;\In)$, for $n\in\mathbb{N}$, satisfy $(\dagger)$, and let $\Lambda\in\mathcal{A}_n$. Then, for each $j=1,2,\dots, n(\Lambda)$, the set $K^{*\left(j\right)}\left(\Lambda[\cdot,1:j]\right)$, as defined in \eqref{eq:K_star_recursion_1} and \eqref{eq:K_star_recursion_n}, is convex. In particular, the set $K^*\left(\Lambda\right)$, as defined in \eqref{eq:function_K_star_n}, is convex.
\end{corollary}
\begin{proof}
Since $\mathcal{I}_n$ is convex and, for each $j=1,2,\dots,n(\Lambda)$, we have $\Lambda[\cdot,j]\in\Delta_n[0,1]$, repeated application of Proposition \ref{n_K_lambda_c_convex} yields the result.
\end{proof}

\begin{example}\label{ex:A2}
We illustrate Definition \ref{def:K*_n} for $n=2$, corresponding to the centralized setting with two policyholders and one insurer. First, observe that the matrices in $\mathcal{A}_2$ fall into 13 classes; that is, we may write $\mathcal{A}_2= \mathcal{A}_2^{(1)}\cup \mathcal{A}_2^{(2)}\cup \cdots \cup \mathcal{A}_2^{(13)} $, where 
\[
\mathcal{A}_2^{(1)} = \left\{ \begin{pmatrix}
\lambda_1 & 0 & 0\\
\lambda_2 & 0 & 0\\
\lambda_3 & 0 & 0\\
\end{pmatrix}:\, \begin{array}{l}  \lambda_1,\lambda_2,\lambda_3\in(0,1)\\ \lambda_1+\lambda_2+\lambda_3=1\end{array}\right\},
\]
\[\mathcal{A}_2^{(2)} = \left\{ \begin{pmatrix}
\lambda_1 & 0 & 0\\
\lambda_2 & 0 & 0\\
0 & 1 & 0\\ 
\end{pmatrix}  :\, \begin{array}{l}  \lambda_1, \lambda_2\in (0,1) \\\lambda_1+\lambda_2=1
 \end{array} \right\},\,  
\mathcal{A}_2^{(3)} = \left\{   \begin{pmatrix}
\lambda_1 & 0 & 0\\
0 & 1 & 0\\
\lambda_3 & 0 & 0\\ 
\end{pmatrix} :\,\begin{array}{l} \lambda_1, \lambda_3\in (0,1)  \\\lambda_1+\lambda_3=1
 \end{array} \right\}\, ,\]
\[\mathcal{A}_2^{(4)} = \left\{ \begin{pmatrix}
0 & 1 & 0\\
\lambda_2 & 0 & 0\\
\lambda_3 & 0 & 0\\ 
\end{pmatrix}  :\, \begin{array}{l}  \lambda_2, \lambda_3\in (0,1) \\\lambda_2+\lambda_3=1
 \end{array} \right\},\,  
\mathcal{A}_2^{(5)} = \left\{   \begin{pmatrix}
1 & 0 & 0\\
0 & \lambda_2 & 0\\
0 & \lambda_3 & 0\\ 
\end{pmatrix} :\,\begin{array}{l}  \lambda_2, \lambda_3\in (0,1) \\\lambda_2+\lambda_3=1
 \end{array} \right\}\, ,\]
\[\mathcal{A}_2^{(6)} = \left\{ \begin{pmatrix}
0 & \lambda_1 & 0\\
1 & 0 & 0\\
0 & \lambda_3 & 0\\ 
\end{pmatrix}  :\, \begin{array}{l}  \lambda_1, \lambda_3\in (0,1) \\\lambda_1+\lambda_3=1
 \end{array} \right\},\,  
\mathcal{A}_2^{(7)} = \left\{   \begin{pmatrix}
0 & \lambda_1 & 0\\
0 & \lambda_2 & 0\\
1 & 0 & 0\\ 
\end{pmatrix} :\,\begin{array}{l}  \lambda_1, \lambda_2\in (0,1)\\\lambda_1+\lambda_2=1 
 \end{array} \right\}\, , \]
\[\mathcal{A}_2^{(8)} = \left\{ \begin{pmatrix}
1 & 0 & 0\\
0 & 1 & 0\\
0 & 0 & 1\\
\end{pmatrix}\right\},\,  
\mathcal{A}_2^{(9)} =  \left\{\begin{pmatrix}
1 & 0 & 0\\
0 & 0 & 1\\
0 & 1 & 0\\
\end{pmatrix}\right\},\, 
\mathcal{A}_2^{(10)} =  \left\{\begin{pmatrix}
0 & 1 & 0\\
1 & 0 & 0\\
0 & 0 & 1\\
\end{pmatrix}\right\}\, ,
\]
\[\mathcal{A}_2^{(11)} = \left\{ \begin{pmatrix}
0 & 0 & 1\\
1 & 0 & 0\\
0 & 1 & 0\\
\end{pmatrix}\right\},\,  
\mathcal{A}_2^{(12)} =  \left\{\begin{pmatrix}
0 & 1 & 0\\
0 & 0 & 1\\
1 & 0 & 0\\
\end{pmatrix}\right\},\, 
\mathcal{A}_2^{(13)} =  \left\{\begin{pmatrix}
0 & 0 & 1\\
0 & 1 & 0\\
1 & 0 & 0\\
\end{pmatrix}\right\}\, .
\]
The sequential optimizations follow from \eqref{eq:function_K_star_n} using a representative matrix from each class 
in $\mathcal{A}_2$, as we show in Appendix~\ref{ap:ex_2}.\qed
\end{example}
\noindent
As illustrated by Example \ref{ex:A2}, the conditions defining the set $\mathcal{A}_n$ in \eqref{A_n} encode all possible orderings of the sequential optimization steps involving the $n+1$ agents, with each agent's risk measure appearing exactly once in the sequence. 

Recall that, for $n=1$, $\mathcal{A}_1$ can be partitioned into three distinct classes of sequential matrices. As shown in Example \ref{ex:A2}, for $n=2$,  $\mathcal{A}_2$ can be partitioned into 13 classes. Each class specifies a unique sequential ordering of the optimization problems defined through the function $K^*$. More generally, for $n\in\mathbb{N}$, let $\mathrm{d}(\mathcal{A}_n)$ denote the number of distinct classes of sequential matrices in $\mathcal{A}_n$; we call it the dimension of $\mathcal{A}_n$. These numbers can be computed recursively; the corresponding formula is relegated to Appendix~\ref{ap:dimension_An}. For instance, $\mathrm{d}(\mathcal{A}_1)=3$, $\mathrm{d}(\mathcal{A}_2)=13$, $\mathrm{d}(\mathcal{A}_3)=75$.

With the generalized notation and definitions for sequential optimization in place, we now explore the extension of the results on Pareto optimality obtained for the bilateral risk sharing problem ($n=1$), namely Proposition~\ref{2_min_PO}, Theorem~\ref{2_PO_min}, Theorem~\ref{2_K^*_PO}, Theorem~\ref{2_PO_K^*}, and, ultimately, Corollary~\ref{2_result}, to the multilateral centralized risk sharing problem ($n\in\mathbb{N}$). First, the following proposition extends Proposition~\ref{2_min_PO} to the case $n\in\mathbb{N}$, which is a standard result.
\begin{proposition}\label{n_min_PO}
Consider the full setting $(\mathcal{N}_n;\In)$, for $n\in\mathbb{N}$. Let a group of agents $\mathcal{G}\subseteq\mathcal{N}_n$, and let a feasible subset $\mathcal{C}\subseteq \mathcal{I}_n$. If $\mathbf{I}^*= \left(I^*_1,I^*_2,\dots, I_n^*\right)\in K\left({\bm\lambda};\mathcal{C}\right)$, as defined in \eqref{eq:K_lambda_c}, for some ${\bm\lambda}\in \Delta_n[0,1]$ with $\nzp(\bm{\lambda})=\mathcal{G}$, then $\mathbf{I}^*\in\mathcal{C}$ is Pareto optimal among 
$\mathcal{G}$ within 
$\mathcal{C}$. In particular, if $\mathbf{I}^*\in K\left({\bm\lambda}\right)$ for some ${\bm\lambda}\in \Delta_n(0,1)$, then $\mathbf{I}^*\in\In$ is Pareto optimal under the full setting.
\end{proposition}
\begin{proof}
Suppose that the decision $\mathbf{I}^*\in\mathcal{C}$ is not Pareto optimal among the agents in the group $\mathcal{G}$ within the feasible subset $\mathcal{C}$. Then there exists $\mathbf{I}\in\mathcal{C}$ such that $\rho_i(Y_i) \leq \rho_i(Y_i^*) $,  for all $i\in\mathcal{G}$, with at least one strict inequality. It follows that $\sum_{i\in\mathcal{G}} \lambda_i\rho_i\left(Y_i \right) < \sum_{i\in\mathcal{G}} \lambda_i\rho_i\left(Y_i^*\right)$, for the given ${\bm\lambda}\in \Delta_n[0,1]$ with $\nzp(\bm{\lambda})=\mathcal{G}$, which implies that $\mathbf{I}^*\notin K\left({\bm\lambda};\mathcal{C}\right)$. This is a contradiction.
\end{proof}

We next extend Theorem~\ref{2_PO_min}, showing that, under the assumptions $(\dagger)$, any Pareto optimal decision $\mathbf{I}^*$, among the agents in a group, minimizes a 
convex combination of the risks for a sub-group of these agents; that is, a Pareto optimal decision does not necessarily minimize a strict convex combination of the risks for all agents in 
consideration. Note that, for each feasible decision $\mathbf{I}\in \mathcal{I}_n$, the mapped vector $(\rho_i(Y_i))_{i\in\mathcal{G}}$, for some $\mathcal{G}\subseteq\mathcal{N}_n$, is an element of $\mathbb{R}^{\vert\mathcal{G}\vert}$.\footnote{We denote by $\vert E\vert $ the cardinality of a set $E$.} Pareto optimality is applied to the image of this mapping, and thus the concept of Pareto optimality in Euclidean spaces is useful; Appendix~\ref{ap:POsets} reviews this concept.
\begin{theorem}\label{n_PO_min}
Let the full setting $(\Nn;\In)$, for $n\in\mathbb{N}$, satisfy $(\dagger)$. Let a group of agents $\mathcal{G}\subseteq\mathcal{N}_n$, and let a feasible subset $\mathcal{C}\subseteq \mathcal{I}_n$. If $\mathbf{I}^*= \left(I^*_1,I^*_2,\dots, I_n^*\right)\in\mathcal{C}$ is Pareto optimal 
among $\mathcal{G}$ within 
$\mathcal{C}$, then $\mathbf{I}^*\in K\left({\bm\lambda};\mathcal{C}\right)$, as defined in \eqref{eq:K_lambda_c}, for some ${\bm\lambda}\in \Delta_n[0,1]$ with $\nzp(\bm{\lambda})\subseteq\mathcal{G}$. In particular, if $\mathbf{I}^*\in\In$ is Pareto optimal under the full setting, then $\mathbf{I}^*\in K\left({\bm\lambda}\right)$ for some ${\bm\lambda}\in \Delta_n[0,1]$.
\end{theorem}

\begin{proof}
Define the set
\begin{equation*}
S=\{\mathbf{y} = (y_i)_{i\in\mathcal{G}}\in\mathbb{R}^{\vert\mathcal{G}\vert}: y_i= \rho_i\left(Y_i\right), \,i\in\mathcal{G}, \,   \mathbf{I} \in \mathcal{C}\}.
\end{equation*}
For the Pareto optimal $\mathbf{I}^*\in\mathcal{C}$, among the agents in the group $\mathcal{G}$ within the feasible subset $\mathcal{C}$, define the point $\mathbf{y}^{\ast,\mathcal{G}} = (y^\ast_i)_{i\in\mathcal{G}}\in S$, where $y^*_i= \rho_i\left(Y^*_i\right)$, for $i\in\mathcal{G}$. By definition, $\mathbf{y}^{\ast,\mathcal{G}}\in S$ is a Pareto optimal point in $S$. Denote by $\overline{S}=\mathrm{conv}(S)$ the convex hull of $S$; see Appendix~\ref{ap:POsets} to recall its definition.


The remaining steps of the proof are outlined as follows. (a) We verify that, for any $\overline{\mathbf{y} }\in \overline{S}$, there exists $\mathbf{y} \in S$ such that $\mathbf{y} \leq \overline{\mathbf{y}} $.\footnote{Unless otherwise specified, (in)equalities between vectors are component-wise.} (b) We prove that the Pareto optimal $\mathbf{y}^{\ast,\mathcal{G}}\in S$ is also Pareto optimal in $\overline{S}$. (c) We show that there exists $\bm{\lambda} = (\lambda_1,\lambda_2,\dots, \lambda_{n+1})^T\in \Delta_n[0,1]$ such that $\nzp(\bm{\lambda})\subseteq\mathcal{G}$ and
\begin{equation}
\mathbf{y}^{\ast}\in\argmin\left\{\bm{\lambda}^T\mathbf{y}:\mathbf{y}=\left(\rho_i\left(Y_i\right)\right)_{i\in\mathcal{N}_n}\in\mathbb{R}^{n+1},\mathbf{I} \in \mathcal{C}\right\},
\label{eq:y^*_equation}
\end{equation}
where $\mathbf{y}^{\ast}=\left(\rho_i\left(Y^*_i\right)\right)_{i\in\mathcal{N}_n}$ for the $\mathbf{I}^*\in\mathcal{C}$, being Pareto optimal among the agents in the group $\mathcal{G}$ within the feasible subset $\mathcal{C}$; this concludes that $\mathbf{I}^*\in K\left({\bm\lambda};\mathcal{C}\right)$.

(a) Let $\overline{\mathbf{y}}\in \overline{S}$. By definition, there exist $m\in\mathbb{N}$, $\mathbf{y}^{(j)}\in S$, $\theta_j \geq 0$, for $j=1,2,\dots,m$, such that $\sum_{j=1}^m\theta_j = 1$ and $\overline{\mathbf{y}}=\sum_{j=1}^m \theta_j\mathbf{y}^{(j)}$; since $\mathbf{y}^{(j)}\in S$, there exists $\mathbf{I}^{(j)}=(I^{(j)}_1,I^{(j)}_2,\dots, I^{(j)}_n)\in\mathcal{C}$ such that
\begin{equation*}
\mathbf{y}^{(j)}=
\begin{cases}
\left(\left(\rho_i\left(Y_i(I^{(j)}_{i})\right)\right)_{i\in\mathcal{G}\backslash\{n+1\}},\rho_{n+1}\left(Y_{n+1}(\mathbf{I}^{(j)})\right)\right)^T,&\text{ if }(n+1)\in\mathcal{G},\\
\left(\rho_i\left(Y_i(I^{(j)}_{i})\right)\right)_{i\in\mathcal{G}},&\text{ if }(n+1)\notin\mathcal{G}.
\end{cases}
\end{equation*}
Due to the convexity of $\mathcal{C}$, $\sum_{j=1}^m\theta_j\mathbf{I}^{(j)}=(\sum_{j=1}^m\theta_jI^{(j)}_1,\sum_{j=1}^m\theta_jI^{(j)}_2,\dots, \sum_{j=1}^m\theta_jI^{(j)}_n)\in\mathcal{C}$, and thus a point $\mathbf{y}\in S$ is defined as:
\begin{equation*}
\mathbf{y}=
\begin{cases}
\left(\left(\rho_i\left(Y_i(\sum_{j=1}^m\theta_jI^{(j)}_i)\right)\right)_{i\in\mathcal{G}\backslash\{n+1\}},\rho_{n+1}\left(Y_{n+1}(\sum_{j=1}^m\theta_j\mathbf{I}^{(j)})\right)\right)^T,&\text{ if }(n+1)\in\mathcal{G},\\
\left(\rho_i\left(Y_i(\sum_{j=1}^m\theta_jI^{(j)}_i)\right)\right)_{i\in\mathcal{G}},&\text{ if }(n+1)\notin\mathcal{G}.
\end{cases}
\end{equation*}
By Lemma~\ref{rhoconvexI}, $\mathbf{y} \leq \overline{\mathbf{y}}$.
(b) Note that $\mathbf{y}^{\ast,\mathcal{G}}\in S \subseteq\overline{S}$. Assume, on the contrary, that $\mathbf{y}^{\ast,\mathcal{G}}$ is not Pareto optimal in $\overline{S}$, and thus there exists $\overline{\mathbf{y}}\in\overline{S}$ such that $\overline{\mathbf{y}} \leq \mathbf{y}^{\ast,\mathcal{G}}$ with $\overline{\mathbf{y}} \neq \mathbf{y}^{\ast,\mathcal{G}}$. By (a), there exists $\tilde{\mathbf{y}}\in S$ such that $\tilde{\mathbf{y}} \leq \overline{\mathbf{y}}$, and consequently, $\tilde{\mathbf{y}} \leq \overline{\mathbf{y}} \leq \mathbf{y}^{\ast,\mathcal{G}}$ with $\tilde{\mathbf{y}} \neq \mathbf{y}^{\ast,\mathcal{G}}$, which contradicts that $\mathbf{y}^{\ast,\mathcal{G}}$ is Pareto optimal in $S$. Therefore, $\mathbf{y}^{\ast,\mathcal{G}}$ is Pareto optimal in $\overline{S}$.

(c) Define the set 
$
T_{\mathbf{y}^{\ast,\mathcal{G}} } = \lbrace  \mathbf{y}= (y_i)_{i\in\mathcal{G}}\in\mathbb{R}^{\vert\mathcal{G}\vert} : \,  \mathbf{y}\leq \mathbf{y}^{\ast,\mathcal{G}} 
\rbrace
$
for the Pareto optimal point $\mathbf{y}^{\ast,\mathcal{G}} $ in $S$ and $\overline{S}$. Clearly, $T_{\mathbf{y}^{\ast,\mathcal{G}}}$ and $\overline{S}$ are both convex sets. By Appendix~\ref{ap:POsets}, with $C=\overline{S}$ and in $\mathbb{R}^{\vert\mathcal{G}\vert}$, there exists $\bm{a}=(a_i)_{i\in\mathcal{G}}\in \mathbb{R}^{\vert\mathcal{G}\vert}$, $\bm{a}\neq \mathbf{0} $, 
such that $\sup_{\mathbf{y}\in T_{\mathbf{y}^{\ast,\mathcal{G}}}}\{\bm{a}^T \mathbf{y}\}\leq \inf_{\mathbf{y}\in \overline{S}}\{\bm{a}^T \mathbf{y}\}$. We argue that $\bm{a}\geq \mathbf{0}$. Assume, on the contrary that, $a_{i_0}<0$ for some $i_0\in\mathcal{G}$. Then, take the point 
$\tilde{\mathbf{y}} = (\tilde{y}_i)_{i\in \mathcal{G}}\in T_{\mathbf{y}^{\ast,\mathcal{G}}}$, with $\tilde{y}_i=y_i^*$ for $i\neq i_0$, and $\tilde{y}_{i_0} = y_{i_0^*} -1$.
Since $a_{i_0}<0$, $a_{i_0}y_{i_0}^*<a_{i_0}(y_{i_0}^*-1)$, and thus a contradiction is induced from $\bm{a}^T\mathbf{y}^{\ast,\mathcal{G}}  <  \bm{a}^T\tilde{\mathbf{y}}\leq\sup_{\mathbf{y}\in T_{\mathbf{y}^{\ast,\mathcal{G}}}}\{\bm{a}^T\mathbf{y}\}\leq \inf_{\mathbf{y}\in \overline{S}}\{\bm{a}^T \mathbf{y}\}\leq\bm{a}^T\mathbf{y}^{\ast,\mathcal{G}}$, where the last inequality is due to that $\mathbf{y}^{\ast,\mathcal{G}}\in \overline{S}$. Therefore, $\bm{a}\geq \mathbf{0}$.


Since $\bm{a}\geq \mathbf{0}$, $\bm{a}^T \mathbf{y}^{\ast,\mathcal{G}}=\sup_{\mathbf{y}\in T_{\mathbf{y}^{\ast,\mathcal{G}}}}\{\bm{a}^T \mathbf{y}\}\leq \inf_{\mathbf{y}\in \overline{S}}\{\bm{a}^T \mathbf{y}\}$. Since $\mathbf{y}^{\ast,\mathcal{G}}\in \overline{S}$, $\mathbf{y}^{\ast,\mathcal{G}}$ attains the minimum for $\inf_{\mathbf{y}\in \overline{S}}\{\bm{a}^T \mathbf{y}\}$; that is, $\inf_{\mathbf{y}\in \overline{S}}\{\bm{a}^T \mathbf{y}\}=\min_{\mathbf{y}\in \overline{S}}\{\bm{a}^T \mathbf{y}\}=\bm{a}^T \mathbf{y}^{\ast,\mathcal{G}}$. Again, by the fact that $\mathbf{y}^{\ast,\mathcal{G}}\in \overline{S}$, and by (a), there exists $\mathbf{y}\in S$ such that $\mathbf{y}\leq \mathbf{y}^{\ast,\mathcal{G}}$; since $\bm{a}\geq \mathbf{0}$, $\bm{a}^T\mathbf{y} \leq \bm{a}^T\mathbf{y}^{\ast,\mathcal{G}}$, which implies that $\inf_{\mathbf{y}\in S}\{\bm{a}^T\mathbf{y}\} \leq \bm{a}^T\mathbf{y}^{\ast,\mathcal{G}}$. Moreover, $S\subseteq \overline{S}$ implies that $\inf_{\mathbf{y}\in \overline{S}}\{\bm{a}^T\mathbf{y}\} \leq \inf_{\mathbf{y}\in S}\{\bm{a}^T\mathbf{y}\}$. These together lead to that
\begin{equation*}
\inf_{\mathbf{y}\in S}\{\bm{a}^T\mathbf{y}\} \leq \bm{a}^T\mathbf{y}^{\ast,\mathcal{G}}=\min_{\mathbf{y}\in \overline{S}}\{\bm{a}^T \mathbf{y}\}=\inf_{\mathbf{y}\in \overline{S}}\{\bm{a}^T \mathbf{y}\}\leq \inf_{\mathbf{y}\in S}\{\bm{a}^T\mathbf{y}\},
\end{equation*}
showing that $\mathbf{y}^{\ast,\mathcal{G}}\in S$ attains the minimum for $\inf_{\mathbf{y}\in S}\{\bm{a}^T\mathbf{y}\}$.

Finally, define $\bm\lambda=\left(\lambda_1,\lambda_2,\dots, \lambda_{n+1}\right)^T $, with $\lambda_i= a_i/(\sum_{j\in\mathcal{G}}a_j)$ for $i\in \mathcal{G}$, and $\lambda_i = 0$ for $i\in \Nn\setminus\mathcal{G}$.
Clearly, ${\bm\lambda}=\left(\lambda_1,\lambda_2,\dots, \lambda_{n+1}\right)^T\in \Delta_n[0,1]$ and is such that $\nzp(\bm{\lambda})\subseteq\mathcal{G}$. Note that \eqref{eq:y^*_equation} indeed holds. For any $\mathbf{I} \in \mathcal{C}$, with $\mathbf{y}^{\mathcal{G}}=\left(\rho_i\left(Y_i\right)\right)_{i\in\mathcal{G}}\in S$ and $\mathbf{y}=\left(\rho_i\left(Y_i\right)\right)_{i\in\mathcal{N}_n}\in\mathbb{R}^{n+1}$,
\begin{align*}
\bm{\lambda}^T\mathbf{y}^{\ast}=&\;\sum_{i\in\mathcal{N}_n}\lambda_i\rho_i\left(Y^*_i\right)=\sum_{i\in\mathcal{G}}\lambda_i\rho_i\left(Y^*_i\right)+\sum_{i\in\mathcal{N}_n\backslash\mathcal{G}}\lambda_i\rho_i\left(Y^*_i\right)=\frac{1}{\sum_{j\in\mathcal{G}}a_j}\sum_{i\in\mathcal{G}}a_i\rho_i\left(Y^*_i\right)\\=&\;\frac{1}{\sum_{j\in\mathcal{G}}a_j}\bm{a}^T\mathbf{y}^{\ast,\mathcal{G}}\leq\frac{1}{\sum_{j\in\mathcal{G}}a_j}\bm{a}^T\mathbf{y}^{\mathcal{G}}=\frac{1}{\sum_{j\in\mathcal{G}}a_j}\sum_{i\in\mathcal{G}}a_i\rho_i\left(Y_i\right)\\=&\;\sum_{i\in\mathcal{G}}\lambda_i\rho_i\left(Y_i\right)+\sum_{i\in\mathcal{N}_n\backslash\mathcal{G}}\lambda_i\rho_i\left(Y_i\right)=\sum_{i\in\mathcal{N}_n}\lambda_i\rho_i\left(Y_i\right)=\bm{\lambda}^T\mathbf{y},
\end{align*}
in which the inequality is due to that $\mathbf{y}^{\ast,\mathcal{G}}\in S$ attains the minimum for $\inf_{\mathbf{y}\in S}\{\bm{a}^T\mathbf{y}\}$, and $\mathbf{y}^{\mathcal{G}}\in S$.
\end{proof}


Denote by $\mathrm{PO}$ the set of Pareto optimal decisions under the full setting. Figure~\ref{fig:n_results12} illustrates Proposition \ref{n_min_PO} and Theorem \ref{n_PO_min},\footnote{Again, under the assumptions $(\dagger)$, the smaller and the larger sets of decisions than the PO set are respectively known as Geoffrion-proper Pareto optimality and weak Pareto optimality.} thereby extending Proposition~\ref{2_min_PO}, Theorem~\ref{2_PO_min}, and Figure~\ref{fig:2_results12} from the bilateral case. In analogy with the bilateral setting, by extending the function $K$ in \eqref{eq:K_lambda_c} to the function $K^*$ in \eqref{eq:function_K_star_n}, we aim to study the relationship between Pareto optimality and sequential optimizations under the full setting. The following result extends Theorem \ref{2_K^*_PO}.
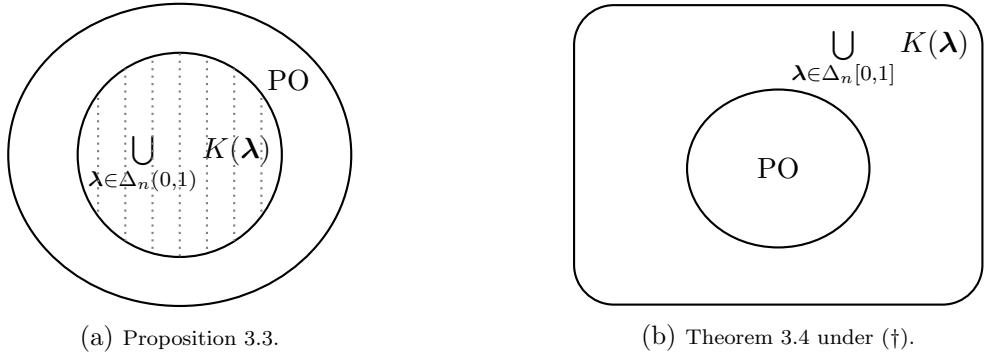
\begin{figure}[h!]
\centering
\begin{subfigure}[C]{0.45\linewidth}
\centering
\begin{tikzpicture}[thick,fill,scale=0.9]
\draw (0,0) circle (1.5cm);
\node at (0,-0.15) {$ \bigcup\limits_{\bm{\lambda}\in \Delta_n(0,1)} K(\bm\lambda)$};
\draw (0,0) ellipse (2.52 and 2.22);
\node at (1.6,1.1) {$\mathrm{PO}$};

\begin{scope} \clip (0,0) circle (1.5cm); 
\foreach \x in {-1.2,-0.8,-0.4,0,0.4,0.8,1.2} { \draw[gray, dotted] (\x,-1.6) -- (\x,1.6); } 
\end{scope}
\end{tikzpicture}
\caption{\scriptsize{Proposition~\ref{n_min_PO}. }}\label{fig:n_min_PO}
\end{subfigure}%
~
\begin{subfigure}[c]{0.45\textwidth}
\centering
\begin{tikzpicture}[thick,fill,scale=0.9]
\draw[rounded corners=15pt] (0,0) rectangle (6,4.4); 
\draw (3,2) ellipse (1.34 and 1.16);
\node at (3,2) {$\mathrm{PO}$};
\node at (4.5,3.6) {$\bigcup\limits_{\bm\lambda\in \Delta_n[0,1]} K(\bm\lambda)$};
\end{tikzpicture}
\caption{\scriptsize{Theorem~\ref{n_PO_min} under $(\dagger)$.}}\label{fig:n_PO_min}
\end{subfigure}
\caption{\footnotesize{Illustration of Proposition~\ref{n_min_PO} and Theorem~\ref{n_PO_min}.}}\label{fig:n_results12}
\end{figure}

\begin{theorem}\label{n_K^*_PO}
Let the full setting $(\mathcal{N}_n;\In)$, for $n\in\mathbb{N}$. If $\mathbf{I}^*= \left(I^*_1,I^*_2,\dots, I_n^*\right)\in K^*(\Lambda)$, as defined in \eqref{eq:function_K_star_n}, for some $\Lambda\in \mathcal{A}_n$, then $\mathbf{I}^*\in\In$ is Pareto optimal under the full setting.
\end{theorem}
\begin{proof} 
Recall the notation and discussions in Remark~\ref{rem:Anpartition}. By definition,
\begin{equation}\label{eq:I_Kstar_(1)}
\I^*\in K^{*(1)}(\Lambda[\cdot, 1]) = K(\bm{\lambda}^{(1)}; \mathcal{I}_n) = \argmin_{\I\in \mathcal{I}_n} \sum_{i\in \mathcal{G}_1} \lambda_i^{(1)} \rho_i(Y_i)\, ,
\end{equation}
and, for $j=2,3,\dots, n(\Lambda)$,
\begin{equation}\label{eq:I_Kstar_(j)}
\I^*\in K^{*(j)}(\Lambda[\cdot, 1:j]) = K(\bm{\lambda}^{(j)}; K^{*(j-1)}(\Lambda[\cdot, 1:j-1])) = \argmin_{\I \in  K^{*(j-1)}(\Lambda[\cdot, 1:j-1]) } \sum_{i\in \mathcal{G}_j} \lambda_i^{(j)} \rho_i(Y_i) \, .
\end{equation}

Assume, on the contrary, that $\mathbf{I}^*$ is not Pareto optimal under the full setting; that is, there exists $\mathbf{I} = \left(I_1,I_2,\dots,I_n\right)\in\mathcal{I}_n$ such that, $\rho_i\left(Y_i\right)\leq\rho_i\left(Y^*_i\right)$, for any $i\in\mathcal{N}_n$, with at least one inequality being strict. Let $D=\{i\in\mathcal{N}_n:\rho_i(Y_i)<\rho_i(Y_i^*)\}$ be the set of strictly improved agents. Under this assumption, $D\neq\emptyset$. Since $\Lambda\in\mathcal{A}_{n}$, the sets $\mathcal{G}_j$, for $j=1,2,\dots,n(\Lambda)$, form a partition of $\mathcal{N}_n$, in particular $\mathcal{N}_n=\cup_{j=1}^{n(\Lambda)}\mathcal{G}_j$. Therefore, there exists at least one stage $j\in\{1,2,\dots,n(\Lambda)\}$ such that $\mathcal{G}_j\cap D\neq\emptyset$. Let $j_0=\min\{j\in\{1,2,\dots,n(\Lambda)\}:\mathcal{G}_j\cap D\neq\emptyset\}$ be the earliest stage at which a strictly improved agent appears, and choose any $i_0\in \mathcal{G}_{j_0}\cap D$.

Suppose that $j_0=1$. Since $\rho_{i_0}\left(Y_{i_0}\right)<\rho_{i_0}\left(Y^*_{i_0}\right)$, and $\rho_i\left(Y_i\right)\leq\rho_i\left(Y^*_i\right)$, for $i\in\mathcal{G}_1\backslash\left\{i_0\right\}$, as well as $\lambda_i^{(1)}>0$, for $i\in\mathcal{G}_1$, then $\sum_{i\in \mathcal{G}_1} \lambda_i^{(1)} \rho_i(Y_i)<\sum_{i\in \mathcal{G}_1} \lambda_i^{(1)} \rho_i(Y^*_i)$, contradicting \eqref{eq:I_Kstar_(1)} as $\mathbf{I}\in\mathcal{I}_n$.

Suppose that $j_0=2,3,\dots,n(\Lambda)$. For $j=1,2,\dots,j_0-1$, since $\mathcal{G}_j\cap D=\emptyset$, $\rho_i\left(Y_i\right)=\rho_i\left(Y^*_i\right)$ for all $i\in\mathcal{G}_j$, and hence $\sum_{i\in\mathcal{G}_j} \lambda_i^{(j)}\rho_i(Y_i)=\sum_{i\in \mathcal{G}_j}\lambda_i^{(j)} \rho_i(Y^*_i)$. Since $\I\in\mathcal{I}_n$ and $\I^*\in K^{*(1)}(\Lambda[\cdot,1])$, we obtain $\I\in K^{*(1)}(\Lambda[\cdot,1])$. In general, suppose that $\I\in  K^{*(j-1)}(\Lambda[\cdot, 1:j-1])$ for some $j=2,3,\dots,j_0-1$. Since $\I\in  K^{*(j-1)}(\Lambda[\cdot, 1:j-1])$ and $\I^*\in K^{*(j)}(\Lambda[\cdot,1:j])$, we obtain $\I\in K^{*(j)}(\Lambda[\cdot,1:j])$. By induction,
$\I \in  K^{*(j_0-1)}(\Lambda[\cdot, 1:j_0-1])\subseteq K^{*(j_0-2)}(\Lambda[\cdot, 1:j_0-2])\subseteq\dots\subseteq K^{*(1)}(\Lambda[\cdot,1])\subseteq\mathcal{I}_n$. Since $\rho_{i_0}\left(Y_{i_0}\right)<\rho_{i_0}\left(Y^*_{i_0}\right)$, and $\rho_i\left(Y_i\right)\leq\rho_i\left(Y^*_i\right)$, for $i\in\mathcal{G}_{j_0}\backslash\left\{i_0\right\}$, as well as $\lambda_i^{(j_0)}>0$, for $i\in\mathcal{G}_{j_0}$, we have $\sum_{i\in \mathcal{G}_{j_0}} \lambda_i^{(j_0)} \rho_i(Y_i)<\sum_{i\in \mathcal{G}_{j_0}} \lambda_i^{(j_0)} \rho_i(Y^*_i)$, contradicting \eqref{eq:I_Kstar_(j)} as $\I \in  K^{*(j_0-1)}(\Lambda[\cdot, 1:j_0-1])$.
Hence, $\mathbf{I}^*$ is Pareto optimal under the full setting.

\end{proof}
Just as Theorem~\ref{2_K^*_PO} generalizes Proposition~\ref{2_min_PO} in the bilateral case, Theorem~\ref{n_K^*_PO} generalizes Proposition~\ref{n_min_PO} in the multilateral setting. Under the full setting, decisions that minimize a strict convex combination of all agents' risks are Pareto optimal, as are decisions obtained through sequential optimization according to the sequential matrices $\Lambda \in \mathcal{A}_n$. Importantly, none of these results require the assumptions $(\dagger)$.

\subsubsection{Convex Combinations Obscure Individual Risks}\label{sec:convex_combination}
This subsection addresses the following question: under the assumptions $(\dagger)$, does every Pareto optimal decision vector in the full setting solve a sequential optimization problem corresponding to some sequential matrix $\Lambda\in \mathcal{A}_n$? An affirmative answer would provide a converse to Theorem~\ref{n_K^*_PO}. Together with that theorem, it would yield a characterization of Pareto optimality in the full setting in terms of sequential optimizations under $(\dagger)$, thereby extending Theorem~\ref{2_PO_K^*} and Corollary~\ref{2_result} from the bilateral case. We first establish the following proposition.
\begin{proposition}\label{lemma:subset_PO}
Let the full setting $(\Nn;\In)$, for $n\in\mathbb{N}$, satisfy $(\dagger)$. Let a group of agents $\mathcal{G}\subseteq\mathcal{N}_n$ with $\vert\mathcal{G}\vert\in\{2,3,\dots,n+1\}$, and a feasible subset $\mathcal{C}\subseteq \mathcal{I}_n$ be convex. If $\mathbf{I}^*= \left(I^*_1,I^*_2,\dots, I_n^*\right)\in\mathcal{C}$ is Pareto optimal among 
$\mathcal{G}$ within 
$\mathcal{C}$, and if there exists $i_0\in\mathcal{G}$ such that $\mathbf{I}^*\in\argmin_{\mathbf{I}\in\mathcal{C}}\rho_{i_0}\left(Y_{i_0}\right)$, then $\mathbf{I}^*$ is Pareto optimal among the agents in the sub-group $\mathcal{G}\backslash\{i_0\}$ within the convex feasible subset $\argmin_{\mathbf{I}\in\mathcal{C}}\rho_{i_0}\left(Y_{i_0}\right)$.
\end{proposition}
\begin{proof}
By Theorem \ref{n_PO_min}, $\mathbf{I}^*\in K\left({\bm\lambda};\mathcal{C}\right)$ for some ${\bm\lambda}\in \Delta_n[0,1]$ with $\nzp(\bm{\lambda})\subseteq\mathcal{G}$. As there exists $i_0\in\mathcal{G}$ such that $\mathbf{I}^*\in\argmin_{\mathbf{I}\in\mathcal{C}}\rho_{i_0}\left(Y_{i_0}\right)$, we may represent this minimization by the weight vector $\mathbf e_{i_0}\in\Delta_n[0,1]$, whose $i_0$-th component equals one and whose other components equal zero. Thus, $\mathbf I^*\in K(\mathbf e_{i_0};\mathcal C)
=
\argmin_{\mathbf I\in\mathcal C}\rho_{i_0}(Y_{i_0})$.

Assume, on the contrary, that 
$\I^*\in K(\mathbf e_{i_0};\mathcal C)$
is not Pareto optimal among 
$\mathcal{G}\backslash\{i_0\}$ within 
the convex feasible subset $K(\mathbf e_{i_0};\mathcal C)$; where the convexity of
$K(\mathbf e_{i_0};\mathcal C)$ follows from Proposition \ref{n_K_lambda_c_convex}. There 
exists 
$\I\in K(\mathbf e_{i_0};\mathcal C)$
such that $\rho_{i}\left(Y_{i}\right)\leq\rho_{i}\left(Y^*_{i}\right)$, for $i\in\mathcal{G}\backslash\{i_0\}$, with at least one inequality being strict. Since 
$\I^*, \I \in K(\mathbf e_{i_0};\mathcal C)$,
we have $\rho_{i_0}\left(Y_{i_0}\right)=\rho_{i_0}\left(Y^*_{i_0}\right)$. Therefore, $\rho_{i}\left(Y_{i}\right)\leq\rho_{i}\left(Y^*_{i}\right)$, for $i\in\mathcal{G}$, with at least one inequality being strict. This contradicts that $\mathbf{I}^*\in\mathcal{C}$ is Pareto optimal among 
$\mathcal{G}$ within 
$\mathcal{C}$, as $\mathbf{I}\in\mathcal{C}$.
\end{proof}
\noindent
Recall Theorem~\ref{n_PO_min} states that a Pareto optimal decision minimizes the convex combination of some agents. Proposition~\ref{lemma:subset_PO} is presented under the assumptions of Theorem~\ref{n_PO_min} and specifies the case when the combination involves only one agent. However, we can prove this result without $(\dagger)$, in which case we do not imply the new feasible set is convex.

Due to Proposition~\ref{lemma:subset_PO}, a Pareto optimal decision for a group of agents that also minimizes the risk of one agent in that group, remains Pareto optimal for the remaining agents once the feasible set is restricted to the minimizers of that agent's risk. 

Proposition~\ref{lemma:subset_PO} also clarifies the proof of Theorem~\ref{2_PO_K^*} in the bilateral setting. Indeed, by Theorem~\ref{2_PO_min} (or Theorem~\ref{n_PO_min}, with $n=1$), any Pareto optimal decision $I^*\in\mathcal{I}_1$ 
among the agents in $\mathcal{N}_1=\{1,2\}$ within the convex feasible set $\mathcal{I}_1$ belongs to $K(\bm{\lambda})$ for some ${\bm\lambda}=\left(\lambda_1,\lambda_2\right)^T\in\left[0,1\right]^{2}$ with $\lambda_1+\lambda_2=1$. More precisely:
\begin{itemize}
\item If $\lambda_1,\lambda_2>0$, then $I^*\in K^*(\Lambda)$ for $\Lambda\in \mathcal{A}_1^{(1)}\subseteq\mathcal{A}_1$;
\item If either $\lambda_1=0$ or $\lambda_2=0$, 
then there exists a unique $i_0\in\mathcal{N}_1=\left\{1,2\right\}$ such that $I^*\in\argmin_{I\in\mathcal{I}_1}\rho_{i_0}\left(Y_{i_0}\right)$. By Proposition~\ref{lemma:subset_PO}, 
$\I^*$
is Pareto optimal among the remaining agent in $\mathcal{N}_1\backslash\{i_0\}$ within the convex feasible subset $\argmin_{I\in\mathcal{I}_1}\rho_{i_0}\left(Y_{i_0}\right)$. Since $\mathcal{N}_1\backslash\{i_0\}$ contains only one agent, Pareto optimality 
is equivalent to minimizing that remaining agent's risk. Therefore, $I^*\in K^*(\Lambda)$ for some $\Lambda\in \mathcal{A}_1^{(2)}\cup\mathcal{A}_1^{(3)}\subseteq\mathcal{A}_1$.
\end{itemize}


In the multilateral setting when $n\in\mathbb{N}\backslash\{1\}$, we would like an analogue result to Proposition~\ref{lemma:subset_PO} that covers all possible cases, so that if a Pareto decision $\I^*\in \mathcal{I}_n$ under assumptions $(\dagger)$ minimizes any possible convex combination of the agent's objectives, then $\I^*$ remains Pareto optimal for the remaining agents within a smaller feasible set, and we would recursively apply Theorem~\ref{n_PO_min}, aiming to prove that $\I^*$ solves a sequential optimization given by some $\Lambda\in \mathcal{A}_n$. In which case, we would obtain a generalization of Theorem~\ref{2_PO_K^*}. In fact, when the Pareto decision $\I^*$ minimizes the convex combination of the risks of all policyholders, we can conclude as desired.

\begin{proposition}\label{prop:PO_minpolicyholders_K^*}
Let the full setting $(\Nn;\In)$, for $n\in \mathbb{N}$, satisfy $(\dagger)$. If $\I^*=(I_1^*,I_2^*, \dots, I_n^*)\in \In$ is Pareto optimal under the full setting, and $\I^*\in \mathcal{C}= \argmin_{\I\in \In} \sum_{i=1}^n \lambda_i \rho_i(Y_i)$, for some $\left(\lambda_1,\lambda_2,\dots,\lambda_n\right)\in\left(0,1\right)^n$ such that $\lambda_1+\lambda_2+\dots+\lambda_n=1$, then $\I^*\in \argmin_{\I\in \mathcal{C}} \rho_{n+1}(Y_{n+1})$.
\end{proposition}
\begin{proof}
Since $\I^*\in \mathcal{C} = K(\bm\lambda)$, for $\bm\lambda = (\lambda_1,\lambda_2, \dots, \lambda_n,0)\in \Delta_n[0,1]$ with the positive weights $\lambda_i$ corresponding to all policyholders, then $\I^*\in \argmin_{\I\in \mathcal{I}_n} \rho_i(Y_i)$, for all $i=1,2,\dots, n$. Assume there exists $\I \in \mathcal{C}$ such that $\rho_{n+1}(Y_{n+1})<\rho_{n+1}(Y_{n+1}^*) $, then, $\I$ would satisfy $\rho_i(Y_i)=\rho_i(Y_i^*)$ for all $i=1,2,\dots, n$ contradicting Pareto optimality of $\I^*$. 
\end{proof}
However, when the central agent is included in the first level of optimization, we cannot show the desired result under $(\dagger)$.
To illustrate this, fix $n=2$, and let $\I^*=(I^*_1,I^*_2)\in \mathcal{I}_2$ be Pareto optimal under the full setting. By the assumptions ($\dagger$) and Theorem \ref{n_PO_min},
\begin{equation*}
\mathbf{I}^* \in K(\bm\lambda)=\argmin_{\I=(I_1,I_2)\in \mathcal{I}_2}\{\lambda_1 \rho_1(Y_1) + \lambda_2\rho_2(Y_2) + \lambda_3\rho_3(Y_3)\},
\end{equation*}
for some ${\bm\lambda}=(\lambda_1, \lambda_2, \lambda_3)^T\in \Delta_2[0,1]$. Imitating the proof of Theorem~\ref{2_PO_K^*}, we consider all possible forms of ${\bm\lambda}\in \Delta_2[0,1]$. Suppose that $\lambda_1,\lambda_3\in(0,1)$ but $\lambda_2=0$. If the converse of Theorem~\ref{n_K^*_PO} were to hold, we would expect that $\mathbf{I}^*\in K^*(\Lambda)$ for some $\Lambda\in \mathcal{A}_2^{(3)}$, as defined in Example~\ref{ex:A2}. Given that
\begin{equation}
\mathbf{I}^*\in K((\lambda_1, 0, \lambda_3)^T)=K((\lambda_1, 0, \lambda_3)^T;\mathcal{I}_2)=K^{\ast(1)} ( (\lambda_1, 0, \lambda_3)^T )=K^{*\left(1\right)}\left(\Lambda[\cdot,1]\right),
\label{eq:3_K_levelone}
\end{equation}
the desired conclusion would be equivalent to
\begin{equation*}
\mathbf{I}^*\in K^{\ast (2)} (\Lambda[\cdot, 1:2])= K((0, 1, 0)^T; K^{*\left(1\right)}\left(\Lambda[\cdot,1]\right)) = \argmin_{\I=(I_1, I_2)\in K^{*\left(1\right)}\left(\Lambda[\cdot,1]\right)}   \rho_2(Y_2);
\end{equation*}
in other words, $\mathbf{I}^*\in K^{*\left(1\right)}\left(\Lambda[\cdot,1]\right)$ is Pareto optimal among the remaining agent in $\{2\}$ within the convex feasible subset $K^{*\left(1\right)}\left(\Lambda[\cdot,1]\right)$.

Assume that $\mathbf{I}^*\notin K^{\ast (2)} (\Lambda[\cdot, 1:2])$, with the aim of deriving a contradiction from the Pareto optimality in the full setting when $n=2$; necessarily, $K^{*\left(1\right)}\left(\Lambda[\cdot,1]\right)$ is not a singleton. There exists a decision vector $\mathbf{I}\in K^{*\left(1\right)}\left(\Lambda[\cdot,1]\right)\backslash\{\mathbf{I}^*\}$ such that $\rho_2(Y_2)<\rho_2(Y^*_2)$. Since $\mathbf{I},\mathbf{I}^*\in K^{*\left(1\right)}\left(\Lambda[\cdot,1]\right)$,
$\lambda_1 \rho_1(Y_1)+\lambda_3\rho_3(Y_3)=\lambda_1 \rho_1(Y^*_1)+\lambda_3\rho_3(Y^*_3)$. Because $\mathbf{I}\neq\mathbf{I}^*$ and $\lambda_1,\lambda_3>0$, as long as $\rho_1(Y_1)\neq\rho_1(Y_1^*)$ and $\rho_3(Y_3)\neq\rho_3(Y_3^*)$, $(\rho_1(Y_1)-\rho_1(Y_1^*))(\rho_3(Y_3)-\rho_3(Y_3^*))<0$ can possibly hold. In this scenario, one of the following two possible cases can happen:
\[\begin{aligned}
 &\rho_1(Y_1) < \rho_1(Y_1^*) ,\,   \rho_2(Y_2) <  \rho_2(Y^*_2),\, \rho_3(Y_3) > \rho_3(Y^*_3) , \text{ or }   \\
 &\rho_1(Y_1) > \rho_1(Y_1^*) ,\,   \rho_2(Y_2) <  \rho_2(Y^*_2),\, \rho_3(Y_3) < \rho_3(Y^*_3) \,.
\end{aligned}\]
Yet, neither of them contradicts Pareto optimality in the full setting when $n=2$. The underlying issue is that the optimization at the first level in \eqref{eq:3_K_levelone} determines only a convex combination of the risks of agents $1$ and $3$, rather than their individual risk values. Consequently, improvements in one agent's risk may be offset by deteriorations in another's without violating optimality with respect to the convex combination.

To further illustrate this issue generally in a Euclidean space, let $S\subseteq\mathbb{R}^3$ be a cube set, and consider the plane set:
\begin{equation*}
P=\left\{\mathbf{y}=(y_1,y_2,y_3)^T\in\mathbb{R}^3: \frac{2}{3} y_1 + \frac{1}{3}y_3 = 2\right\}.
\end{equation*}
In Figure~\ref{fig:counterexample}, this plane supports the cube $S$ along the bottom face ABCD. From the figure, it is clear that Pareto optimal points of the set $S$ are those $\mathbf{y}^*$ lying on the left-bottom edge AB and the bottom edge BC, which belong to the bottom face subset ABCD by definition; an example is $\mathbf{y}^*= (2, 1.41,2)^T$, the mid-point along the bottom edge BC. Let
\begin{equation*}
\mathcal{C} =  \argmin_{\mathbf{y}=(y_1,y_2,y_3)^T\in S} \left\lbrace\frac{2}{3} y_1 + \frac{1}{3} y_3\right\rbrace \, ,
\end{equation*}
which is in fact the bottom face subset ABCD. Therefore, the Pareto optimal $\mathbf{y}^* = (2, 1.41,2)^T\in\mathcal{C}$, but $\mathbf{y}^*\notin\argmin_{\mathbf{y}\in \mathcal{C}} y_2$. From Figure~\ref{fig:counterexample}, there clearly exist other Pareto optimal points $\mathbf{y}\in\mathcal{C}$ such that $y_2<y_2^*$ moving along the bottom edge BC towards the point B.

\begin{figure}[h!]
\centering
\includegraphics[width=0.5\linewidth]{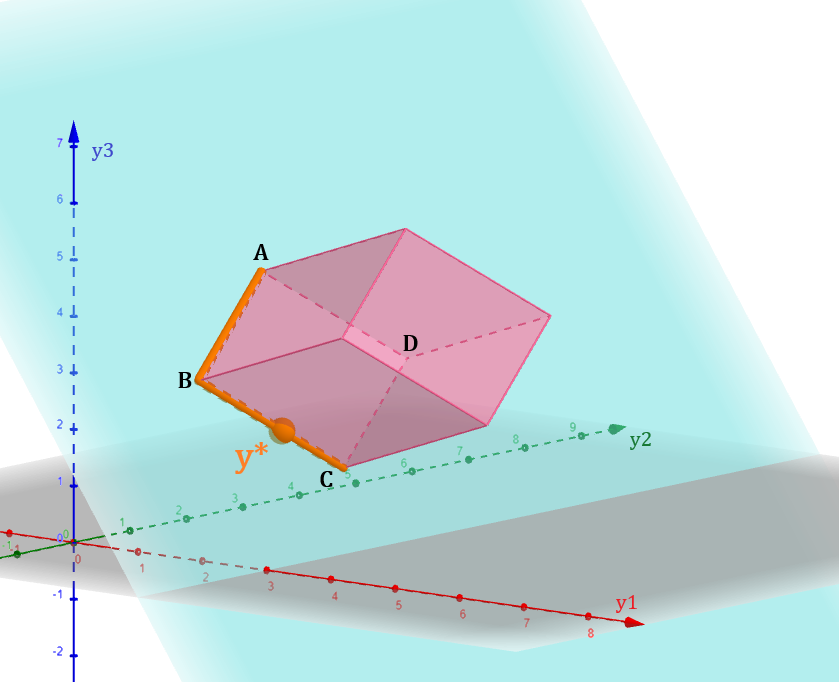}
\caption{\scriptsize{Illustration when the set $S$ is the cube (red), supported by the plane $P$ given by $y_3 = 6-2y_1$ (blue), the bottom face of the cube delimited by the points A, B, C, D, is the set $\mathcal{C}$, and the horizontal plane (grey) is given by $y_3=0$. The thick edges of the cube (orange) denote all Pareto optimal points of $S$, and the bullet point (orange) is a particular Pareto optimal point $\mathbf{y}^*$ of $S$. }}
\label{fig:counterexample}
\end{figure}

To sum up, this subsection shows that the proof of Theorem~\ref{2_PO_K^*} cannot be extended to establish the converse of Theorem~\ref{n_K^*_PO} in the multilateral setting under the assumptions $(\dagger)$. Although the central agent plays a distinct role, in general the core issue is that an aggregation of risks does not determine individual risk components, while Pareto optimality is defined in terms of the individual risks of agents.

\subsubsection{Exclusive and Unfair Sequential Optimization}\label{sec:exclusive_unfair}
Despite not having the desired conclusion that Pareto optimality holds for remaining agents, we investigate the implications of Pareto optimality over a relaxed set of matrices.
To this end, we extend the set of sequential matrices over which we apply the function $K^*$.

For $m\in\mathbb{N}$, let $\mathcal{M}_{n+1,m}$ denote the set of 
matrices of order $(n+1)\times m$; define
\begin{equation*}
\tilde{\mathcal{A}}_{n,m}=\left\{\Lambda\in\mathcal{M}_{n+1,m}:\Lambda[\cdot,j]\in\Delta_{n}[0,1],\text{ for }j=1,2,\dots,m\right\},
\end{equation*}
and let $\tilde{\mathcal{A}}_{n}=\cup_{m=1}^{\infty}\tilde{\mathcal{A}}_{n,m}$. For any $\Lambda\in\tilde{\mathcal{A}}_{n}$, define $K^{*(1)}(\Lambda[\cdot,1])\subseteq\mathcal{I}_n$ as in \eqref{eq:K_star_recursion_1}, and recursively define $K^{*(j)}(\Lambda[\cdot,1:j])\subseteq\mathcal{I}_n$ as in \eqref{eq:K_star_recursion_n}, for $j=2,3,\dots$. Then define the set-valued function $K^*:\tilde{\mathcal{A}}_n \to 2^{\mathcal I_n}$ by $K^*(\Lambda)=K^{*(m)}(\Lambda[\cdot,1:m])$, for $\Lambda\in \tilde{\mathcal{A}}_n$, where $m$ is the number of columns of $\Lambda$.

\begin{proposition}\label{n_PO_new_result}
Let the full setting $(\Nn;\In)$, for $n\in\mathbb{N}$, satisfy $(\dagger)$. If $\mathbf{I}^*=\left(I^*_1,I^*_2,\dots,I_n^*\right)\in \In$ is Pareto optimal under the full setting, then 
there exists $\Lambda\in\tilde{A}_{n}$
such that $\mathbf{I}^*\in K^*(\Lambda)$.

\end{proposition}

\begin{proof}
By Theorem~\ref{n_PO_min}, there exists $\bm\lambda \in \Delta_n[0,1]$ such that $\mathbf{I}^* \in K(\bm\lambda)$. Set $\Lambda=[\bm\lambda]\in\tilde{\mathcal{A}}_{n}$, a one-column matrix. 
Then, by the definition of $K^*$ on $\tilde{\mathcal{A}}_{n}$, $K^*(\Lambda)=K(\bm\lambda)$. Hence, $\mathbf{I}^*\in K^*(\Lambda)$.
\end{proof}

\begin{remark}
To prove Proposition~\ref{n_PO_new_result}, we selected a one-column matrix as a direct consequence of Theorem~\ref{n_PO_min}. Moreover, using Theorem~\ref{n_PO_min} and Proposition~\ref{n_K_lambda_c_convex}, we can construct, for any $m\in\mathbb{N}$,
$\bm\lambda^{(1)},\dots,\bm\lambda^{(m)}\in\Delta_n[0,1]$ recursively such that $\mathbf{I}^*\in K^{*(m)}([\bm\lambda^{(1)};\dots;\bm\lambda^{(m)}])$, where clearly $[\bm\lambda^{(1)};\dots;\bm\lambda^{(m)}]\in\tilde{A}_n $.

Indeed, by Theorem~\ref{n_PO_min}, $\mathbf{I}^* \in K(\bm\lambda^{(1)})$ for some $\bm\lambda^{(1)} \in \Delta_n[0,1]$. Since $\mathbf{I}^*\in\mathcal{I}_n$ is Pareto optimal 
under the full setting $(\Nn, \In)$,
and $\I^*\in K(\bm\lambda^{(1)}) \subseteq \mathcal{I}_n$, Pareto optimality is preserved within the smaller subset, i.e., $\mathbf{I}^* \in K(\bm\lambda^{(1)})$ remains Pareto optimal among 
$\mathcal{N}_n$ within 
$K(\bm\lambda^{(1)})$, which is convex by Proposition~\ref{n_K_lambda_c_convex}.
Applying Theorem~\ref{n_PO_min} again, we have $\mathbf{I}^* \in K(\bm\lambda^{(2)};K(\bm\lambda^{(1)}))$ for some $\bm\lambda^{(2)} \in \Delta_n[0,1]$. Since $K(\bm\lambda^{(2)};K(\bm\lambda^{(1)}))\subseteq K(\bm\lambda^{(1)})$, Pareto optimality of $\mathbf{I}^*$ is preserved within the subset $K(\bm\lambda^{(2)};K(\bm\lambda^{(1)}))$, which is convex by Proposition~\ref{n_K_lambda_c_convex}.
Repeating this argument recursively yields the desired result, where $\Lambda = [\bm\lambda^{(1)};\bm\lambda^{(2)};\dots;\bm\lambda^{(m)}]\in \tilde{A}_n$ for some $m\in\mathbb{N}$.\qed
\end{remark}

Comparing the sequential matrices in $\tilde{\mathcal{A}}_n$ with those in $\mathcal{A}_n$ (see \eqref{A_n}), we observe that, for $\Lambda\in\tilde{\mathcal{A}}_n$, there may exist $i\in\mathcal{N}_n$ such that $\left|\nzp(\Lambda[i,\cdot]^T)\right|\neq 1$, which could be $0$ or strictly larger than $1$. Consequently, the sequential optimizations encoded by the matrices in $\tilde{\mathcal{A}}_n$ may be exclusive, in the sense that some agents' risks are never optimized (similar to weak Pareto optimality), or unfair, in the sense that some agents' risks are optimized multiple times, while others are ignored or optimized fewer times.

As a matter of fact, such unfair sequential optimizations can lead to Pareto optimal decisions in the full setting. For $m\in\mathbb{N}$, define
\begin{align*}
\mathcal{A}'_{n,m}   = \biggr\{  &  \nonumber  \Lambda\in \mathcal{M}_{n+1,m}:\\ 
&(\text{a})\;\Lambda[\cdot, j]\in \Delta_{n}[0,1],   \text{ for } j=1,2,\dots, m , \text{ and,} \\
&(\text{b})\;\text{for all }i\in\mathcal{N}_n,\text{ there exists }j=1,2,\dots,m\text{ such that }j\in\nzp(\Lambda[i,\cdot]^T)\biggr\}\,;\nonumber
\end{align*}
let $\mathcal{A}'_{n}=\bigcup_{m=1}^{\infty}\mathcal{A}'_{n,m}$. Note that, for $\Lambda\in\mathcal{A}'_{n}$ and any $i\in\mathcal{N}_n$, we have $\left|\nzp(\Lambda[i,\cdot]^T)\right|\geq 1$, in contrast to $\Lambda\in\mathcal{A}_n$, where $\left|\nzp(\Lambda[i,\cdot]^T)\right|=1$.

\begin{proposition}\label{n_K^*_PO_new_result}
Let the full setting $(\mathcal{N}_n;\In)$, for $n\in\mathbb{N}$. If $\mathbf{I}^*= \left(I^*_1,I^*_2,\dots, I_n^*\right)\in K^*(\Lambda)$ for some $\Lambda\in \mathcal{A}'_n$, then $\mathbf{I}^*\in\In$ is Pareto optimal under the full setting.
\end{proposition}

\begin{proof}
The proof follows the same earliest-stage argument as in Theorem~\ref{n_K^*_PO}. Indeed, the proof of Theorem~\ref{n_K^*_PO} only uses the fact that each column belongs to $\Delta_n[0,1]$ and that the union of the columns' support covers $\mathcal N_n$; it does not use the disjointness of the supports. These two properties also hold for every $\Lambda\in\mathcal{A}'_n$. Hence the same proof applies.
\end{proof}
\noindent
Therefore, although the sequential optimizations encoded by matrices in $\mathcal{A}'_n$ are inclusive, in the sense that all agents' risks are optimized at least once, they may still be unfair, as some agents' risks are optimized multiple times while others are optimized fewer times. Most importantly, by Proposition~\ref{n_K^*_PO_new_result}, they lead to Pareto optimality in the full setting, extending Theorem~\ref{n_K^*_PO}. Figure~\ref{fig:PO_sequential} shows an illustration of these results, with some abuse of natation we denote by $K^*(\An)$ the set $\cup_{\Lambda\in \An} K^*(\Lambda)$, and similarly for the families $\mathcal{A}_n'$ and $\tilde{\mathcal{A}}_n$.
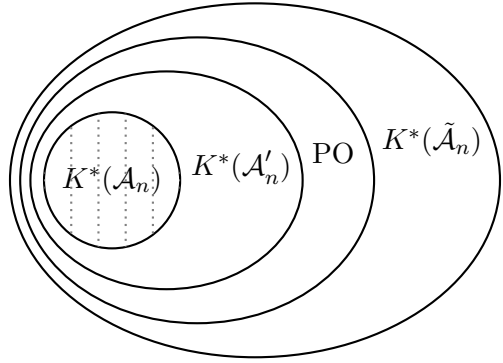
\begin{figure}[h!] 
\centering
\begin{tikzpicture}[thick,fill,scale=0.9]
\draw (-1.8,0) circle (1cm);
\node at (-1.8,0) {$  K^*(\An)$};
\draw (-1,0) ellipse (2 and 1.6); 
\draw (-0.55,0) ellipse (2.6 and 2.1); 
\draw (0.3,0) ellipse (3.6 and 2.6); 
\node at (0.1,0.2) {{${K^*(\mathcal{A}}_n')$}};      
\node at (1.45,0.4) {{$\mathrm{PO}$}};           
\node at (2.9,0.6) {{$K^*(\tilde{\mathcal{A}}_n)$}}; 
\begin{scope} \clip (-1.8,0) circle (1cm); 
\foreach \x in {-2.8,-2.4,-2,-1.6,-1.2,-0.8,-0.4,0,0.4,0.8} { \draw[gray, dotted] (\x,-1) -- (\x,1); } 
\end{scope}
\end{tikzpicture} \caption{\footnotesize{
Illustration of Theorem~\ref{n_K^*_PO}, Proposition~\ref{n_PO_new_result}, and Proposition~\ref{n_K^*_PO_new_result} altogether under $(\dagger)$.}}\label{fig:PO_sequential}
\end{figure}

\subsubsection{Motivation for New Equilibrium Concept}\label{motivation}
The preceding subsections clarify the relationship between Pareto optimality and sequential optimization in the multilateral setting. Section~\ref{sec:convex_combination} shows that the converse of Theorem~\ref{n_K^*_PO} cannot be obtained by directly extending the bilateral argument: optimizing a convex combination of several agents' risks does not determine the corresponding individual risk components. Section~\ref{sec:exclusive_unfair} further shows that, if one enlarges the family of sequential matrices, then Pareto optimality can be related to broader sequential procedures; however, these procedures need not possess a balanced account of the agent's objectives as encoded by the original class $\mathcal A_n$.

The class $\mathcal A_n$ has two structural features that make it especially natural for multilateral sequential optimization. First, it is {\it inclusive}: every agent's risk is incorporated into the sequential procedure. Second, it is {\it fair} in the sense of representation: each agent's risk appears at exactly one level of optimization. Thus, no agent is omitted, and no agent is repeatedly prioritized over others. Moreover, every matrix in $\mathcal A_n$ has at most $n+1$ non-zero columns, so the associated sequential optimization procedure consists of finitely many stages. This finite structure contrasts with the larger families $\tilde{\mathcal A}_n$ and $\mathcal A'_n$, and makes the sequential optimizations induced by $\mathcal A_n$ more transparent and computationally tractable.

These observations suggest that the sequential optimizations encoded by $\mathcal A_n$ should correspond to an equilibrium notion that reflects both group-level aggregation and balanced participation of all agents. Such a notion should also be consistent with the bilateral case, where Pareto optimality is particularly tractable and admits a complete characterization through sequential optimization. Classical Pareto optimality, as defined in Definition~\ref{def:PO}, is formulated in terms of individual risks. By contrast, a sequential optimization in $\mathcal A_n$ may optimize, at a given level, a convex combination of the risks of a group of agents. Therefore, to characterize the decisions generated by $\mathcal A_n$, the equilibrium concept should compare not only individual risks but also suitable aggregate risks of groups of agents, while preserving the property that each agent is included exactly once in the ordered sequence. This motivates the introduction of an equilibrium concept, called \emph{inclusive and fair Pareto optimality}, which we define and study in the next section.

\section{Inclusive and Fair Pareto Optimality}\label{sec:iPO}
To define inclusive and fair Pareto optimality, three preliminary concepts must be introduced first: a group risk measure, Pareto optimality between two groups of agents, and an ordered set partition. To simplify the notation, for any $\I\in\In$, denote 
\begin{equation*}
\Y =\mathbf{Y}(\I) = (Y_1(I_1), Y_2(I_2), \dots , Y_n(I_n), Y_{n+1}(\I)),
\end{equation*}
and similarly, for an optimal $\I^*\in\In$,
$
\Y^*=\mathbf{Y}(\I^*) = (Y_1(I_1^*), Y_2(I_2^*), \dots , Y_n(I_n^*), Y_{n+1}(\I^*))
$.

We begin by defining a group risk measure, which represents the collective risk of the agents within a group.
\begin{definition}[Convex-group risk measure]\label{def:grouprisk}
Let the full setting $(\Nn;\In)$, for $n\in\mathbb{N}$, and let a non-empty group of agents $\mathcal{G}\subseteq\Nn$. For any $\I\in\In$, and for any $\bm\lambda^{\mathcal{G}}\in\Delta_n[0,1]$ with $\nzp(\bm\lambda^\mathcal{G})=\mathcal{G}$, define the convex-group risk measure for the group of agents $\mathcal{G}$, as the strict convex combination of their risk measures:
\begin{equation}\label{eq:grouprisk}
\rho(\Y ;\bm\lambda^{\mathcal{G}})  =    \sum_{i\in \mathcal{G}} \lambda^{\mathcal{G}}_i \rho_i(Y_i)\, . 
\end{equation} \qed
\end{definition}
\noindent
We choose the group risk measure to be given by the convex-group risk measure. Consequently, the concepts and results presented in this paper depend on this choice. Here, the convexity refers to the strict convex combination of individual risk measures. The convex-group risk measure itself need not be a convex functional nor a convex risk measure; this depends on the properties of the individual risk measures of the agents involved.

Next, we define Pareto optimality between two groups of agents.
\begin{definition}[Pareto optimality between two groups with convex-group risks]
Consider the full setting $(\Nn;\In)$, for $n\in\mathbb{N}$. Let two groups of agents $\mathcal{G}_1,\mathcal{G}_2\subseteq\Nn$ such that they are non-empty and disjoint; i.e., $\mathcal{G}_1,\mathcal{G}_2\neq \emptyset$ and $\mathcal{G}_1\cap\mathcal{G}_2 = \emptyset$. Each of these groups is endowed with a convex-group risk measure given in \eqref{eq:grouprisk}, via $\bm\lambda^{\mathcal{G}_1},\bm\lambda^{\mathcal{G}_2}\in\Delta_n[0,1]$, with $\nzp(\bm\lambda^{\mathcal{G}_1})=\mathcal{G}_1$ and $\nzp(\bm\lambda^{\mathcal{G}_2})=\mathcal{G}_2$, respectively. Let a feasible subset $\mathcal{C} \subseteq \mathcal{I}_n$.
%
%

We say a decision $\I^* = \left(I^*_{1},I^*_{2},\dots,I^*_{n}\right)\in \mathcal{C}$ is Pareto optimal between the two groups, $\mathcal{G}_1$ and $\mathcal{G}_2$, within $\mathcal{C}$ if there does not exist another decision $\I = (I_1,I_2, \dots, I_n)\in\mathcal{C}$ such that, $\rho(\Y;\bm\lambda^{\mathcal{G}_1}) \leq \rho(\Y^*; \bm\lambda^{\mathcal{G}_1})$, and $\rho(\Y;\bm\lambda^{\mathcal{G}_2}) \leq \rho(\Y^*; \bm\lambda^{\mathcal{G}_2})$, with at least one inequality being strict.
\qed
\end{definition}

The last building block to introduce the concept of inclusive and fair Pareto optimality is called an ordered set partition.
\begin{definition}
An ordered set partition $\mathcal{P}$ of the set $\mathcal{N}_n$ for the $n+1$ agents is a partition of $\Nn$, given as a sequence, i.e., an ordered list, of non-empty subsets of $\mathcal{N}_n$; that is, $\mathcal{P}=(N_1,N_2, \dots, N_m)$, for some $m=1,2,\dots,n+1$, such that (i) $N_k\neq \emptyset$, for all $k=1,2,\dots, m$, $N_k \cap N_l = \emptyset$, for all $k,l=1,2,\dots, m$ with $k\neq l$, and $\mathcal{N}_n = \cup_{k=1}^m N_k$; (ii) two ordered set partitions $(N_1,N_2, \dots, N_m)$ and $(M_1,M_2, \dots, M_r)$, for some $m,r=1,2,\dots,m+1$, are equal if and only if $m=r$ and $N_k=M_k$ for all $k=1,2,\dots,m$.
\qed
\end{definition}
\noindent
The ordered set partition is a well-known concept in combinatorics. 
Denote by $\mathcal{P}(\mathcal{N}_n)$ the set of all ordered set partitions of $\mathcal{N}_n$. The following example illustrates these.

\begin{example}\label{ex:set_order}
When $n=1$, $\mathcal{N}_1=\{1,2\}$, and its set of all ordered set partitions is
\begin{equation}\label{eq:1_set_order}
\mathcal{P}(\mathcal{N}_1) = \{ \mathcal{P}^{(1)}_1=(\mathcal{N}_1), \mathcal{P}_1^{(2)}=( \{1\}, \{2\}), \mathcal{P}_1^{(3)}= (\{2\}, \{1\}) \}\, . 
\end{equation}
When $n=2$, $\mathcal{N}_2=\{1,2,3\}$, and its set of all ordered set partitions is
\begin{align}\label{eq:2_set_order}
\mathcal{P}(\mathcal{N}_2) = \{& \mathcal{P}_2^{(1)}=(\mathcal{N}_2),\,  \nonumber\mathcal{P}_2^{(2)}=( \{1,2\}, \{3\} ), \mathcal{P}_2^{(3)} = ( \{1,3\}, \{2\} ), \mathcal{P}_2^{(4)} = ( \{2,3\}, \{1\} ), \\ \nonumber
&\mathcal{P}_2^{(5)}=( \{1\}, \{2,3\}), \mathcal{P}_2^{(6)}=( \{2\}, \{1,3\}), \mathcal{P}_2^{(7)}=( \{3\}, \{1,2\} ),\\ 
&\mathcal{P}_2^{(8)}=(\{1\}, \{2\}, \{3\}), \mathcal{P}_2^{(9)}=(\{1\}, \{3\}, \{2\}), \mathcal{P}_2^{(10)}=(\{2\}, \{1\}, \{3\}),\\  \nonumber
& \mathcal{P}_2^{(11)}=(\{2\}, \{3\}, \{1\}), \mathcal{P}_2^{(12)}=(\{3\}, \{1\}, \{2\}), \mathcal{P}_2^{(13)}=(\{3\}, \{2\}, \{1\}) \}\, . 
\qed
\end{align}
\end{example}

Equipped with these three concepts, we now define inclusive and fair Pareto optimality.

\begin{definition}[Inclusive and fair Pareto optimality with convex-group risks]\label{def:ipo}
Consider the full setting $(\Nn;\In)$, for $n\in\mathbb{N}$. We say a decision $\mathbf{I}^* = (I_1^*, I_2^*, \dots, I_n^*)\in \mathcal{I}_n$ is inclusive and fair Pareto optimal, with convex-group risks given in \eqref{eq:grouprisk}, under the full setting if there exist an ordered set partition $\mathcal{P}=(N_1,N_2,\dots , N_m)\in \mathcal{P}(\mathcal{N}_n)$, for some $m=1,2,\dots,n+1$, and vectors $\bm\lambda^{N_k}\in \Delta_n[0,1]$ with $\nzp(\bm\lambda^{N_k})=N_k $, for $k=1,2,\dots,m$,  such that both of the following hold:
\begin{enumerate}
\item[(i)] 
$\I^* \in \mathcal{I}^{(1)} =\argmin_{\I \in \mathcal{I}_n} \rho(\Y; \bm\lambda^{N_1})$; and,
\item[(ii)] if $m=2,3,\dots,n+1$, 
$\I^*\in \mathcal{I}^{(l)} $, recursively for $l=2,3,\dots, m $, 
which is the set of all Pareto optimal decision vectors between two groups, $\mathcal{H}_{l-1} = \cup_{k=1}^{l-1} N_k$ and $N_l$, in $\mathcal{I}^{(l-1)}$, with convex-group risks $\rho(\Y;\bm\lambda^{\mathcal{H}_{l-1}})$ and $\rho(\Y;\bm\lambda^{N_l})$, where 
\begin{equation}
\bm\lambda^{\mathcal{H}_{l-1}} = \frac{1}{l-1}\sum_{k=1}^{l-1} \bm\lambda^{N_k} \in \Delta_n[0,1],
\label{eq:inclusive_group_lambda}
\end{equation}
with $\nzp(\bm\lambda^{\mathcal{H}_{l-1}})=\mathcal{H}_{l-1}=\cup_{k=1}^{l-1} N_k$.
\qed
\end{enumerate}
\end{definition}
\noindent
In the sequel, for $l=1,2,\dots, m$, the set $\mathcal{I}^{(l)}$ shall be denoted by $\mathcal{I}^{(l)}([\bm\lambda^{N_1}; \bm\lambda^{N_2}; \dots; \bm\lambda^{N_l}])$, when emphasizing its dependence on the ordered set partition and the vectors.

The definition of inclusive and fair Pareto optimality depends on the ordered set partition $\mathcal{P}=(N_1,N_2,\dots , N_m)\in \mathcal{P}(\mathcal{N}_n)$. From this, $m$ inclusive growing groups of agents are formed:
\begin{equation*}
\mathcal{H}_1=N_1,\;\mathcal{H}_2=N_1\cup N_2,\;\dots,\;\mathcal{H}_{m-1}=\cup_{k=1}^{m-1}N_k,\;N_m.
\end{equation*}
By the condition (i), the inclusive and fair Pareto optimal decision $\mathbf{I}^*$ minimizes directly the convex-group risk of the first group of agents $\mathcal{H}_1=N_1$. When more inclusive groups of agents are formed, so $m=2,3,\dots,n+1$, by condition (ii), for each level $l=2,3,\dots,m$, an inclusive and fair Pareto optimal decision $\mathbf{I}^*$ is Pareto optimal, between the current $l$-th group of agents $N_l$ and the previous inclusive group of agents $\mathcal{H}_{l-1} = \cup_{k=1}^{l-1} N_k$, among those (Pareto) optimal decisions at the previous level $l-1$. Meaning that, recursively, new groups that join the scheme must preserve Pareto optimality with the collective of agents who joined previously. We note that, the convex-group risk of an inclusive group $\mathcal{H}_{l-1}= \cup_{k=1}^{l-1} N_k$, for $l=2,3,\dots,m$, is defined via the vector $\bm\lambda^{\mathcal{H}_{l-1}}$ in \eqref{eq:inclusive_group_lambda}, which is the average of the vectors corresponding to each group of agents considered in that collective, i.e., we assign equal weights $1/(l-1)$ to the groups $N_1,N_2,\dots,N_{l-1}$. However, the choice of the vector $\bm\lambda^{\mathcal{H}_{l-1}}$ 
can be generalized by assigning any positive weights $w^{(l-1)}_k>0$ to the groups $N_k$, for $k=1,2,\dots,l-1$, so that $\bm\lambda^{\mathcal{H}_{l-1}} = \sum_{k=1}^{l-1} w^{(l-1)}_k\bm\lambda^{N_k}$, so long as $\bm{\lambda}^{\mathcal{H}_{l-1}}\in \Delta_n[0,1]$, with $\nzp(\bm{\lambda}^{\mathcal{H}_{l-1}}) = \mathcal{H}_{l-1}= \cup_{k=1}^{l-1} N_k$; such choice will not affect the results of this paper, and hence the equal weights are taken in the definition for simplicity. Following the recursion, the search for inclusive and fair Pareto optimal decisions shrinks the feasible set through $\In\supseteq\mathcal{I}^{(1)}\supseteq\dots\supseteq\mathcal{I}^{(m-1)}\supseteq\mathcal{I}^{(m)}.$

The following proposition shows that inclusive and fair Pareto optimality imply Pareto optimality for the multilateral centralized risk sharing problem.
\begin{proposition}\label{n_iPO_PO}
Let the full setting $(\mathcal{N}_n;\In)$, for $n\in\mathbb{N}$. If $\I^*=(I_1^*, I_2^*, \dots, I_n^*)\in\mathcal{I}_n$ is inclusive and fair Pareto optimal, with convex-group risks given in \eqref{eq:grouprisk}, under the full setting, then $\I^*$ is Pareto optimal under the full setting.
\end{proposition}
\begin{proof}
Assume, on the contrary, that $\mathbf{I}^*$ is not Pareto optimal under the full setting; that is, there exists $\mathbf{I} = \left(I_1,I_2,\dots,I_n\right)\in\mathcal{I}_n$ such that, $\rho_i\left(Y_i\right)\leq\rho_i\left(Y^*_i\right)$, for any $i\in\mathcal{N}_n$, with at least one inequality being strict. Without loss of generality, let $i_0\in\mathcal{N}_n$ such that $\rho_{i_0}\left(Y_{i_0}\right)<\rho_{i_0}\left(Y^*_{i_0}\right)$. By the definition of inclusive and fair Pareto optimality, there exists a unique $k_0=1,2,\dots,m$, such that $i_0\in N_{k_0}$ while $i_0\notin\cup_{k=1;k\neq k_0}^{m}N_{k}$.

Suppose that $k_0=1$. Since $\rho_{i_0}\left(Y_{i_0}\right)<\rho_{i_0}\left(Y^*_{i_0}\right)$, and $\rho_i\left(Y_i\right)\leq\rho_i\left(Y^*_i\right)$, for $i\in N_1\backslash\left\{i_0\right\}$, as well as $\lambda_i^{N_1}>0$, then $\rho(\Y; \bm\lambda^{N_1})=\sum_{i\in N_1} \lambda_i^{N_1} \rho_i(Y_i)<\sum_{i\in N_1} \lambda_i^{N_1} \rho_i(Y^*_i)=\rho(\Y^*; \bm\lambda^{N_1})$ for $i\in N_1$, and thus $\mathbf{I}^*\notin\mathcal{I}^{(1)}$ as $\mathbf{I}\in\mathcal{I}_n$, contradicting the condition (i) of Definition \ref{def:ipo}.

Suppose that $k_0=2,3,\dots,m$. For $k=1,2,\dots,k_0-1$, since $\rho_i\left(Y_i\right)\leq\rho_i\left(Y^*_i\right)$ and $\lambda_i^{N_k}>0$, then $\rho(\Y; \bm\lambda^{N_k})=\sum_{i\in N_k} \lambda_i^{N_k} \rho_i(Y_i)\leq\sum_{i\in N_k} \lambda_i^{N_k} \rho_i(Y^*_i)=\rho(\Y^*; \bm\lambda^{N_k})$ for $i\in N_k$. For $k=1$, since $\mathbf{I}\in\mathcal{I}_n$, by the condition (i) of Definition \ref{def:ipo}, $\rho(\Y^*; \bm\lambda^{N_1})\leq\rho(\Y; \bm\lambda^{N_1})$, and thus $\mathbf{I}^*,\mathbf{I}\in\mathcal{I}^{(1)}$. Recursively, for $k=2,3,\dots,k_0-1$, assume that, although $\mathbf{I}\in\mathcal{I}^{(k-1)}\subseteq\mathcal{I}^{(k-2)}\subseteq\dots\subseteq\mathcal{I}^{(1)}$, on the contrary, $\mathbf{I}\notin\mathcal{I}^{(k)}$; that is, $\mathbf{I}$ is not Pareto optimal between two groups, $\mathcal{H}_{k-1} = \cup_{r=1}^{k-1} N_r$ and $N_k$, in $\mathcal{I}^{(k-1)}$, with convex group risks $\rho(\Y;\bm\lambda^{\mathcal{H}_{k-1}})$ and $\rho(\Y;\bm\lambda^{N_k})$, and hence there exists $\tilde{\mathbf{I}}\in\mathcal{I}^{(k-1)}$ such that $\rho(\tilde{\Y};\bm\lambda^{\mathcal{H}_{k-1}})\leq\rho(\Y;\bm\lambda^{\mathcal{H}_{k-1}})$ and $\rho(\tilde{\Y};\bm\lambda^{N_k})\leq\rho(\Y;\bm\lambda^{N_k})$, with at least one inequality being strict, where $\tilde{\Y} =\mathbf{Y}(\tilde{\I}) = (Y_1(\tilde{I_1}), Y_2(\tilde{I_2}), \dots , Y_n(\tilde{I_n}), Y_{n+1}(\tilde{\I}))$. Since
\begin{equation}
\begin{aligned}
\rho(\Y;\bm\lambda^{\mathcal{H}_{k-1}})=&\;\sum_{i\in \mathcal{H}_{k-1}}\lambda_i^{\mathcal{H}_{k-1}} \rho_i(Y_i)=\sum_{r=1}^{k-1}\sum_{i\in N_r}\lambda_i^{\mathcal{H}_{k-1}}\rho_i(Y_i)=\sum_{r=1}^{k-1}\frac{1}{k-1}\sum_{i\in N_r}\lambda_i^{N_r}\rho_i(Y_i)\\=&\;\sum_{r=1}^{k-1}\frac{1}{k-1}\rho(\Y; \bm\lambda^{N_r})\leq\sum_{r=1}^{k-1}\frac{1}{k-1}\rho(\Y^*; \bm\lambda^{N_r})=\sum_{r=1}^{k-1}\frac{1}{k-1}\sum_{i\in N_r}\lambda_i^{N_r}\rho_i(Y^*_i)\\=&\;\rho(\Y^*;\bm\lambda^{\mathcal{H}_{k-1}}),
\end{aligned}
\label{eq:aggregate_inequality}
\end{equation}
$\rho(\tilde{\Y};\bm\lambda^{\mathcal{H}_{k-1}})\leq\rho(\Y^*;\bm\lambda^{\mathcal{H}_{k-1}})$ and $\rho(\tilde{\Y};\bm\lambda^{N_k})\leq\rho(\Y^*;\bm\lambda^{N_k})$, with at least one inequality being strict, and thus $\mathbf{I}^*\notin\mathcal{I}^{(k)}$ as $\tilde{\mathbf{I}}\in\mathcal{I}^{(k-1)}$, contradicting the condition (ii) of Definition \ref{def:ipo}. Hence, $\mathbf{I}\in\mathcal{I}^{(k)}\subseteq\mathcal{I}^{(k-1)}\subseteq\dots\subseteq\mathcal{I}^{(1)}$, for $k=2,3,\dots,k_0-1$, recursively, and consequently, $\mathbf{I}\in\mathcal{I}^{(k_0-1)}\subseteq\mathcal{I}^{(k_0-2)}\subseteq\dots\subseteq\mathcal{I}^{(1)}\subseteq\mathcal{I}_n$. Since $\rho_{i_0}\left(Y_{i_0}\right)<\rho_{i_0}\left(Y^*_{i_0}\right)$, and $\rho_i\left(Y_i\right)\leq\rho_i\left(Y^*_i\right)$, for $i\in N_{k_0}\backslash\left\{i_0\right\}$, as well as $\lambda_i^{N_{k_0}}>0$, then $\rho(\Y; \bm\lambda^{N_{k_0}})=\sum_{i\in N_{k_0}} \lambda_i^{N_{k_0}} \rho_i(Y_i)<\sum_{i\in N_{k_0}} \lambda_i^{N_{k_0}} \rho_i(Y^*_i)=\rho(\Y^*; \bm\lambda^{N_{k_0}})$ for $i\in N_{k_0}$. By similar arguments as in \eqref{eq:aggregate_inequality}, $\rho(\Y;\bm\lambda^{\mathcal{H}_{k_0-1}})\leq\rho(\Y^*;\bm\lambda^{\mathcal{H}_{k_0-1}})$. Therefore, $\mathbf{I}^*\notin\mathcal{I}^{(k_0)}$ as $\mathbf{I}\in\mathcal{I}^{(k_0-1)}$, contradicting the condition (ii) of Definition \ref{def:ipo}.

Hence, $\mathbf{I}^*$ is Pareto optimal under the full setting.
\end{proof}


Denote by $\mathrm{iFPO}$ the set of inclusive and fair Pareto optimal decisions for the multilateral centralized risk sharing problem. Proposition~\ref{n_iPO_PO} states that $\mathrm{iFPO}\subseteq \mathrm{PO}$. Recall that, Theorem~\ref{n_K^*_PO} also states that $\cup_{\Lambda\in \An} K^*(\Lambda) \subseteq \mathrm{PO}$. Observing that the proofs of Theorem~\ref{n_K^*_PO} and Proposition~\ref{n_iPO_PO} are similar, we shall show the main result of this paper that $\mathrm{iFPO}=\cup_{\Lambda\in \An} K^*(\Lambda)$. It means that, for the multilateral centralized risk sharing problem, the inclusive and fair Pareto optimality is equivalent to the sequential optimizations. This shall be shown by the following two results, namely, Theorem~\ref{n_iPO_K^*} and Theorem~\ref{n_K^*_iPO}. The first theorem states that an inclusive and fair Pareto optimal decision vector solves a sequential optimization, while the second theorem states the converse that, a decision vector solving a sequential optimization is inclusive and fair Pareto optimal. Our proofs shall reveal the pairing, between an ordered set partition $\mathcal{P}\in \mathcal{P}(\Nn)$, and an equivalence class of sequential matrices $\Lambda$ in $\mathcal{A}_n$.

\begin{theorem}\label{n_iPO_K^*}
Consider the full setting $(\mathcal{N}_n;\In)$, for $n\in\mathbb{N}$. If $\I^*=(I_1^*, I_2^*, \dots, I_n^*)\in\mathcal{I}_n$ is inclusive and fair Pareto optimal, with convex-group risks given in \eqref{eq:grouprisk}, under the full setting, then $\mathbf{I}^*\in K^*(\Lambda)$,
for some $\Lambda\in \mathcal{A}_n$, where $K^*$ is given in \eqref{eq:function_K_star_n}.
\end{theorem}
\begin{proof}
By the definition of inclusive and fair Pareto optimality, there exist, an ordered set partition $\mathcal{P}=(N_1,N_2,\dots , N_m)\in \mathcal{P}(\mathcal{N}_n)$, for some $m=1,2,\dots,n+1$, and vectors $\bm\lambda^{N_k}\in \Delta_n[0,1]$, for $k=1,2,\dots,m$, with $\nzp(\bm\lambda^{N_k})=N_k $, such that the conditions (i) and (ii) of Definition \ref{def:ipo} hold.

Define a square matrix, of dimension $n+1$, $\Lambda = [\bm\lambda^{N_1};\bm\lambda^{N_2}; \dots; \bm\lambda^{N_m}; \mathbf{0}; \dots ; \mathbf{0}]$. By construction, $\Lambda\in \mathcal{M}_{n+1}([0,1];0)$; this is because $\bm\lambda^{N_k}\in \Delta_n[0,1]$, and $\nzp(\bm\lambda^{N_k})=N_k\neq\emptyset$, for $k=1,2,\dots,m$. Therefore, $n(\Lambda)=m$. For $j=1,2,\dots,n(\Lambda)$, $\Lambda[\cdot, j]=\bm\lambda^{N_j}\in \Delta_n[0,1]$. Fix an $i\in\mathcal{N}_n$. On the one hand, if $\nzp(\Lambda[i,\cdot]^T)=\emptyset$, then $i\notin N_j$, for all $j=1,2,\dots,n(\Lambda)$, which contradicts that $i\in\Nn=\cup_{j=1}^{n(\Lambda)}N_j$; on the other hand, if $\nzp(\Lambda[i,\cdot]^T)$ contains two distinct indices $j$ and $l$, then $i\in N_j\cap N_l\neq\emptyset$, which contradicts that $N_j\cap N_l=\emptyset$. Hence, for any $i=1,2,\dots,n+1$, there exists a unique $j=1,2,\dots,n(\Lambda)$ such that $\nzp(\Lambda[i,\cdot]^T) =\{j\}$. This proves that $\Lambda\in\mathcal{A}_n$.

We first claim that, for all $l=1,2,\dots,m$,
\begin{equation}
\mathcal{I}^{(l)}(\Lambda[\cdot, 1:l])=K^{*(l)}(\Lambda[\cdot, 1:l]).
\label{eq:mathcal_I_K^*_equal}
\end{equation}
Since $\I^*$ is inclusive and fair Pareto optimal in the full setting $(\Nn; \In)$, $\I^*\in\mathcal{I}^{(m)}(\Lambda[\cdot, 1:m])\subseteq\mathcal{I}^{(m-1)}(\Lambda[\cdot, 1:m-1])\subseteq\dots\subseteq\mathcal{I}^{(1)}(\Lambda[\cdot,1])\subseteq\mathcal{I}_n$. By \eqref{eq:mathcal_I_K^*_equal}, $\I^*\in K^{*(m)}(\Lambda[\cdot, 1:m])\subseteq K^{*(m-1)}(\Lambda[\cdot, 1:m-1])\subseteq\dots\subseteq K^{*(1)}(\Lambda[\cdot,1])\subseteq\mathcal{I}_n$. Because $m=n(\Lambda)$, by \eqref{eq:function_K_star_n}, $\I^*\in K^{*\left(n\left(\Lambda\right)\right)}\left(\Lambda[\cdot,1:n\left(\Lambda\right)]\right)=K^*\left(\Lambda\right)$, as desired.

In the following, we shall prove that the equalities in \eqref{eq:mathcal_I_K^*_equal} hold, for all $l=1,2,\dots,m$. For $l=1$, the equality clearly holds:
\begin{equation*}
\mathcal{I}^{(1)}(\Lambda[\cdot,1])=\argmin_{\I \in \mathcal{I}_n} \rho(\Y; \bm\lambda^{N_1})=\argmin_{\I \in \mathcal{I}_n}\sum_{i\in N_1} \lambda^{N_1}_i \rho_i(Y_i)=K(\bm\lambda ^{N_1})=K^{*\left(1\right)}\left(\Lambda[\cdot,1]\right).
\end{equation*}

Fix a $j=2,3,\dots,m$. Suppose that \eqref{eq:mathcal_I_K^*_equal} hold(s), for all $l=1,2,\dots,j-1$. Assume, on the contrary, that there exists $\mathbf{I}\in \mathcal{I}^{(j)}(\Lambda[\cdot, 1:j])$ such that $\mathbf{I}\notin K^{*(j)}(\Lambda[\cdot, 1:j])$. Since $\mathbf{I}\in \mathcal{I}^{(j)}(\Lambda[\cdot, 1:j])$, the inclusions hold that $\mathbf{I}\in \mathcal{I}^{(j)}(\Lambda[\cdot, 1:j])\subseteq\mathcal{I}^{(j-1)}(\Lambda[\cdot, 1:j-1])=K^{*(j-1)}(\Lambda[\cdot, 1:j-1])\subseteq\dots\subseteq\mathcal{I}^{(1)}(\Lambda[\cdot,1])=K^{*(1)}(\Lambda[\cdot,1])\subseteq\mathcal{I}_n$. However, as $\mathbf{I}\notin K^{*(j)}(\Lambda[\cdot, 1:j])$, there exists $\J\in K^{*(j-1)}(\Lambda[\cdot, 1:j-1])$ such that $\rho(\Y(\mathbf{J}) ; \bm\lambda^{N_j})<\rho(\Y(\mathbf{I}) ; \bm\lambda^{N_j})$. Moreover, the inclusions hold that $\mathbf{J}\in K^{*(j-1)}(\Lambda[\cdot, 1:j-1])\subseteq K^{*(j-2)}(\Lambda[\cdot, 1:j-2])\subseteq\dots\subseteq K^{*(1)}(\Lambda[\cdot,1])\subseteq\mathcal{I}_n$. Therefore, for all $l=1,2,\dots,j-1$, $\rho(\Y(\mathbf{J}) ; \bm\lambda^{N_l})=\rho(\Y(\mathbf{I}) ; \bm\lambda^{N_l})$, which implies that
\begin{equation}
\rho(\Y(\J); \bm\lambda^{\mathcal{H}_{j-1}})=\frac{1}{j-1} \sum_{l=1}^{j-1} \rho(\Y(\J); \bm\lambda^{N_l})= \frac{1}{j-1} \sum_{l=1}^{j-1} \rho(\Y(\I); \bm\lambda^{N_l})=\rho(\Y(\I); \bm\lambda^{\mathcal{H}_{j-1}}).
\label{eq:I_J_G_{l-1}_risk_same}
\end{equation}
This contradicts the fact that $\mathbf{I}\in \mathcal{I}^{(j)}(\Lambda[\cdot, 1:j])$, is Pareto optimal between $\mathcal{H}_{j-1}$ and $N_j$, as $\J\in K^{*(j-1)}(\Lambda[\cdot, 1:j-1])=\mathcal{I}^{(j-1)}(\Lambda[\cdot, 1:j-1])$. Hence, $\mathcal{I}^{(j)}(\Lambda[\cdot, 1:j])\subseteq K^{*(j)}(\Lambda[\cdot, 1:j])$.

Assume, on the contrary, that there exists $\mathbf{I}\in K^{*(j)}(\Lambda[\cdot, 1:j])$ such that $\mathbf{I}\notin \mathcal{I}^{(j)}(\Lambda[\cdot, 1:j])$. Since $\mathbf{I}\in K^{*(j)}(\Lambda[\cdot, 1:j])$, the inclusions hold that $\mathbf{I}\in K^{*(j)}(\Lambda[\cdot, 1:j])\subseteq K^{*(j-1)}(\Lambda[\cdot, 1:j-1])=\mathcal{I}^{(j-1)}(\Lambda[\cdot, 1:j-1])\subseteq\dots\subseteq K^{*(1)}(\Lambda[\cdot,1])=\mathcal{I}^{(1)}(\Lambda[\cdot,1])\subseteq\mathcal{I}_n$. As $\mathbf{I}\notin \mathcal{I}^{(j)}(\Lambda[\cdot, 1:j])$, there exists $\J\in\mathcal{I}^{(j-1)}(\Lambda[\cdot, 1:j-1])$ such that $\rho(\Y(\J); \bm\lambda^{\mathcal{H}_{j-1}})\leq\rho(\Y(\I); \bm\lambda^{\mathcal{H}_{j-1}})$ and $\rho(\Y(\mathbf{J}) ; \bm\lambda^{N_j})\leq\rho(\Y(\mathbf{I}) ; \bm\lambda^{N_j})$, with at least one inequality being strict. Moreover, the inclusions hold that $\mathbf{J}\in\mathcal{I}^{(j-1)}(\Lambda[\cdot, 1:j-1])=K^{*(j-1)}(\Lambda[\cdot, 1:j-1])\subseteq K^{*(j-2)}(\Lambda[\cdot, 1:j-2])\subseteq\dots\subseteq K^{*(1)}(\Lambda[\cdot,1])\subseteq\mathcal{I}_n$. Therefore, for all $l=1,2,\dots,j-1$, $\rho(\Y(\mathbf{J}) ; \bm\lambda^{N_l})=\rho(\Y(\mathbf{I}) ; \bm\lambda^{N_l})$, implying that \eqref{eq:I_J_G_{l-1}_risk_same} holds, and thus, $\rho(\Y(\mathbf{J}) ; \bm\lambda^{N_j})<\rho(\Y(\mathbf{I}) ; \bm\lambda^{N_j})$, which contradicts that $\mathbf{I}\in K^{*(j)}(\Lambda[\cdot, 1:j])$, as $\mathbf{J}\in\mathcal{I}^{(j-1)}(\Lambda[\cdot, 1:j-1])=K^{*(j-1)}(\Lambda[\cdot, 1:j-1])$. Hence, $K^{*(j)}(\Lambda[\cdot, 1:j])\subseteq\mathcal{I}^{(j)}(\Lambda[\cdot, 1:j])$.

To conclude, for any $j=2,3,\dots,m$, if \eqref{eq:mathcal_I_K^*_equal} hold(s), for all $l=1,2,\dots,j-1$, then the equality in \eqref{eq:mathcal_I_K^*_equal} also holds for $l=j$. Therefore, given that the equality in \eqref{eq:mathcal_I_K^*_equal} holds for $l=1$, recursively, for $l=2,3,\dots,m$, \eqref{eq:mathcal_I_K^*_equal} also holds.
\end{proof}

\begin{theorem}\label{n_K^*_iPO}
Let the full setting $(\mathcal{N}_n;\In)$, for $n\in\mathbb{N}$. If $\I^*=(I_1^*, I_2^*, \dots, I_n^*)\in K^*(\Lambda)$, where 
$K^*$ is given in \eqref{eq:function_K_star_n}, for some $\Lambda\in \mathcal{A}_n$, then $\mathbf{I}^*\in\In$ is inclusive and fair Pareto optimal, with convex-group risks given in \eqref{eq:grouprisk}, under the full setting.
\end{theorem}
\begin{proof}
By definition, there exists $\Lambda = [\bm\lambda^{(1)}; \bm\lambda^{(2)};\dots ; \bm\lambda^{(n(\Lambda))}; \mathbf{0}; \dots; \mathbf{0} ]\in\mathcal{M}_{n+1} ([0,1]; 0)$, for some $n(\Lambda)=1,2,\dots,n+1$, such that, for $j=1,2,\dots, n(\Lambda)$, $\Lambda[\cdot, j]=\bm\lambda^{(j)} \in \Delta_{n}[0,1]$; for $i\in\mathcal{N}_n$, there exists a unique $j=1,2,\dots, n(\Lambda)$, such that $\nzp(\Lambda[i,\cdot]^T)  = \{j\}$; and, $\I^*\in K^*(\Lambda)=K^{*\left(n\left(\Lambda\right)\right)}\left(\Lambda[\cdot,1:n\left(\Lambda\right)]\right)\subseteq K^{*\left(n\left(\Lambda\right)-1\right)}\left(\Lambda[\cdot,1:n\left(\Lambda\right)-1]\right)\subseteq\dots\subseteq K^{*(1)}(\Lambda[\cdot,1])\subseteq\mathcal{I}_n$. Define subset(s), $N_j=\nzp(\Lambda[\cdot, j])=\nzp(\bm\lambda^{(j)})\subseteq\mathcal{N}_n$, for $j=1,2,\dots, n(\Lambda)$. By Remark~\ref{rem:Anpartition}, the subset(s) $N_j$, for $j=1,2,\dots, n(\Lambda)$, define a partition of $\Nn$. Following the same ordering as the columns of $\Lambda$, define an ordered set partition $\mathcal{P}=(N_1,N_2, \dots,N_{m})\in \mathcal{P}(\mathcal{N}_n)$, with $m=n(\Lambda)$.

By exactly the same arguments as in the proof of Theorem \ref{n_iPO_K^*}, for all $l=1,2,\dots,m$,
\begin{equation*}
K^{*(l)}(\Lambda[\cdot, 1:l])=\mathcal{I}^{(l)}(\Lambda[\cdot, 1:l]).
\end{equation*}
Since $\I^*\in K^{*\left(n\left(\Lambda\right)\right)}\left(\Lambda[\cdot,1:n\left(\Lambda\right)]\right)\subseteq K^{*\left(n\left(\Lambda\right)-1\right)}\left(\Lambda[\cdot,1:n\left(\Lambda\right)-1]\right)\subseteq\dots\subseteq K^{*(1)}(\Lambda[\cdot,1])\subseteq\mathcal{I}_n$, and $n(\Lambda)=m$, it follows that $\I^*\in\mathcal{I}^{(m)}(\Lambda[\cdot, 1:m])\subseteq\mathcal{I}^{(m-1)}(\Lambda[\cdot, 1:m-1])\subseteq\dots\subseteq\mathcal{I}^{(1)}(\Lambda[\cdot,1])\subseteq\mathcal{I}_n$, and hence $\I^*$ is inclusive and fair Pareto optimal, as desired.
\end{proof}

By Theorem~\ref{n_iPO_K^*} and Theorem~\ref{n_K^*_iPO}, the following corollary states the characterization of the inclusive and fair Pareto optimality and the sequential optimizations. We also highlight the equivalence of the two sets in \eqref{eq:mathcal_I_K^*_equal}, which is critical for the characterization.
\begin{corollary}\label{n_iPO=PO}
Consider the full setting $(\mathcal{N}_n;\In)$, for $n\in\mathbb{N}$. Then, $\I^*=(I_1^*, I_2^*, \dots, I_n^*)\in\mathcal{I}_n$ is inclusive and fair Pareto optimal, with convex-group risks given in \eqref{eq:grouprisk}, under the full setting, if and only if $\mathbf{I}^*\in K^*(\Lambda)$, for some $\Lambda\in \mathcal{A}_n$,
where $K^*$ is given as in \eqref{eq:function_K_star_n}. Moreover, for any such $\I^*$, we have that $\I^*\in\mathcal{I}^{(l)}(\Lambda[\cdot, 1:l])=K^{*(l)}(\Lambda[\cdot, 1:l])$, for all $l=1,2,\dots,m$, with $m=n(\Lambda)$.
\end{corollary}

We note several remarks for Theorem~\ref{n_iPO_K^*}, Theorem~\ref{n_K^*_iPO}, and Corollary~\ref{n_iPO=PO}. First, 
Theorem~\ref{n_iPO_K^*} shows that an inclusive and fair Pareto optimal decision vector solves an associated sequential optimization problem, extending Theorem~\ref{2_PO_K^*} from the bilateral to the multilateral setting in an inclusive and fair framework.

A further difference is that, unlike Theorem~\ref{2_PO_K^*}, Theorem~\ref{n_iPO_K^*} does not rely on the assumptions ($\dagger$). This is a consequence of our choice of the (convex-)group risk measures in Definition~\ref{def:grouprisk},
%
which are adopted to define the inclusive and fair Pareto optimality 
in Definition~\ref{def:ipo}. This modelling choice does not require within-group Pareto equilibria. Thus, dropping the assumptions ($\dagger$) is achieved at the expense of restricting attention to the convex-group risk measures. Importantly, this restriction is also related to the scalarization result in Theorem~\ref{n_PO_min}: under $(\dagger)$, Pareto optimality 
solves the convex-group risk for the agents in a certain group,
and thus, our characterization of inclusive Pareto optimality does not require $(\dagger)$ since condition $(i)$ in Definition~\ref{def:ipo} directly implies this step.
Furthermore, requiring  $\nzp(\lambda_{\mathcal{G}})=\mathcal{G}$ in Definition~\ref{def:ipo} is 
an additional modelling choice, imposed to ensure that all agents in the group are explicitly represented. 
Extending the concept of inclusive and fair Pareto optimality to allow for within-group equilibria shall be an interesting direction for future research, such as within-group Geoffrion-proper Pareto optimality under $(\dagger)$.

Next, Theorem~\ref{n_K^*_iPO} strengthens Theorem~\ref{n_K^*_PO}; notably, neither result requires the assumptions~($\dagger$). Theorem~\ref{n_K^*_iPO} establishes that any decision vector that solves a sequential optimization problem is not just Pareto optimal (by Theorem~\ref{n_K^*_PO}) but, in fact, is inclusive and fair Pareto optimal, which in turn implies canonical Pareto optimality. Moreover, Theorem~\ref{n_K^*_iPO} combined with Proposition~\ref{n_iPO_PO} yields Theorem~\ref{n_K^*_PO}.

Finally, Corollary~\ref{n_iPO=PO} encapsulates the main results of this paper. The proofs of Theorems~\ref{n_iPO_K^*} and~\ref{n_K^*_iPO} establish a pairing between (i) an ordered set partition $\mathcal{P}\in\mathcal{P}(\Nn)$ and (ii) an equivalence class of sequential matrices $\Lambda$ in $\mathcal{A}_n$. Consequently, the number of the ordered set partitions of $\Nn$, denoted by $\lvert\mathcal{P}(\Nn)\rvert$, coincides with the number of the distinct equivalence classes in $\mathcal{A}_n$, denoted by $\mathrm{d}(\mathcal{A}_n)$ (recall that from Section~\ref{sec:Multilateral}); that is, $\vert\mathcal{P}(\mathcal{N}_n)\vert= \mathrm{d}(\mathcal{A}_n)$. Moreover, these pairings induce a finer matching between an ordered set partition together with the $m$ vectors in Definition~\ref{def:ipo} and a representative matrix $\Lambda\in\mathcal{A}_n$ within its class. This yields the 
key equivalence:  
$\mathcal{I}^{(l)}(\Lambda[\cdot, 1:l])=K^{*(l)}(\Lambda[\cdot, 1:l])$, for all $l=1,2,\dots,n(\Lambda)$, with $n(\Lambda)=m$.

To summarize, 
$\mathrm{iFPO} = \cup_{\Lambda\in \An} K^*(\Lambda)\subseteq\mathrm{PO}$, as depicted in Figure~\ref{fig:n_iPO=K^*_PO}. Let $\An^{(j)}$, for $j=1,2,\dots,\mathrm{d}(\An)$, denote the $j$-th distinct equivalence classes in $\mathcal{A}_n$. Then, one can write $\cup_{\Lambda\in \An} K^*(\Lambda) = K^*(\An^{(1)}) \cup K^*(\An^{(2)}) \cup \cdots \cup K^*(\An^{(\mathrm{d}(\An))})$, where $K^*(\An^{(j)}) = \cup_{\Lambda\in \An^{(j)}} K^*(\Lambda)$. By Theorem~\ref{n_K^*_PO}, these $\mathrm{d}(\An)$ classes of the sequentially optimal decision vectors are Pareto optimal; this is illustrated in Figure~\ref{fig:n_iPO_circles}. As noted for the bilateral case, the figure depicts each set separately, but they need not be disjoint. Comparing Figure~\ref{fig:n_iPO=K^*_PO} and Figure~\ref{fig:n_iPO_circles}, the introduction of the equilibrium concept, the inclusive and fair Pareto optimality, unifies all of these sequentially optimal decision vectors which are Pareto optimal.

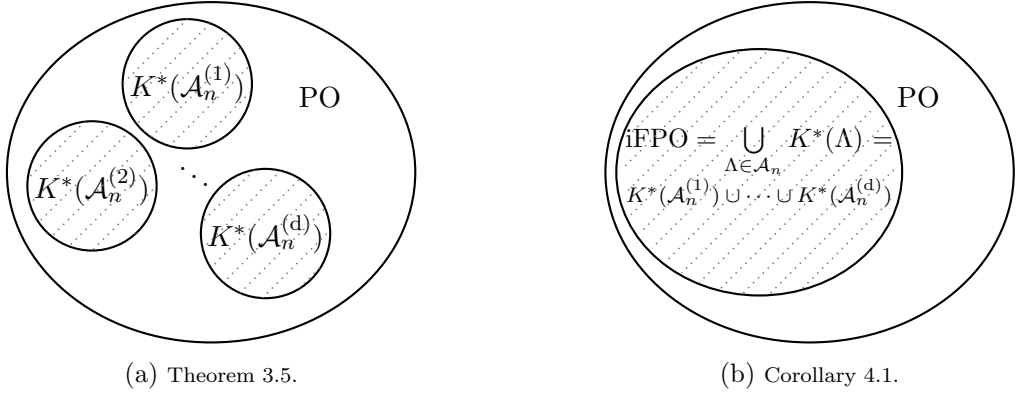
\begin{figure}[h!] 
\begin{subfigure}[C]{0.45\linewidth}
\centering
\begin{tikzpicture}[thick,fill,scale=0.9]
\draw (-0.35,1.3) circle (0.95cm);
\draw (-1.75,-0.2) circle (0.95cm) ;
\draw (0.8,-0.9) circle (0.95cm)  ;
\draw (0,0) ellipse (3 and 2.5);
\node at (1.6,1.1) {$\mathrm{PO}$};
\node at (-0.25,0.05) {$\ddots$};
\node at (-0.35,1.3) {$K^*(\An^{(1)})$};
\node at (-1.75,-0.2) {$K^*(\An^{(2)})$};
\node at (0.8,-0.9) {$K^*(\An^{(\mathrm{d})})$};
\begin{scope}
  \clip (-0.35,1.3) circle (0.95cm);
  \foreach \t in {-3.0,-2.6,...,3.0} {
    \draw[gray, dotted, line width=0.6pt] (-3+\t,-3) -- (3+\t,3);
  }
\end{scope}
\begin{scope}
  \clip (-1.75,-0.2) circle (0.95cm);
  \foreach \t in {-3.0,-2.6,...,3.0} {
    \draw[gray, dotted, line width=0.6pt] (-3+\t,-3) -- (3+\t,3);
  }
\end{scope}
\begin{scope}
  \clip (0.8,-0.9) circle (0.95cm);
  \foreach \t in {-3.0,-2.6,...,3.0} {
    \draw[gray, dotted, line width=0.6pt] (-3+\t,-3) -- (3+\t,3);
  }
\end{scope}
\end{tikzpicture}
\caption{\scriptsize{Theorem~\ref{n_K^*_PO}.
}}\label{fig:n_iPO_circles}
\end{subfigure}%
~
\begin{subfigure}[c]{0.45\textwidth}
\centering
\begin{tikzpicture}[thick,fill,scale=0.9]
\draw (0,0) ellipse (3 and 2.5);
\draw (-0.74, 0) ellipse (2.1 and 1.81 );
\node at (1.6,1.1) {$\mathrm{PO}$};
\node at (-0.74,0.3) { \small{$\mathrm{iFPO} = \bigcup\limits_{\Lambda\in \An} K^*(\Lambda)$ =}};
\node at (-0.74,-0.3) { \scriptsize{$K^*(\An^{(1)})   \cup \cdots \cup K^*(\An^{(\mathrm{d})})  $ }};
\begin{scope}
  \clip (-0.75,0) ellipse (2 and 1.8);
  \foreach \t in {-4,-3.6,...,4} {
    \draw[gray, dotted, line width=0.6pt] (-4+\t,-4) -- (4+\t,4);
  }
  \end{scope}
\end{tikzpicture}
\caption{\scriptsize{Corollary~\ref{n_iPO_PO}.}}\label{fig:n_iPO=K^*_PO}
\end{subfigure}
\caption{\footnotesize{Illustrations of (a) Theorem~\ref{n_K^*_PO}, and (b) Corollary~\ref{n_iPO_PO}, by rewriting  $ \cup_{\Lambda\in \An}K^*(\Lambda) $ as the union of $K^{*}(\An^{(j)})$, for $j=1,2,\dots, \mathrm{d}$, where $\mathrm{d}=\mathrm{d}(\An)$.}}
\label{fig:n_iPO}
\end{figure}

Finally, we remark that if $\mathbf{I}^*\in K\left({\bm\lambda}\right)$ for some ${\bm\lambda}\in \Delta_n(0,1)$, then $\mathbf{I}^*\in\cup_{\Lambda\in \An} K^*(\Lambda)=\mathrm{iFPO}$. Together with Proposition~\ref{n_iPO_PO}, under the full setting, Geoffrion-proper Pareto optimality implies inclusive and fair Pareto optimality, which further implies classical Pareto optimality.



\subsection{Revisiting Bilateral Case}
Under the bilateral setting, Corollary~\ref{2_result} and Corollary~\ref{n_iPO=PO} lead to the following summary.
\begin{corollary}\label{2_iPO_K^*_PO}
Let the bilateral setting $(\mathcal{N}_1; \mathcal{I}_1)$ satisfy $(\dagger)$. The following statements are equivalent.
\begin{enumerate}
\item[(i)] $I^*\in\mathcal{I}_1$ is Pareto optimal under the bilateral setting.
\vspace{-3mm}
\item[(ii)] $I^*\in K^*(\Lambda)$, where the set-valued function $K^*$ is given as in \eqref{eq:function_K_star_n}, for some $\Lambda\in \mathcal{A}_1$ given in \eqref{eq:A1}.
\vspace{-3mm}
\item[(iii)] $I^*$ is inclusive and fair Pareto optimal, with convex-group risks given in \eqref{eq:grouprisk}, under the bilateral setting.
\end{enumerate}
\end{corollary}

Figure~\ref{fig:2_iPO=K^*_PO} illustrates the set equalities, $\mathrm{PO}=\cup_{\Lambda\in\mathcal{A}_1} K^*(\Lambda) = \mathrm{iFPO}$, as stated in Corollary~\ref{2_iPO_K^*_PO}.
The ordered set partitions $\mathcal{P}\in\mathcal{P}(\mathcal{N}_1)$ are listed in \eqref{eq:1_set_order}. When $m=1$, i.e., $\mathcal{P}=\mathcal{P}_1^{(1)}=(\mathcal{N}_1)$, a decision vector is inclusive and fair Pareto optimal if and only if it minimizes a strict convex combination of the two agents' risk measures. Equivalently, it lies in $K^*(\mathcal{A}_1^{(1)})$, where $\mathcal{A}_1^{(1)}$ is given in \eqref{eq:A_1_1}. For $m=2$, consider $\mathcal{P}=\mathcal{P}_1^{(2)}=(\{1\},\{2\})$ (resp. $\mathcal{P}=\mathcal{P}_1^{(3)}=(\{2\},\{1\})$). Then, a decision vector is inclusive and fair Pareto optimal if and only if it is Pareto optimal between the two agents among decisions that are optimal for agent~$1$ (resp. agent~$2$). Equivalently, it lies in $K^*(\mathcal{A}_1^{(2)})$ (resp. $K^*(\mathcal{A}_1^{(3)})$), where $\mathcal{A}_1^{(2)}$ and $\mathcal{A}_1^{(3)}$ are given in \eqref{eq:A_1}. Although this case involves a bi-level optimization problem, each group actually consists of a single agent. Finally, Figure~\ref{fig:2_K*partition_A} re-interprets Figure~\ref{fig:2_partition*} in terms of the $\mathrm{d}(\mathcal{A}_1)=3$ classes of the sequential optimization problems.

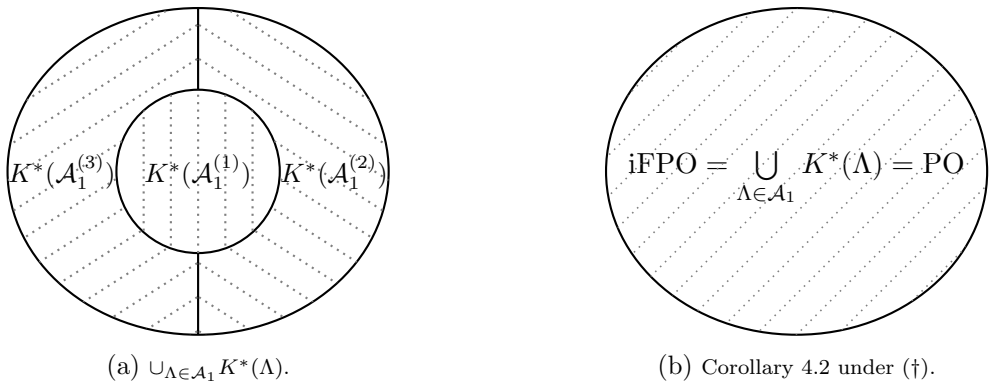
\begin{figure}[h!] 
\begin{subfigure}[C]{0.45\linewidth}
\centering
\begin{tikzpicture}[thick,fill,scale=0.9]

\draw (0,-2.4) -- (0,-1.2);
\draw (0, 1.2) -- (0, 2.4);

\draw (0,0) circle (1.2cm);
\draw (0,0) ellipse (2.8cm and 2.4cm); 

\begin{scope}
  \clip (0,0) ellipse (2.8cm and 2.4cm);        
  \clip (-10,-10) rectangle (10,10) (0,0) circle (1.2cm); 
  \clip (-3,-2.4) rectangle (0,2.4);             
  \foreach \y in {-4,-3.5,...,8} {
      \draw[gray, dotted] (-4,\y-2.4) -- (1,\y+0.8);
  }
\end{scope}

\begin{scope}
  \clip (0,0) ellipse (2.8cm and 2.4cm);        
  \clip (-10,-10) rectangle (10,10) (0,0) circle (1.2cm); 
  \clip (0,-2.4) rectangle (3,2.4);              
  \foreach \y in {-4,-3.5,...,8} {
      \draw[gray, dotted] (4,\y-2.4) -- (-1,\y+0.8);
  }
\end{scope}
\begin{scope} \clip (0,0) circle (1.2cm); 
\foreach \x in {-0.8,-0.4,...,1.2} { \draw[gray, dotted] (\x,-1.2) -- (\x,1.4); } 
\end{scope}

\node  at (0,0) {\small{$K^*(\mathcal{A}_1^{(1)})$}};
\node  at (-2, 0) {\small{$K^*(\mathcal{A}_1^{(3)})$}};
\node  at (2, 0)  {\small{$K^*(\mathcal{A}_1^{(2)})$}};
\end{tikzpicture}
\caption{\scriptsize{$\cup_{\Lambda\in \mathcal{A}_1}K^*(\Lambda)$.
}}\label{fig:2_K*partition_A}
\end{subfigure}%
~
\begin{subfigure}[c]{0.45\textwidth}
\centering
\begin{tikzpicture}[thick,fill,scale=0.9]
\draw (0,0) ellipse (2.8cm and 2.4cm);
\node at (0,-0.1) {$\mathrm{iFPO} = \bigcup\limits_{\Lambda\in \mathcal{A}_1} K^*(\Lambda)=\mathrm{PO}$};
\begin{scope} \clip (0,0) ellipse (2.8cm and 2.4cm); 
\foreach \t in {-6,-5.4,...,6} { \draw[gray, dotted, line width=0.6pt] (-6+\t,-6) -- (6+\t,6); } 
\end{scope}
\end{tikzpicture}
\caption{\scriptsize{Corollary~\ref{2_iPO_K^*_PO} under $(\dagger)$.}}\label{fig:2_iPO=K^*_PO}
\end{subfigure}
\caption{\footnotesize{Illustrations of (a) $ \cup_{\Lambda\in \mathcal{A}_1}K^*(\Lambda) $ for $n=1$, and (b) Corollary~\ref{2_iPO_K^*_PO} under ($\dagger$) for $n=1$.}}
\label{fig:2_iPO}
\end{figure}


\section{Illustrative Example}\label{sec:example_section}
This section illustrates inclusive and fair Pareto optimal decisions, equivalently sequentially optimal decisions with respect to the class $\mathcal{A}_n$, in the multilateral centralized risk-sharing problem. We consider the case $n=2$, corresponding to $n+1=3$ agents: the first two agents are policyholders holding initial risks $X_1$ and $X_2$, respectively, whereas the third agent is the central agent (insurer) and holds no initial risk.

Let $\Omega=\{\omega_1,\omega_2,\omega_3,\omega_4\}$, and suppose that the risk vector $(X_1,X_2)$ is given by
\begin{equation*}
\left(X_1,X_2\right)\left(\omega_1\right)=\left(0,0\right),\quad\left(X_1,X_2\right)\left(\omega_2\right)=\left(a,0\right),
\end{equation*}
\begin{equation*}
\left(X_1,X_2\right)\left(\omega_3\right)=\left(0,b\right),\quad\left(X_1,X_2\right)\left(\omega_4\right)=\left(a,b\right),
\end{equation*}
where $a,b>0$. Without loss of generality, assume $a<b$. Under the 
probability measure $P$,
\begin{equation*}
P\left(\omega=\omega_1\right)=P\left(X_1=0,X_2=0\right)=p_{00},\quad P\left(\omega=\omega_2\right)=P\left(X_1=a,X_2=0\right)=p_{a0},
\end{equation*}
\begin{equation*}
P\left(\omega=\omega_3\right)=P\left(X_1=0,X_2=b\right)=p_{0b},\quad P\left(\omega=\omega_4\right)=P\left(X_1=a,X_2=b\right)=p_{ab},
\end{equation*}
where $p_{00},p_{a0},p_{0b},p_{ab}>0$ and $p_{00}+p_{a0}+p_{0b}+p_{ab}=1$. Thus, $X_1$ and $X_2$ are Bernoulli-type random variables taking values in $\{0,a\}$ and $\{0,b\}$, respectively, with success probabilities $p_a=P(X_1=a)=p_{a0}+p_{ab}$, and $p_b=P(X_2=b)=p_{0b}+p_{ab}$.

For the feasible set, consider $\mathcal{I}_2=\mathcal{I}_0^2$. Given a ceding decision vector $\I=(I_1,I_2)\in\mathcal{I}_2$, the ceded risk vector $(I_1(X_1),I_2(X_2))$ is given by
\begin{equation*}
\left(I_1(X_1),I_2(X_2)\right)\left(\omega_1\right)=\left(0,0\right),\quad\left(I_1(X_1),I_2(X_2)\right)\left(\omega_2\right)=\left(i_a,0\right),
\end{equation*}
\begin{equation*}
\left(I_1(X_1),I_2(X_2)\right)\left(\omega_3\right)=\left(0,i_b\right),\quad\left(I_1(X_1),I_2(X_2)\right)\left(\omega_4\right)=\left(i_a,i_b\right),
\end{equation*}
for some $(i_a,i_b)\in[0,a]\times[0,b]$, with the respective probabilities $p_{00},p_{a0},p_{0b},p_{ab}$ under the objective probability measure $P$. Therefore, $I_1(X_1)$ and $I_2(X_2)$ are also Bernoulli-type random variables, taking values in $\{0,i_a\}$ and $\{0,i_b\}$, respectively, with success probabilities $p_a$ and $p_b$. Due to the Bernoulli nature of the risks in this example, with a slight abuse of notation, denote $\mathcal{I}_2=[0,a]\times[0,b]$, which is convex, and a feasible ceded loss decision vector $(i_a,i_b)\in\mathcal{I}_2$.

The two policyholders and the insurer are all endowed with expected shortfall risk measures, $\rho_1(\cdot)$, $\rho_2(\cdot)$, and $\rho_3(\cdot)$, with their respective risk tolerance levels denoted as $\alpha_1,\alpha_2,\alpha\in[0,1]$.\footnote{For $Z\in\mathcal{X}$ and $\beta\in[0,1]$, the expected shortfall risk measure of the random variable $Z$ at the risk tolerance level $\beta$ is defined as
\begin{equation*}
\mathrm{ES}_\beta(Z)=
\begin{cases}
\frac{1}{1-\beta}\int_{\beta}^{1}\mathrm{VaR}_{\gamma}(Z)d\gamma,&\text{ if }\beta\in[0,1),\\
\mathrm{esssup}(Z),&\text{ if }\beta=1,
\end{cases}
\end{equation*}
where the Value-at-Risk of $Z$ at $\gamma$ is defined by $\mathrm{VaR}_{\gamma}(Z)=\inf\{z\in\mathbb{R}:F_Z(z)\geq\gamma\}$, in which $F_Z(z)$, $z\in\mathbb{R}$, is the distribution function of $Z$, and where $\mathrm{esssup}(Z)$ denotes the essential supremum of $Z$.} The pricing rules, $\Pi_1(\cdot)$ and $\Pi_2(\cdot)$, are given by the expected value premium principles, with their respective safety loadings $\theta_1,\theta_2\geq 0$. The risk measures are convex, and the pricing rules are semi-linear.

By translation invariance of the expected shortfall risk measures, the time-0 risk position of the policyholders are given by: for $(i_a,i_b)\in\mathcal{I}_2$,
\begin{align*}
\rho_1(Y_1)=&\;\mathrm{ES}_{\alpha_1}(Y_1)=\mathrm{ES}_{\alpha_1}(X_1 - I_1(X_1) + \Pi_1(I_1(X_1)))\\=&\;\mathrm{ES}_{\alpha_1}(X_1 - I_1(X_1)) + (1+\theta_1) E[I_1(X_1)]\\
=&\;\begin{cases}
 \frac{p_a}{1-\alpha_1} a   + \left((1+\theta_1)   - \frac{1}{1-\alpha_1}\right) p_a i_a, & \text{ if } \alpha_1\in [0,1-p_a],\\
a + ((1+\theta_1)p_a-1)i_a, & \text{ if } \alpha_1\in [1-p_a,1];\\
\end{cases}
\end{align*}
\begin{align*}
\rho_2(Y_2)=&\;\mathrm{ES}_{\alpha_2}(Y_2)=\mathrm{ES}_{\alpha_2}(X_2 - I_2(X_2) + \Pi_2(I_2(X_2)))\\
=&\;\begin{cases}
 \frac{p_b}{1-\alpha_2} b   + \left((1+\theta_2)   - \frac{1}{1-\alpha_2}\right) p_b i_b, & \text{ if } \alpha_2\in [0,1-p_b],\\
b + ((1+\theta_2)p_b-1)i_b, & \text{ if } \alpha_2\in [1-p_b,1].\\
\end{cases}
\end{align*}
Also, by the translation invariance, the time-0 risk position of the insurer is given by: for $(i_a,i_b)\in\mathcal{I}_2$,
\begin{align*}
\rho_3(Y_3)=&\;\mathrm{ES}_{\alpha}(Y_3)=\mathrm{ES}_{\alpha}(I_1(X_1) + I_2(X_2) - \Pi_1(I_1(X_1)) - \Pi_2(I_2(X_2)))\\
=&\;\mathrm{ES}_{\alpha}(I_1(X_1) + I_2(X_2))- \Pi_1(I_1(X_1)) - \Pi_2(I_2(X_2));
\end{align*}
if $i_a\leq i_b$, then
\begin{equation*}
\rho_3(Y_3)=
\begin{cases}
\left(\frac{1}{1-\alpha}-\left(1+\theta_1\right)\right)p_ai_a+\left(\frac{1}{1-\alpha}-\left(1+\theta_2\right)\right)p_bi_b, &\text{ if }\alpha\in[0,p_{00}],\\
\left(\frac{1-p_{0b}-\alpha}{1-\alpha}-\left(1+\theta_1\right)p_a\right)i_a+\left(\frac{1}{1-\alpha}-\left(1+\theta_2\right)\right)p_bi_b, &\text{ if }\alpha\in[p_{00},1-p_b],\\
\left(\frac{p_{ab}}{1-\alpha}-\left(1+\theta_1\right)p_a\right)i_a+\left(1-\left(1+\theta_2\right)p_b\right)i_b, &\text{ if }\alpha\in[1-p_b,1-p_{ab}],\\
\left(1-\left(1+\theta_1\right)p_a\right)i_a+\left(1-\left(1+\theta_2\right)p_b\right)i_b, &\text{ if }\alpha\in[1-p_{ab},1];\\
\end{cases}
\end{equation*}
if $i_a\geq i_b$, then
\begin{equation*}
\rho_3(Y_3)=
\begin{cases}
\left(\frac{1}{1-\alpha}-\left(1+\theta_1\right)\right)p_ai_a+\left(\frac{1}{1-\alpha}-\left(1+\theta_2\right)\right)p_bi_b, &\text{ if }\alpha\in[0,p_{00}],\\
\left(\frac{1}{1-\alpha}-\left(1+\theta_1\right)\right)p_ai_a+\left(\frac{1-p_{a0}-\alpha}{1-\alpha}-\left(1+\theta_2\right)p_b\right)i_b, &\text{ if }\alpha\in[p_{00},1-p_a],\\
\left(1-\left(1+\theta_1\right)p_a\right)i_a+\left(\frac{p_{ab}}{1-\alpha}-\left(1+\theta_2\right)p_b\right)i_b, &\text{ if }\alpha\in[1-p_a,1-p_{ab}],\\
\left(1-\left(1+\theta_1\right)p_a\right)i_a+\left(1-\left(1+\theta_2\right)p_b\right)i_b, &\text{ if }\alpha\in[1-p_{ab},1].\\
\end{cases}
\end{equation*}

Assume that $\alpha_1\in [1-p_a,1]$, $\alpha_2\in [1-p_b,1]$, and $\alpha\in[\max\{1-p_a,1-p_b\},1-p_{ab}]$. For any $\bm{\lambda} = (\lambda_1, \lambda_2,\lambda_3)^T\in\Delta_2[0,1]$, and for any feasible subset $\mathcal{C}\subseteq\mathcal{I}_2=[0,a]\times[0,b]$,
\begin{align*}
K\left({\bm\lambda};\mathcal{C}\right)=&\;\argmin_{(i_a,i_b)\in\mathcal{C}}\left\{\lambda_1\rho_1\left(Y_1\right)+\lambda_2\rho_2\left(Y_2\right)+\lambda_3\rho_3\left(Y_3\right)\right\}\\
=&\;\argmin_{(i_a,i_b)\in\mathcal{C}}\bigg\{\left(\left(\lambda_1-\lambda_3\right)\left(\left(1+\theta_1\right)p_a-1\right)+\lambda_3\left(\frac{p_{ab}}{1-\alpha}-1\right)\mathds{1}_{\{i_a\leq i_b\}}\right)i_a+\\&\;\quad\quad\quad\quad\;\left(\left(\lambda_2-\lambda_3\right)\left(\left(1+\theta_2\right)p_b-1\right)+\lambda_3\left(\frac{p_{ab}}{1-\alpha}-1\right)\mathds{1}_{\{i_a>i_b\}}\right)i_b\bigg\}.
\end{align*}

Moreover, assume that $\left(1+\theta_1\right)p_a<1$, $\frac{p_{ab}}{1-\alpha}<\left(1+\theta_2\right)p_b<1$, and $1-\frac{p_{ab}}{1-\alpha}<\left(1-\left(1+\theta_1\right)p_a\right)+\left(1-\left(1+\theta_2\right)p_b\right)$. For any $(\lambda_1,\lambda_3)^T\in\Delta_1(0,1)$, there exists $\tilde{\lambda}\in(0,1)$ such that
\begin{align*}
&\lambda_1\left(\left(1+\theta_1\right)p_a-1\right)+\lambda_3\left(\frac{p_{ab}}{1-\alpha}-\left(1+\theta_1\right)p_a-\left(1+\theta_2\right)p_b+1\right)\\&
\begin{cases}
>0, & \text{ if }\lambda_1<\tilde{\lambda}\text{ and }\lambda_3>1-\tilde{\lambda},\\
=0, & \text{ if }\lambda_1=\tilde{\lambda}\text{ and }\lambda_3=1-\tilde{\lambda},\\
<0, & \text{ if }\lambda_1>\tilde{\lambda}\text{ and }\lambda_3<1-\tilde{\lambda}.
\end{cases}
\end{align*}
For any $i_a\in[0,a]$, and for any $(\lambda_1,\lambda_3)^T\in\Delta_1(0,1)$,
\begin{equation*}
\left(\left(\lambda_1-\lambda_3\right)\left(\left(1+\theta_1\right)p_a-1\right)+\lambda_3\left(\frac{p_{ab}}{1-\alpha}-1\right)\right)i_a-\lambda_3\left(\left(1+\theta_2\right)p_b-1\right)i_b
\end{equation*}
is strictly increasing in $i_b\in(i_a,b]$, while 
\begin{equation*}
\left(\lambda_1-\lambda_3\right)\left(\left(1+\theta_1\right)p_a-1\right)i_a+\lambda_3\left(\frac{p_{ab}}{1-\alpha}-\left(1+\theta_2\right)p_b\right)i_b
\end{equation*}
is strictly decreasing in $i_b\in[0,i_a)$. For any $\Lambda\in\mathcal{A}_2^{(3)}$, with $n\left(\Lambda\right)=2$,
\begin{align*}
&\;K^{*\left(1\right)}\left(\Lambda[\cdot,1]\right)=K\left((\lambda_1,0,\lambda_3)^T;\mathcal{I}_2\right)\\
=&\;\argmin_{\substack{(i_a,i_b)\in[0,a]\times[0,b]:\\i_a=i_b}}\left(\lambda_1\left(\left(1+\theta_1\right)p_a-1\right)+\lambda_3\left(\frac{p_{ab}}{1-\alpha}-\left(1+\theta_1\right)p_a-\left(1+\theta_2\right)p_b+1\right)\right)i_a\\
=&
\begin{cases}
\{(0,0)\}, & \text{ if }\lambda_1<\tilde{\lambda}\text{ and }\lambda_3>1-\tilde{\lambda},\\
\{(i_a,i_b)\in[0,a]\times[0,b]:i_a=i_b\}, & \text{ if }\lambda_1=\tilde{\lambda}\text{ and }\lambda_3=1-\tilde{\lambda},\\
\{(a,a)\}, & \text{ if }\lambda_1>\tilde{\lambda}\text{ and }\lambda_3<1-\tilde{\lambda},
\end{cases}
\end{align*}
and thus,
\begin{align*}
K^*\left(\Lambda\right)=&\;K^{*\left(2\right)}\left(\Lambda[\cdot,1:2]\right)=K\left(\Lambda[\cdot,2];K^{*\left(1\right)}\left(\Lambda[\cdot,1]\right)\right)=K\left((0,1,0)^T;K^{*\left(1\right)}\left(\Lambda[\cdot,1]\right)\right)\\
=&\;\argmin_{(i_a,i_b)\in K^{*\left(1\right)}\left(\Lambda[\cdot,1]\right)}\left(\left(1+\theta_2\right)p_b-1\right)i_b\\
=&
\begin{cases}
\{(0,0)\}, & \text{ if }\lambda_1<\tilde{\lambda}\text{ and }\lambda_3>1-\tilde{\lambda},\\
\{(a,a)\}, & \text{ if }\lambda_1\geq\tilde{\lambda}\text{ and }\lambda_3\leq 1-\tilde{\lambda}.
\end{cases}
\end{align*}
Hence, $K^*(\mathcal{A}_2^{(3)})=\{(0,0),(a,a)\}$. By Corollary \ref{n_iPO=PO}, the ceded loss decision vectors $(0,0)$ and $(a,a)$ are inclusive and fair Pareto optimal with the ordered set partition $\mathcal{P}_2^{(3)} = ( \{1,3\}, \{2\} )$; at the first level of the optimizations, $(0,0)$ corresponds to the weight vectors $(\lambda_1,\lambda_3)^T\in\Delta_1(0,1)$ with $\lambda_1<\tilde{\lambda}$ and $\lambda_3>1-\tilde{\lambda}$, while $(a,a)$ corresponds to those with $\lambda_1\geq\tilde{\lambda}$ and $\lambda_3\leq 1-\tilde{\lambda}$.

We close this section by outlining a particular set of numerical parameters that satisfies all assumptions; for example, $p_{00}=0.1$, $p_{a0}=0.35$, $p_{0b}=0.4$, $p_{ab}=0.15$, $\alpha_1=0.95$, $\alpha_2=0.95$, $\alpha=0.75$, $\theta_1=0.5$, $\theta_2=0.2$.


\section{Concluding Remarks and Future Directions}\label{sec:conclusion}
This paper studies multilateral centralized risk sharing with endogenous prices. In contrast to models where premiums are exogenous decisions, premiums here are determined by pricing functionals applied to ceded risks. This makes the risk-sharing problem naturally multiobjective while preventing the usual deterministic-transfer argument for tracing the Pareto frontier. We show that classical Pareto optimality, although still the canonical efficiency criterion, does not by itself capture whether all agents are represented in a balanced sequential decision process.

The paper introduces inclusive and fair Pareto optimality as a representation-based refinement of Pareto optimality. The concept requires every agent to appear exactly once in a finite ordered sequence of optimizations, either individually or as part of a group. The main result proves that inclusive and fair Pareto optimality is exactly characterized by balanced sequential optimization. Thus, Geoffrion-proper Pareto optimality implies inclusive and fair Pareto optimality, which in turn implies classical Pareto optimality. The proposed concept therefore lies between proper efficiency and Pareto optimality; like proper efficiency, it refines Pareto optimality, but it does so by controlling representation of agents rather than the magnitude of trade-offs.

Several directions for future research follow naturally. First, the framework could be applied to decentralized insurance and peer-to-peer risk sharing (see, such as, \cite{denuit2022risksharing,abdikerimova2022peer,DENUIT20251,anthropelos2026}, and beyond. In such settings, risk transfers may occur among participants through mutual pools, coalitions, or platforms rather than through a central agent. The question of representation becomes particularly important; one may ask whether all participants are incorporated in the construction of the risk-sharing allocation, whether some coalitions are repeatedly prioritized, and how contribution rules or platform fees should be treated as endogenous counterparts of premiums. Inclusive and fair Pareto optimality may provide a useful way to distinguish balanced decentralized risk-sharing arrangements from efficient but representationally imbalanced ones.

Second, it would be valuable to study the relationship between inclusive and fair Pareto optimality and Stackelberg equilibrium. In centralized insurance, the insurer may act as a leader who chooses pricing rules, contract menus, or underwriting terms, while policyholders respond by choosing coverage or ceding functions. Stackelberg equilibrium captures strategic timing and incentive responses, whereas inclusive and fair Pareto optimality captures balanced representation in a sequential optimization procedure. Future work could ask when Stackelberg equilibria are inclusive and fair Pareto optimal, and conversely when inclusive and fair Pareto optimal decisions can be implemented through suitable leaders-followers mechanisms.

Third, future research could generalize the group-risk functional. This paper uses convex-group risk measures, which are analytically tractable and ensure that all agents in a group are explicitly represented. Other applications may call for non-linear group-risk measures, robust aggregation, bargaining-based group objectives, or coalition-specific welfare criteria. Studying which group-risk functionals preserve the sequential characterization would broaden the applicability of the framework. Finally, because the number of ordered set partitions grows quickly with the number of agents, computational methods for identifying inclusive and fair Pareto optimal decisions in large markets are an important direction for future work.





\bibliographystyle{apacite}
{\small{
\bibliography{ref.bib}{}

\begin{thebibliography}{}

\bibitem [\protect \citeauthoryear {%
Abdikerimova%
\ \BBA {} Feng%
}{%
Abdikerimova%
\ \BBA {} Feng%
}{%
{\protect \APACyear {2022}}%
}]{%
abdikerimova2022peer}
\APACinsertmetastar {%
abdikerimova2022peer}%
\begin{APACrefauthors}%
Abdikerimova, S.%
\BCBT {}\ \BBA {} Feng, R.%
\end{APACrefauthors}%
\unskip\
\newblock
\APACrefYearMonthDay{2022}{}{}.
\newblock
{\BBOQ}\APACrefatitle {Peer-to-peer multi-risk insurance and mutual aid} {Peer-to-peer multi-risk insurance and mutual aid}.{\BBCQ}
\newblock
\APACjournalVolNumPages{European Journal of Operational Research}{299}{2}{735--749}.
\PrintBackRefs{\CurrentBib}

\bibitem [\protect \citeauthoryear {%
Anthropelos%
, Feng%
\BCBL {}\ \BBA {} Kim%
}{%
Anthropelos%
\ \protect \BOthers {.}}{%
{\protect \APACyear {2026}}%
}]{%
anthropelos2026}
\APACinsertmetastar {%
anthropelos2026}%
\begin{APACrefauthors}%
Anthropelos, M.%
, Feng, R.%
\BCBL {}\ \BBA {} Kim, S.%
\end{APACrefauthors}%
\unskip\
\newblock
\APACrefYearMonthDay{2026}{}{}.
\newblock
{\BBOQ}\APACrefatitle {On the expansion of risk pooling} {On the expansion of risk pooling}.{\BBCQ}
\newblock
\APACjournalVolNumPages{Management Science}{}{}{}.
\PrintBackRefs{\CurrentBib}

\bibitem [\protect \citeauthoryear {%
Arrow%
}{%
Arrow%
}{%
{\protect \APACyear {1963}}%
}]{%
arrow1963uncertainty}
\APACinsertmetastar {%
arrow1963uncertainty}%
\begin{APACrefauthors}%
Arrow, K\BPBI J.%
\end{APACrefauthors}%
\unskip\
\newblock
\APACrefYearMonthDay{1963}{}{}.
\newblock
{\BBOQ}\APACrefatitle {Uncertainty and the Welfare Economics of Medical Care} {Uncertainty and the welfare economics of medical care}.{\BBCQ}
\newblock
\APACjournalVolNumPages{American Economic Review}{53}{5}{941--973}.
\PrintBackRefs{\CurrentBib}

\bibitem [\protect \citeauthoryear {%
Arrow%
}{%
Arrow%
}{%
{\protect \APACyear {1974}}%
}]{%
Arrow1974}
\APACinsertmetastar {%
Arrow1974}%
\begin{APACrefauthors}%
Arrow, K\BPBI J.%
\end{APACrefauthors}%
\unskip\
\newblock
\APACrefYearMonthDay{1974}{}{}.
\newblock
{\BBOQ}\APACrefatitle {Optimal insurance and generalized deductibles} {Optimal insurance and generalized deductibles}.{\BBCQ}
\newblock
\APACjournalVolNumPages{Scandinavian Actuarial Journal}{1974}{1}{1--42}.
\PrintBackRefs{\CurrentBib}

\bibitem [\protect \citeauthoryear {%
Asimit%
\ \BBA {} Boonen%
}{%
Asimit%
\ \BBA {} Boonen%
}{%
{\protect \APACyear {2018}}%
}]{%
ASIMIT2018}
\APACinsertmetastar {%
ASIMIT2018}%
\begin{APACrefauthors}%
Asimit, V.%
\BCBT {}\ \BBA {} Boonen, T\BPBI J.%
\end{APACrefauthors}%
\unskip\
\newblock
\APACrefYearMonthDay{2018}{}{}.
\newblock
{\BBOQ}\APACrefatitle {Insurance with multiple insurers: A game-theoretic approach} {Insurance with multiple insurers: A game-theoretic approach}.{\BBCQ}
\newblock
\APACjournalVolNumPages{European Journal of Operational Research}{267}{2}{778--790}.
\PrintBackRefs{\CurrentBib}

\bibitem [\protect \citeauthoryear {%
Asimit%
, Boonen%
, Chi%
\BCBL {}\ \BBA {} Chong%
}{%
Asimit%
\ \protect \BOthers {.}}{%
{\protect \APACyear {2021}}%
}]{%
ASIMIT2021587}
\APACinsertmetastar {%
ASIMIT2021587}%
\begin{APACrefauthors}%
Asimit, V.%
, Boonen, T\BPBI J.%
, Chi, Y.%
\BCBL {}\ \BBA {} Chong, W\BPBI F.%
\end{APACrefauthors}%
\unskip\
\newblock
\APACrefYearMonthDay{2021}{}{}.
\newblock
{\BBOQ}\APACrefatitle {Risk sharing with multiple indemnity environments} {Risk sharing with multiple indemnity environments}.{\BBCQ}
\newblock
\APACjournalVolNumPages{European Journal of Operational Research}{295}{2}{587--603}.
\PrintBackRefs{\CurrentBib}

\bibitem [\protect \citeauthoryear {%
Benson%
}{%
Benson%
}{%
{\protect \APACyear {1998}}%
}]{%
benson1998outer}
\APACinsertmetastar {%
benson1998outer}%
\begin{APACrefauthors}%
Benson, H\BPBI P.%
\end{APACrefauthors}%
\unskip\
\newblock
\APACrefYearMonthDay{1998}{}{}.
\newblock
{\BBOQ}\APACrefatitle {An outer approximation algorithm for generating all efficient extreme points in the outcome set of a multiple objective linear programming problem} {An outer approximation algorithm for generating all efficient extreme points in the outcome set of a multiple objective linear programming problem}.{\BBCQ}
\newblock
\APACjournalVolNumPages{Journal of Global Optimization}{13}{1}{1--24}.
\PrintBackRefs{\CurrentBib}

\bibitem [\protect \citeauthoryear {%
Bernard%
, He%
, Yan%
\BCBL {}\ \BBA {} Zhou%
}{%
Bernard%
\ \protect \BOthers {.}}{%
{\protect \APACyear {2015}}%
}]{%
Bernard2015}
\APACinsertmetastar {%
Bernard2015}%
\begin{APACrefauthors}%
Bernard, C.%
, He, X.%
, Yan, J\BHBI A.%
\BCBL {}\ \BBA {} Zhou, X\BPBI Y.%
\end{APACrefauthors}%
\unskip\
\newblock
\APACrefYearMonthDay{2015}{}{}.
\newblock
{\BBOQ}\APACrefatitle {Optimal insurance design under rank-dependent expected utility} {Optimal insurance design under rank-dependent expected utility}.{\BBCQ}
\newblock
\APACjournalVolNumPages{Mathematical Finance}{25}{1}{154--186}.
\PrintBackRefs{\CurrentBib}

\bibitem [\protect \citeauthoryear {%
Bernard%
, Liu%
\BCBL {}\ \BBA {} Vanduffel%
}{%
Bernard%
\ \protect \BOthers {.}}{%
{\protect \APACyear {2020}}%
}]{%
BERNARD2020}
\APACinsertmetastar {%
BERNARD2020}%
\begin{APACrefauthors}%
Bernard, C.%
, Liu, F.%
\BCBL {}\ \BBA {} Vanduffel, S.%
\end{APACrefauthors}%
\unskip\
\newblock
\APACrefYearMonthDay{2020}{}{}.
\newblock
{\BBOQ}\APACrefatitle {Optimal insurance in the presence of multiple policyholders} {Optimal insurance in the presence of multiple policyholders}.{\BBCQ}
\newblock
\APACjournalVolNumPages{Journal of Economic Behavior \& Organization}{180}{}{638-656}.
\PrintBackRefs{\CurrentBib}

\bibitem [\protect \citeauthoryear {%
B{\'e}rub{\'e}%
, Gendreau%
\BCBL {}\ \BBA {} Potvin%
}{%
B{\'e}rub{\'e}%
\ \protect \BOthers {.}}{%
{\protect \APACyear {2009}}%
}]{%
BERUBE2009}
\APACinsertmetastar {%
BERUBE2009}%
\begin{APACrefauthors}%
B{\'e}rub{\'e}, J\BHBI F.%
, Gendreau, M.%
\BCBL {}\ \BBA {} Potvin, J\BHBI Y.%
\end{APACrefauthors}%
\unskip\
\newblock
\APACrefYearMonthDay{2009}{}{}.
\newblock
{\BBOQ}\APACrefatitle {An exact $\varepsilon$-constraint method for bi-objective combinatorial optimization problems: Application to the Traveling Salesman Problem with Profits} {An exact $\varepsilon$-constraint method for bi-objective combinatorial optimization problems: Application to the traveling salesman problem with profits}.{\BBCQ}
\newblock
\APACjournalVolNumPages{European Journal of Operational Research}{194}{1}{39--50}.
\PrintBackRefs{\CurrentBib}

\bibitem [\protect \citeauthoryear {%
Boonen%
, Chong%
\BCBL {}\ \BBA {} Ghossoub%
}{%
Boonen%
\ \protect \BOthers {.}}{%
{\protect \APACyear {2024}}%
}]{%
boonenchongghossub2024}
\APACinsertmetastar {%
boonenchongghossub2024}%
\begin{APACrefauthors}%
Boonen, T\BPBI J.%
, Chong, W\BPBI F.%
\BCBL {}\ \BBA {} Ghossoub, M.%
\end{APACrefauthors}%
\unskip\
\newblock
\APACrefYearMonthDay{2024}{}{}.
\newblock
{\BBOQ}\APACrefatitle {Pareto-efficient risk sharing in centralized insurance markets with application to flood risk} {Pareto-efficient risk sharing in centralized insurance markets with application to flood risk}.{\BBCQ}
\newblock
\APACjournalVolNumPages{Journal of Risk and Insurance}{91}{2}{449--488}.
\PrintBackRefs{\CurrentBib}

\bibitem [\protect \citeauthoryear {%
Boonen%
\ \BBA {} Ghossoub%
}{%
Boonen%
\ \BBA {} Ghossoub%
}{%
{\protect \APACyear {2023}}%
}]{%
boonen2023}
\APACinsertmetastar {%
boonen2023}%
\begin{APACrefauthors}%
Boonen, T\BPBI J.%
\BCBT {}\ \BBA {} Ghossoub, M.%
\end{APACrefauthors}%
\unskip\
\newblock
\APACrefYearMonthDay{2023}{}{}.
\newblock
{\BBOQ}\APACrefatitle {Bowley vs. {P}areto optima in reinsurance contracting} {Bowley vs. {P}areto optima in reinsurance contracting}.{\BBCQ}
\newblock
\APACjournalVolNumPages{European Journal of Operational Research}{307}{1}{382--391}.
\PrintBackRefs{\CurrentBib}

\bibitem [\protect \citeauthoryear {%
Boonen%
, Han%
, Liu%
\BCBL {}\ \BBA {} Wang%
}{%
Boonen%
\ \protect \BOthers {.}}{%
{\protect \APACyear {2025}}%
}]{%
boonen2025pareto}
\APACinsertmetastar {%
boonen2025pareto}%
\begin{APACrefauthors}%
Boonen, T\BPBI J.%
, Han, X.%
, Liu, P.%
\BCBL {}\ \BBA {} Wang, J.%
\end{APACrefauthors}%
\unskip\
\newblock
\APACrefYearMonthDay{2025}{}{}.
\newblock
{\BBOQ}\APACrefatitle {Pareto-optimal reinsurance under dependence uncertainty} {Pareto-optimal reinsurance under dependence uncertainty}.{\BBCQ}
\newblock
\APACjournalVolNumPages{arXiv preprint arXiv:2512.11430}{}{}{}.
\PrintBackRefs{\CurrentBib}

\bibitem [\protect \citeauthoryear {%
Boonen%
\ \BBA {} Jiang%
}{%
Boonen%
\ \BBA {} Jiang%
}{%
{\protect \APACyear {2022}}%
}]{%
BOONEN2022}
\APACinsertmetastar {%
BOONEN2022}%
\begin{APACrefauthors}%
Boonen, T\BPBI J.%
\BCBT {}\ \BBA {} Jiang, W.%
\end{APACrefauthors}%
\unskip\
\newblock
\APACrefYearMonthDay{2022}{}{}.
\newblock
{\BBOQ}\APACrefatitle {A marginal indemnity function approach to optimal reinsurance under the {V}ajda condition} {A marginal indemnity function approach to optimal reinsurance under the {V}ajda condition}.{\BBCQ}
\newblock
\APACjournalVolNumPages{European Journal of Operational Research}{303}{2}{928--944}.
\PrintBackRefs{\CurrentBib}

\bibitem [\protect \citeauthoryear {%
Boonen%
, Ng%
, Ng%
\BCBL {}\ \BBA {} Nguyen%
}{%
Boonen%
\ \protect \BOthers {.}}{%
{\protect \APACyear {2026}}%
}]{%
boonen2026pareto}
\APACinsertmetastar {%
boonen2026pareto}%
\begin{APACrefauthors}%
Boonen, T\BPBI J.%
, Ng, K\BPBI T\BPBI H.%
, Ng, T\BPBI W.%
\BCBL {}\ \BBA {} Nguyen, T.%
\end{APACrefauthors}%
\unskip\
\newblock
\APACrefYearMonthDay{2026}{}{}.
\newblock
{\BBOQ}\APACrefatitle {Pareto and {B}owley Reinsurance Games in Peer-to-Peer Insurance} {Pareto and {B}owley reinsurance games in peer-to-peer insurance}.{\BBCQ}
\newblock
\APACjournalVolNumPages{arXiv preprint arXiv:2602.14223}{}{}{}.
\PrintBackRefs{\CurrentBib}

\bibitem [\protect \citeauthoryear {%
Boonen%
\ \BBA {} Zhang%
}{%
Boonen%
\ \BBA {} Zhang%
}{%
{\protect \APACyear {2021}}%
}]{%
BoonenZhang2021}
\APACinsertmetastar {%
BoonenZhang2021}%
\begin{APACrefauthors}%
Boonen, T\BPBI J.%
\BCBT {}\ \BBA {} Zhang, Y.%
\end{APACrefauthors}%
\unskip\
\newblock
\APACrefYearMonthDay{2021}{}{}.
\newblock
{\BBOQ}\APACrefatitle {OPTIMAL REINSURANCE DESIGN WITH DISTORTION RISK MEASURES AND ASYMMETRIC INFORMATION} {Optimal reinsurance design with distortion risk measures and asymmetric information}.{\BBCQ}
\newblock
\APACjournalVolNumPages{ASTIN Bulletin}{51}{2}{607--629}.
\PrintBackRefs{\CurrentBib}

\bibitem [\protect \citeauthoryear {%
Borch%
}{%
Borch%
}{%
{\protect \APACyear {1960}}%
}]{%
borch1960attempt}
\APACinsertmetastar {%
borch1960attempt}%
\begin{APACrefauthors}%
Borch, K.%
\end{APACrefauthors}%
\unskip\
\newblock
\APACrefYearMonthDay{1960}{}{}.
\newblock
{\BBOQ}\APACrefatitle {An attempt to determine the optimum amount of stop loss reinsurance} {An attempt to determine the optimum amount of stop loss reinsurance}.{\BBCQ}
\newblock
\APACjournalVolNumPages{Transactions of the 16th International Congress of Actuaries}{I}{3}{597--610}.
\PrintBackRefs{\CurrentBib}

\bibitem [\protect \citeauthoryear {%
Borwein%
\ \BBA {} Zhuang%
}{%
Borwein%
\ \BBA {} Zhuang%
}{%
{\protect \APACyear {1993}}%
}]{%
borwein1993super}
\APACinsertmetastar {%
borwein1993super}%
\begin{APACrefauthors}%
Borwein, J\BPBI M.%
\BCBT {}\ \BBA {} Zhuang, D.%
\end{APACrefauthors}%
\unskip\
\newblock
\APACrefYearMonthDay{1993}{}{}.
\newblock
{\BBOQ}\APACrefatitle {Super efficiency in vector optimization} {Super efficiency in vector optimization}.{\BBCQ}
\newblock
\APACjournalVolNumPages{Transactions of the American Mathematical Society}{338}{1}{105--122}.
\PrintBackRefs{\CurrentBib}

\bibitem [\protect \citeauthoryear {%
Cai%
, Liu%
\BCBL {}\ \BBA {} Wang%
}{%
Cai%
\ \protect \BOthers {.}}{%
{\protect \APACyear {2017}}%
}]{%
CAI2017}
\APACinsertmetastar {%
CAI2017}%
\begin{APACrefauthors}%
Cai, J.%
, Liu, H.%
\BCBL {}\ \BBA {} Wang, R.%
\end{APACrefauthors}%
\unskip\
\newblock
\APACrefYearMonthDay{2017}{}{}.
\newblock
{\BBOQ}\APACrefatitle {Pareto-optimal reinsurance arrangements under general model settings} {Pareto-optimal reinsurance arrangements under general model settings}.{\BBCQ}
\newblock
\APACjournalVolNumPages{Insurance: Mathematics and Economics}{77}{}{24--37}.
\PrintBackRefs{\CurrentBib}

\bibitem [\protect \citeauthoryear {%
Cai%
, Tan%
, Weng%
\BCBL {}\ \BBA {} Zhang%
}{%
Cai%
\ \protect \BOthers {.}}{%
{\protect \APACyear {2008}}%
}]{%
CAI2008}
\APACinsertmetastar {%
CAI2008}%
\begin{APACrefauthors}%
Cai, J.%
, Tan, K\BPBI S.%
, Weng, C.%
\BCBL {}\ \BBA {} Zhang, Y.%
\end{APACrefauthors}%
\unskip\
\newblock
\APACrefYearMonthDay{2008}{}{}.
\newblock
{\BBOQ}\APACrefatitle {Optimal reinsurance under {VaR} and {CTE} risk measures} {Optimal reinsurance under {VaR} and {CTE} risk measures}.{\BBCQ}
\newblock
\APACjournalVolNumPages{Insurance: Mathematics and Economics}{43}{1}{185--196}.
\PrintBackRefs{\CurrentBib}

\bibitem [\protect \citeauthoryear {%
Cheung%
}{%
Cheung%
}{%
{\protect \APACyear {2010}}%
}]{%
Cheung2010}
\APACinsertmetastar {%
Cheung2010}%
\begin{APACrefauthors}%
Cheung, K\BPBI C.%
\end{APACrefauthors}%
\unskip\
\newblock
\APACrefYearMonthDay{2010}{}{}.
\newblock
{\BBOQ}\APACrefatitle {Optimal Reinsurance Revisited – A Geometric Approach} {Optimal reinsurance revisited – a geometric approach}.{\BBCQ}
\newblock
\APACjournalVolNumPages{ASTIN Bulletin}{40}{1}{221--239}.
\PrintBackRefs{\CurrentBib}

\bibitem [\protect \citeauthoryear {%
Cheung%
, Chong%
\BCBL {}\ \BBA {} Lo%
}{%
Cheung%
\ \protect \BOthers {.}}{%
{\protect \APACyear {2019}}%
}]{%
Cheung2019}
\APACinsertmetastar {%
Cheung2019}%
\begin{APACrefauthors}%
Cheung, K\BPBI C.%
, Chong, W\BPBI F.%
\BCBL {}\ \BBA {} Lo, A.%
\end{APACrefauthors}%
\unskip\
\newblock
\APACrefYearMonthDay{2019}{}{}.
\newblock
{\BBOQ}\APACrefatitle {Budget-constrained optimal reinsurance design under coherent risk measures} {Budget-constrained optimal reinsurance design under coherent risk measures}.{\BBCQ}
\newblock
\APACjournalVolNumPages{Scandinavian Actuarial Journal}{2019}{9}{729--751}.
\PrintBackRefs{\CurrentBib}

\bibitem [\protect \citeauthoryear {%
Cheung%
, Chong%
\BCBL {}\ \BBA {} Yam%
}{%
Cheung%
\ \protect \BOthers {.}}{%
{\protect \APACyear {2015}}%
}]{%
Cheung2015}
\APACinsertmetastar {%
Cheung2015}%
\begin{APACrefauthors}%
Cheung, K\BPBI C.%
, Chong, W\BPBI F.%
\BCBL {}\ \BBA {} Yam, S\BPBI C\BPBI P.%
\end{APACrefauthors}%
\unskip\
\newblock
\APACrefYearMonthDay{2015}{}{}.
\newblock
{\BBOQ}\APACrefatitle {Convex ordering for insurance preferences} {Convex ordering for insurance preferences}.{\BBCQ}
\newblock
\APACjournalVolNumPages{Insurance: Mathematics and Economics}{64}{}{409--416}.
\PrintBackRefs{\CurrentBib}

\bibitem [\protect \citeauthoryear {%
Chi%
\ \BBA {} Tan%
}{%
Chi%
\ \BBA {} Tan%
}{%
{\protect \APACyear {2011}}%
}]{%
ChiTan2011}
\APACinsertmetastar {%
ChiTan2011}%
\begin{APACrefauthors}%
Chi, Y.%
\BCBT {}\ \BBA {} Tan, K\BPBI S.%
\end{APACrefauthors}%
\unskip\
\newblock
\APACrefYearMonthDay{2011}{}{}.
\newblock
{\BBOQ}\APACrefatitle {Optimal Reinsurance under {VaR} and {CVaR} Risk Measures: A Simplified Approach} {Optimal reinsurance under {VaR} and {CVaR} risk measures: A simplified approach}.{\BBCQ}
\newblock
\APACjournalVolNumPages{ASTIN Bulletin}{41}{2}{487--509}.
\PrintBackRefs{\CurrentBib}

\bibitem [\protect \citeauthoryear {%
Chi%
\ \BBA {} Wei%
}{%
Chi%
\ \BBA {} Wei%
}{%
{\protect \APACyear {2020}}%
}]{%
chi2020optimal}
\APACinsertmetastar {%
chi2020optimal}%
\begin{APACrefauthors}%
Chi, Y.%
\BCBT {}\ \BBA {} Wei, W.%
\end{APACrefauthors}%
\unskip\
\newblock
\APACrefYearMonthDay{2020}{}{}.
\newblock
{\BBOQ}\APACrefatitle {Optimal insurance with background risk: An analysis of general dependence structures} {Optimal insurance with background risk: An analysis of general dependence structures}.{\BBCQ}
\newblock
\APACjournalVolNumPages{Finance and Stochastics}{24}{4}{903--937}.
\PrintBackRefs{\CurrentBib}

\bibitem [\protect \citeauthoryear {%
Cooper%
, Hunter%
\BCBL {}\ \BBA {} Nagaraj%
}{%
Cooper%
\ \protect \BOthers {.}}{%
{\protect \APACyear {2020}}%
}]{%
Cooper2020}
\APACinsertmetastar {%
Cooper2020}%
\begin{APACrefauthors}%
Cooper, K.%
, Hunter, S\BPBI R.%
\BCBL {}\ \BBA {} Nagaraj, K.%
\end{APACrefauthors}%
\unskip\
\newblock
\APACrefYearMonthDay{2020}{}{}.
\newblock
{\BBOQ}\APACrefatitle {Biobjective Simulation Optimization on Integer Lattices Using the Epsilon-Constraint Method in a Retrospective Approximation Framework} {Biobjective simulation optimization on integer lattices using the epsilon-constraint method in a retrospective approximation framework}.{\BBCQ}
\newblock
\APACjournalVolNumPages{INFORMS Journal on Computing}{32}{4}{1080--1100}.
\PrintBackRefs{\CurrentBib}

\bibitem [\protect \citeauthoryear {%
Cui%
, Yang%
\BCBL {}\ \BBA {} Wu%
}{%
Cui%
\ \protect \BOthers {.}}{%
{\protect \APACyear {2013}}%
}]{%
CUI2013}
\APACinsertmetastar {%
CUI2013}%
\begin{APACrefauthors}%
Cui, W.%
, Yang, J.%
\BCBL {}\ \BBA {} Wu, L.%
\end{APACrefauthors}%
\unskip\
\newblock
\APACrefYearMonthDay{2013}{}{}.
\newblock
{\BBOQ}\APACrefatitle {Optimal reinsurance minimizing the distortion risk measure under general reinsurance premium principles} {Optimal reinsurance minimizing the distortion risk measure under general reinsurance premium principles}.{\BBCQ}
\newblock
\APACjournalVolNumPages{Insurance: Mathematics and Economics}{53}{1}{74--85}.
\PrintBackRefs{\CurrentBib}

\bibitem [\protect \citeauthoryear {%
Denuit%
, Dhaene%
, Ghossoub%
\BCBL {}\ \BBA {} Robert%
}{%
Denuit%
\ \protect \BOthers {.}}{%
{\protect \APACyear {2025}}%
}]{%
DENUIT20251}
\APACinsertmetastar {%
DENUIT20251}%
\begin{APACrefauthors}%
Denuit, M.%
, Dhaene, J.%
, Ghossoub, M.%
\BCBL {}\ \BBA {} Robert, C\BPBI Y.%
\end{APACrefauthors}%
\unskip\
\newblock
\APACrefYearMonthDay{2025}{}{}.
\newblock
{\BBOQ}\APACrefatitle {Comonotonicity and {P}areto optimality, with application to collaborative insurance} {Comonotonicity and {P}areto optimality, with application to collaborative insurance}.{\BBCQ}
\newblock
\APACjournalVolNumPages{Insurance: Mathematics and Economics}{120}{}{1-16}.
\PrintBackRefs{\CurrentBib}

\bibitem [\protect \citeauthoryear {%
Denuit%
, Dhaene%
\BCBL {}\ \BBA {} Robert%
}{%
Denuit%
\ \protect \BOthers {.}}{%
{\protect \APACyear {2022}}%
}]{%
denuit2022risksharing}
\APACinsertmetastar {%
denuit2022risksharing}%
\begin{APACrefauthors}%
Denuit, M.%
, Dhaene, J.%
\BCBL {}\ \BBA {} Robert, C\BPBI Y.%
\end{APACrefauthors}%
\unskip\
\newblock
\APACrefYearMonthDay{2022}{}{}.
\newblock
{\BBOQ}\APACrefatitle {Risk-sharing rules and their properties, with applications to peer-to-peer insurance} {Risk-sharing rules and their properties, with applications to peer-to-peer insurance}.{\BBCQ}
\newblock
\APACjournalVolNumPages{Journal of Risk and Insurance}{89}{3}{615-667}.
\PrintBackRefs{\CurrentBib}

\bibitem [\protect \citeauthoryear {%
Evans%
}{%
Evans%
}{%
{\protect \APACyear {1984}}%
}]{%
evans1984}
\APACinsertmetastar {%
evans1984}%
\begin{APACrefauthors}%
Evans, G\BPBI W.%
\end{APACrefauthors}%
\unskip\
\newblock
\APACrefYearMonthDay{1984}{}{}.
\newblock
{\BBOQ}\APACrefatitle {An Overview of Techniques for Solving Multiobjective Mathematical Programs} {An overview of techniques for solving multiobjective mathematical programs}.{\BBCQ}
\newblock
\APACjournalVolNumPages{Management Science}{30}{11}{1268--1282}.
\PrintBackRefs{\CurrentBib}

\bibitem [\protect \citeauthoryear {%
{Garc{\'\i}a-Casta\~no}%
, Melguizo-Padial%
\BCBL {}\ \BBA {} Parzanese%
}{%
{Garc{\'\i}a-Casta\~no}%
\ \protect \BOthers {.}}{%
{\protect \APACyear {2023}}%
}]{%
GarciaCastanoFernando2023}
\APACinsertmetastar {%
GarciaCastanoFernando2023}%
\begin{APACrefauthors}%
{Garc{\'\i}a-Casta\~no}, F.%
, Melguizo-Padial, M\BPBI {\'A}.%
\BCBL {}\ \BBA {} Parzanese, G.%
\end{APACrefauthors}%
\unskip\
\newblock
\APACrefYearMonthDay{2023}{}{}.
\newblock
{\BBOQ}\APACrefatitle {Sublinear scalarizations for proper and approximate proper efficient points in nonconvex vector optimization} {Sublinear scalarizations for proper and approximate proper efficient points in nonconvex vector optimization}.{\BBCQ}
\newblock
\APACjournalVolNumPages{Mathematical Methods of Operations Research}{97}{3}{367--382}.
\PrintBackRefs{\CurrentBib}

\bibitem [\protect \citeauthoryear {%
Geoffrion%
}{%
Geoffrion%
}{%
{\protect \APACyear {1968}}%
}]{%
GEOFFRION1968}
\APACinsertmetastar {%
GEOFFRION1968}%
\begin{APACrefauthors}%
Geoffrion, A\BPBI M.%
\end{APACrefauthors}%
\unskip\
\newblock
\APACrefYearMonthDay{1968}{}{}.
\newblock
{\BBOQ}\APACrefatitle {Proper efficiency and the theory of vector maximization} {Proper efficiency and the theory of vector maximization}.{\BBCQ}
\newblock
\APACjournalVolNumPages{Journal of Mathematical Analysis and Applications}{22}{3}{618--630}.
\PrintBackRefs{\CurrentBib}

\bibitem [\protect \citeauthoryear {%
Ghossoub%
\ \BBA {} Zhu%
}{%
Ghossoub%
\ \BBA {} Zhu%
}{%
{\protect \APACyear {2024}}%
}]{%
GHOSSOUBSTACK2024}
\APACinsertmetastar {%
GHOSSOUBSTACK2024}%
\begin{APACrefauthors}%
Ghossoub, M.%
\BCBT {}\ \BBA {} Zhu, M\BPBI B.%
\end{APACrefauthors}%
\unskip\
\newblock
\APACrefYearMonthDay{2024}{}{}.
\newblock
{\BBOQ}\APACrefatitle {Stackelberg equilibria with multiple policyholders} {Stackelberg equilibria with multiple policyholders}.{\BBCQ}
\newblock
\APACjournalVolNumPages{Insurance: Mathematics and Economics}{116}{}{189--201}.
\PrintBackRefs{\CurrentBib}

\bibitem [\protect \citeauthoryear {%
Haimes%
}{%
Haimes%
}{%
{\protect \APACyear {1971}}%
}]{%
haimes1971bicriterion}
\APACinsertmetastar {%
haimes1971bicriterion}%
\begin{APACrefauthors}%
Haimes, Y.%
\end{APACrefauthors}%
\unskip\
\newblock
\APACrefYearMonthDay{1971}{}{}.
\newblock
{\BBOQ}\APACrefatitle {On a bicriterion formulation of the problems of integrated system identification and system optimization} {On a bicriterion formulation of the problems of integrated system identification and system optimization}.{\BBCQ}
\newblock
\APACjournalVolNumPages{IEEE Transactions on Systems, Man, and Cybernetics}{}{3}{296--297}.
\PrintBackRefs{\CurrentBib}

\bibitem [\protect \citeauthoryear {%
Herzel%
, Ruzika%
\BCBL {}\ \BBA {} Thielen%
}{%
Herzel%
\ \protect \BOthers {.}}{%
{\protect \APACyear {2021}}%
}]{%
herzel2021}
\APACinsertmetastar {%
herzel2021}%
\begin{APACrefauthors}%
Herzel, A.%
, Ruzika, S.%
\BCBL {}\ \BBA {} Thielen, C.%
\end{APACrefauthors}%
\unskip\
\newblock
\APACrefYearMonthDay{2021}{}{}.
\newblock
{\BBOQ}\APACrefatitle {Approximation Methods for Multiobjective Optimization Problems: A Survey} {Approximation methods for multiobjective optimization problems: A survey}.{\BBCQ}
\newblock
\APACjournalVolNumPages{INFORMS Journal on Computing}{33}{4}{1284--1299}.
\PrintBackRefs{\CurrentBib}

\bibitem [\protect \citeauthoryear {%
Kaluszka%
\ \BBA {} Okolewski%
}{%
Kaluszka%
\ \BBA {} Okolewski%
}{%
{\protect \APACyear {2008}}%
}]{%
Kaluszka2008}
\APACinsertmetastar {%
Kaluszka2008}%
\begin{APACrefauthors}%
Kaluszka, M.%
\BCBT {}\ \BBA {} Okolewski, A.%
\end{APACrefauthors}%
\unskip\
\newblock
\APACrefYearMonthDay{2008}{}{}.
\newblock
{\BBOQ}\APACrefatitle {An extension of {A}rrow's result on optimal reinsurance contract} {An extension of {A}rrow's result on optimal reinsurance contract}.{\BBCQ}
\newblock
\APACjournalVolNumPages{The Journal of Risk and Insurance}{75}{2}{275--288}.
\PrintBackRefs{\CurrentBib}

\bibitem [\protect \citeauthoryear {%
Laumanns%
, Thiele%
\BCBL {}\ \BBA {} Zitzler%
}{%
Laumanns%
\ \protect \BOthers {.}}{%
{\protect \APACyear {2006}}%
}]{%
LAUMANNS2006}
\APACinsertmetastar {%
LAUMANNS2006}%
\begin{APACrefauthors}%
Laumanns, M.%
, Thiele, L.%
\BCBL {}\ \BBA {} Zitzler, E.%
\end{APACrefauthors}%
\unskip\
\newblock
\APACrefYearMonthDay{2006}{}{}.
\newblock
{\BBOQ}\APACrefatitle {An efficient, adaptive parameter variation scheme for metaheuristics based on the epsilon-constraint method} {An efficient, adaptive parameter variation scheme for metaheuristics based on the epsilon-constraint method}.{\BBCQ}
\newblock
\APACjournalVolNumPages{European Journal of Operational Research}{169}{3}{932--942}.
\PrintBackRefs{\CurrentBib}

\bibitem [\protect \citeauthoryear {%
Lo%
\ \BBA {} Tang%
}{%
Lo%
\ \BBA {} Tang%
}{%
{\protect \APACyear {2019}}%
}]{%
lo2019pareto}
\APACinsertmetastar {%
lo2019pareto}%
\begin{APACrefauthors}%
Lo, A.%
\BCBT {}\ \BBA {} Tang, Z.%
\end{APACrefauthors}%
\unskip\
\newblock
\APACrefYearMonthDay{2019}{}{}.
\newblock
{\BBOQ}\APACrefatitle {Pareto-optimal reinsurance policies in the presence of individual risk constraints} {Pareto-optimal reinsurance policies in the presence of individual risk constraints}.{\BBCQ}
\newblock
\APACjournalVolNumPages{Annals of Operations Research}{274}{1}{395--423}.
\PrintBackRefs{\CurrentBib}

\bibitem [\protect \citeauthoryear {%
Miller%
}{%
Miller%
}{%
{\protect \APACyear {1972}}%
}]{%
miller1972}
\APACinsertmetastar {%
miller1972}%
\begin{APACrefauthors}%
Miller, R\BPBI B.%
\end{APACrefauthors}%
\unskip\
\newblock
\APACrefYearMonthDay{1972}{}{}.
\newblock
{\BBOQ}\APACrefatitle {Insurance contracts, as two-person games} {Insurance contracts, as two-person games}.{\BBCQ}
\newblock
\APACjournalVolNumPages{Management Science}{18}{7}{444--447}.
\PrintBackRefs{\CurrentBib}

\bibitem [\protect \citeauthoryear {%
Pareto%
}{%
Pareto%
}{%
{\protect \APACyear {1906}}%
}]{%
pareto1906manuale}
\APACinsertmetastar {%
pareto1906manuale}%
\begin{APACrefauthors}%
Pareto, V.%
\end{APACrefauthors}%
\unskip\
\newblock
\APACrefYear{1906}.
\newblock
\APACrefbtitle {Manuale di {E}conomia {P}olitica} {Manuale di {E}conomia {P}olitica}.
\newblock
\APACaddressPublisher{Milan}{Societ\'a Editrice Libraria}.
\newblock
\APACrefnote{English translation by A. S. Schwier, \textit{Manual of Political Economy}, New York: Kelley Publishers, 1971}
\PrintBackRefs{\CurrentBib}

\bibitem [\protect \citeauthoryear {%
Pascoletti%
\ \BBA {} Serafini%
}{%
Pascoletti%
\ \BBA {} Serafini%
}{%
{\protect \APACyear {1984}}%
}]{%
PascolettiA1984}
\APACinsertmetastar {%
PascolettiA1984}%
\begin{APACrefauthors}%
Pascoletti, A.%
\BCBT {}\ \BBA {} Serafini, P.%
\end{APACrefauthors}%
\unskip\
\newblock
\APACrefYearMonthDay{1984}{}{}.
\newblock
{\BBOQ}\APACrefatitle {Scalarizing vector optimization problems} {Scalarizing vector optimization problems}.{\BBCQ}
\newblock
\APACjournalVolNumPages{Journal of {O}ptimization {T}heory and {A}pplications}{42}{4}{499--524}.
\PrintBackRefs{\CurrentBib}

\bibitem [\protect \citeauthoryear {%
Ravanelli%
\ \BBA {} Svindland%
}{%
Ravanelli%
\ \BBA {} Svindland%
}{%
{\protect \APACyear {2014}}%
}]{%
ravanelli2014comonotone}
\APACinsertmetastar {%
ravanelli2014comonotone}%
\begin{APACrefauthors}%
Ravanelli, C.%
\BCBT {}\ \BBA {} Svindland, G.%
\end{APACrefauthors}%
\unskip\
\newblock
\APACrefYearMonthDay{2014}{}{}.
\newblock
{\BBOQ}\APACrefatitle {Comonotone {P}areto optimal allocations for law invariant robust utilities on ${L}^1$} {Comonotone {P}areto optimal allocations for law invariant robust utilities on ${L}^1$}.{\BBCQ}
\newblock
\APACjournalVolNumPages{Finance and Stochastics}{18}{1}{249--269}.
\PrintBackRefs{\CurrentBib}

\bibitem [\protect \citeauthoryear {%
Rockafellar%
}{%
Rockafellar%
}{%
{\protect \APACyear {1997}}%
}]{%
rockafellar1997convex}
\APACinsertmetastar {%
rockafellar1997convex}%
\begin{APACrefauthors}%
Rockafellar, R.%
\end{APACrefauthors}%
\unskip\
\newblock
\APACrefYear{1997}.
\newblock
\APACrefbtitle {Convex {A}nalysis} {Convex {A}nalysis}.
\newblock
\APACaddressPublisher{}{Princeton University Press}.
\PrintBackRefs{\CurrentBib}

\bibitem [\protect \citeauthoryear {%
Sung%
, Yam%
, Yung%
\BCBL {}\ \BBA {} Zhou%
}{%
Sung%
\ \protect \BOthers {.}}{%
{\protect \APACyear {2011}}%
}]{%
sung2011}
\APACinsertmetastar {%
sung2011}%
\begin{APACrefauthors}%
Sung, K\BPBI C\BPBI J.%
, Yam, S\BPBI C\BPBI P.%
, Yung, S\BPBI P.%
\BCBL {}\ \BBA {} Zhou, J\BPBI H.%
\end{APACrefauthors}%
\unskip\
\newblock
\APACrefYearMonthDay{2011}{}{}.
\newblock
{\BBOQ}\APACrefatitle {Behavioral optimal insurance} {Behavioral optimal insurance}.{\BBCQ}
\newblock
\APACjournalVolNumPages{Insurance: Mathematics and Economics}{49}{3}{418--428}.
\PrintBackRefs{\CurrentBib}

\bibitem [\protect \citeauthoryear {%
Wang%
, Boonen%
, Jiang%
\BCBL {}\ \BBA {} Zhang%
}{%
Wang%
\ \protect \BOthers {.}}{%
{\protect \APACyear {2026}}%
}]{%
WANG2026261}
\APACinsertmetastar {%
WANG2026261}%
\begin{APACrefauthors}%
Wang, W.%
, Boonen, T\BPBI J.%
, Jiang, W.%
\BCBL {}\ \BBA {} Zhang, Y.%
\end{APACrefauthors}%
\unskip\
\newblock
\APACrefYearMonthDay{2026}{}{}.
\newblock
{\BBOQ}\APACrefatitle {Optimal insurance design under distortion risk measures with variance constraint} {Optimal insurance design under distortion risk measures with variance constraint}.{\BBCQ}
\newblock
\APACjournalVolNumPages{European Journal of Operational Research}{333}{1}{261--278}.
\PrintBackRefs{\CurrentBib}

\bibitem [\protect \citeauthoryear {%
Wierzbicki%
}{%
Wierzbicki%
}{%
{\protect \APACyear {1980}}%
}]{%
wierzbicki1980}
\APACinsertmetastar {%
wierzbicki1980}%
\begin{APACrefauthors}%
Wierzbicki, A\BPBI P.%
\end{APACrefauthors}%
\unskip\
\newblock
\APACrefYearMonthDay{1980}{}{}.
\newblock
{\BBOQ}\APACrefatitle {The Use of Reference Objectives in Multiobjective Optimization} {The use of reference objectives in multiobjective optimization}.{\BBCQ}
\newblock
\BIn{} G.~Fandel\ \BBA {} T.~Gal\ (\BEDS), \APACrefbtitle {Multiple {C}riteria {D}ecision {M}aking {T}heory and {A}pplication} {Multiple {C}riteria {D}ecision {M}aking {T}heory and {A}pplication}\ (\BPGS\ 468--486).
\newblock
\APACaddressPublisher{Berlin, Heidelberg}{Springer Berlin Heidelberg}.
\PrintBackRefs{\CurrentBib}

\bibitem [\protect \citeauthoryear {%
Xia%
}{%
Xia%
}{%
{\protect \APACyear {2004}}%
}]{%
xia2004multi}
\APACinsertmetastar {%
xia2004multi}%
\begin{APACrefauthors}%
Xia, J.%
\end{APACrefauthors}%
\unskip\
\newblock
\APACrefYearMonthDay{2004}{}{}.
\newblock
{\BBOQ}\APACrefatitle {Multi-agent investment in incomplete markets} {Multi-agent investment in incomplete markets}.{\BBCQ}
\newblock
\APACjournalVolNumPages{Finance and Stochastics}{8}{2}{241--259}.
\PrintBackRefs{\CurrentBib}

\bibitem [\protect \citeauthoryear {%
Xu%
, Zhou%
\BCBL {}\ \BBA {} Zhuang%
}{%
Xu%
\ \protect \BOthers {.}}{%
{\protect \APACyear {2019}}%
}]{%
Xu2019}
\APACinsertmetastar {%
Xu2019}%
\begin{APACrefauthors}%
Xu, Z\BPBI Q.%
, Zhou, X\BPBI Y.%
\BCBL {}\ \BBA {} Zhuang, S\BPBI C.%
\end{APACrefauthors}%
\unskip\
\newblock
\APACrefYearMonthDay{2019}{}{}.
\newblock
{\BBOQ}\APACrefatitle {Optimal insurance under rank-dependent utility and incentive compatibility} {Optimal insurance under rank-dependent utility and incentive compatibility}.{\BBCQ}
\newblock
\APACjournalVolNumPages{Mathematical Finance}{29}{2}{659--692}.
\PrintBackRefs{\CurrentBib}

\bibitem [\protect \citeauthoryear {%
Zhuang%
, Weng%
, Tan%
\BCBL {}\ \BBA {} Assa%
}{%
Zhuang%
\ \protect \BOthers {.}}{%
{\protect \APACyear {2016}}%
}]{%
Zhuang2016}
\APACinsertmetastar {%
Zhuang2016}%
\begin{APACrefauthors}%
Zhuang, S\BPBI C.%
, Weng, C.%
, Tan, K\BPBI S.%
\BCBL {}\ \BBA {} Assa, H.%
\end{APACrefauthors}%
\unskip\
\newblock
\APACrefYearMonthDay{2016}{}{}.
\newblock
{\BBOQ}\APACrefatitle {Marginal indemnification function formulation for optimal reinsurance} {Marginal indemnification function formulation for optimal reinsurance}.{\BBCQ}
\newblock
\APACjournalVolNumPages{Insurance: Mathematics and Economics}{67}{}{65--76}.
\PrintBackRefs{\CurrentBib}

\end{thebibliography}
}}

\begin{appendices}



\section{Proof of Lemma~\ref{rhoconvexI}}\label{ap:rhoconvexI}
\begin{proof} Without loss of generality, we prove this result for $m=2$.
Since $\mathcal{C}$ is convex, let $\I=(I_1, I_2, \dots, I_n),\;\mathbf{J}= (J_1, J_2, \dots , J_n) \in \mathcal{C}$, and  $\theta_1\I + \theta_2 \mathbf{J}\in\mathcal{C}$, with $\theta_1\in [0,1]$, where $\theta_2=1-\theta_1\in [0,1]$. For policyholder $i=1,2,\dots, n$,
\begin{align*}
&\rho_i\left(Y_i(\theta_1 I_i + \theta_2 J_i)\right) = \rho_i\left( X_i - (\theta_1I_{i}(X_i) + \theta_2 J_i(X_i)) + \Pi_i( \theta_1 I_{i }(X_i) + \theta_2 J_{i}(X_i) ) \right)\\
& =  \rho_i\left( X_i - (\theta_1I_{i}(X_i) + \theta_2 J_i(X_i)) + \theta_1 \Pi_i(   I_{i }(X_i))  + \theta_2 \Pi_i( J_{i}(X_i) ) \right)\\
& = \rho_i\left( \theta_1 (X_i - I_{i}(X_i) +  \Pi_i(  I_{i}(X_i)) ) + \theta_2 (X_i- J_{i}(X_i) +\Pi_i( J_{i }(X_i) ) ) \right)\\
& \leq \theta_1 \rho_i\left( X_i - I_{i}(X_i) +  \Pi_i(  I_{i}(X_i)) \right) + \theta_2 \rho_i\left( X_i- J_{i}(X_i) +\Pi_i( J_{i }(X_i) ) \right) \\
& = \theta_1 \rho_i \left( Y_i\left(I_{i}\right) \right) + \theta_2 \rho_i \left( Y_i\left(J_{i}\right) \right) \,.
\end{align*}
Similarly, for the insurer,  
\begin{align*}
\rho_{n+1} & \left( Y_{n+1}(\theta_1 \mathbf{I} +\theta_2 \mathbf{J})\right) =\rho_{n+1}\left(  \sum_{i=1}^n \left(\theta_1 I_i(X_i) + \theta_2 J_i(X_i)\right)  - \sum_{i=1}^n \Pi_i (\theta_1 I_{i}(X_i) + \theta_2 J_i(X_i))   \right) \\ 
&=  \rho_{n+1}\left(  \sum_{i=1}^n \left(\theta_1 I_i(X_i) + \theta_2 J_i(X_i)\right)  - \sum_{i=1}^n \left( \theta_1\Pi_i ( I_{i}(X_i))  + \theta_2 \Pi_i(J_i(X_i)) \right)\right)\\
&= \rho_{n+1}\left(\theta_1  \sum_{i=1}^n \left( I_i(X_i) - \Pi_i ( I_{i}(X_i))\right) + \theta_2\sum_{i=1}^n      \left( J_i(X_i)  
 - \Pi_i(J_i(X_i)) \right)  \right)\\
& \leq \theta_1 \rho_{n+1}\left(  \sum_{i=1}^n  \left( I_i(X_i) - \Pi_i ( I_{i}(X_i))\right) \right) + \theta_2 \rho_{n+1}\left( \sum_{i=1}^n  \left( J_i(X_i)  - \Pi_i(J_i(X_i)) \right)  \right)\\
&= \theta_1\rho_{n+1}\left( Y_{n+1}\left( \mathbf{I} \right) \right) + \theta_2\rho_{n+1}\left( Y_{n+1}\left(\mathbf{J}\right) \right)\, ,
\end{align*}
where, for each agent, the second equality is due to the semi-linearity of the premium principles, while the inequality is due to the convexity of the risk measures. 
\end{proof}

\section{Details of Remark~\ref{rem:Anpartition}}\label{ap:Anpartition}
Property (i) holds because each $\bm\lambda^{(j)}\in \Delta_n[0,1]$ is a non-zero column, for all $j=1,2,\dots, n(\Lambda)$. To prove the disjointness in property (ii), suppose that there exists $i\in \Nn$ such that $i\in \mathcal{G}_j\cap \mathcal{G}_k$, for some distinct $j,k\in\{1,2,\dots,n(\Lambda)\}$. Then both $\lambda_i^{(j)}$ and $\lambda_i^{(k)}$ are non-zero, and hence $\nzp(\Lambda[i,\cdot]^T)\supseteq\{j,k\}$, which contradicts the second condition in the definition of $\mathcal{A}_n$. Therefore, the sets $\mathcal{G}_j$, for $j=1,2,\dots, n(\Lambda)$, are pairwise disjoint. To prove the covering in property (ii), suppose that there exists $i\in \Nn$ such that $i\notin \mathcal{G}_j$, for all $j=1,2,\dots, n(\Lambda)$. Then $\lambda_i^{(j)}=0$ for all $j=1,2,\dots,n(\Lambda)$, so that $\nzp(\Lambda[i, \cdot]^T) =\emptyset$, again contradicting the second condition in the definition of $\mathcal{A}_n$. Hence, $\Nn\subseteq \cup_{j=1}^{n(\Lambda)}\mathcal{G}_j $. The reverse inclusion $\cup_{j=1}^{n(\Lambda)}\mathcal{G}_j\subseteq\Nn$ is immediate from the definition of each set $\mathcal{G}_j$.

\section{Proof of Proposition~\ref{n_K_lambda_c_convex}}\label{ap:n_K_lambda_c_convex}
\begin{proof}
Let $\mathbf{I}^{(1)}=(I^{(1)}_1,I^{(1)}_2,\dots, I^{(1)}_n),\mathbf{I}^{(2)}=(I^{(2)}_1,I^{(2)}_2,\dots, I^{(2)}_n)\in K(\bm\lambda;\mathcal{C})$, and let $\theta_1,\theta_2\in[0,1]$ with $\theta_1+\theta_2=1$. Since $\mathcal{C}$ is convex, we have $\theta_1\I^{(1)}+\theta_2\I^{(2)}\in\mathcal{C}$. For any $\I=\left(I_1,I_2,\dots, I_n\right)\in\mathcal{C}$,
\begin{align*}
&\;\sum_{i=1}^{n}\lambda_i \rho_i\left(Y_i\left(\theta_1 I_i^{(1)}+\theta_2I_i^{(2)}\right)\right)+\lambda_{n+1} \rho_{n+1}\left(Y_{n+1}\left(\theta_1\I^{(1)}+\theta_2\I^{(2)}\right)\right)
\\\leq&\;\sum_{i=1}^{n}\lambda_i\left(\theta_1\rho_i\left(Y_i\left(I_i^{(1)}\right)\right)+\theta_2\rho_i\left(Y_i\left(I_i^{(2)}\right)\right)\right)\\&\;+\lambda_{n+1}\left(\theta_1\rho_{n+1}\left(Y_{n+1}\left(\I^{(1)}\right)\right)+\theta_2\rho_{n+1}\left(Y_{n+1}\left(\I^{(2)}\right)\right)\right)
\\=&\;\theta_1\left(\sum_{i=1}^{n}\lambda_i \rho_i\left(Y_i\left(I_i^{(1)}\right)\right)+\lambda_{n+1} \rho_{n+1}\left(Y_{n+1}\left(\I^{(1)}\right)\right)\right)\\&\;+\theta_2\left(\sum_{i=1}^{n}\lambda_i \rho_i\left(Y_i\left(I_i^{(2)}\right)\right)+\lambda_{n+1} \rho_{n+1}\left(Y_{n+1}\left(\I^{(2)}\right)\right)\right)
\\\leq&\;\theta_1\left(\sum_{i=1}^{n}\lambda_i \rho_i\left(Y_i\left(I_i\right)\right)+\lambda_{n+1} \rho_{n+1}\left(Y_{n+1}\left(\I\right)\right)\right)+\theta_2\left(\sum_{i=1}^{n}\lambda_i \rho_i\left(Y_i\left(I_i\right)\right)+\lambda_{n+1} \rho_{n+1}\left(Y_{n+1}\left(\I\right)\right)\right)
\\=&\;\sum_{i=1}^{n}\lambda_i \rho_i\left(Y_i\left(I_i\right)\right)+\lambda_{n+1} \rho_{n+1}\left(Y_{n+1}\left(\I\right)\right).
\end{align*}
The first inequality follows from Lemma \ref{rhoconvexI}, and the second one follows from the fact that $\mathbf{I}^{(1)},\mathbf{I}^{(2)}\in K(\bm\lambda;\mathcal{C})$. Hence, $\theta_1\I^{(1)}+\theta_2\I^{(2)}\in K(\bm\lambda;\mathcal{C})$, thus $K(\bm\lambda;\mathcal{C})$ is convex.
\end{proof}

\section{Example of Sequential Optimizations}\label{ap:ex_2}
To continue with Example~\ref{ex:A2}, we illustrate the sequential problems for the matrices in $\mathcal{A}_2$.
\begin{itemize}
\item For any $\Lambda \in \mathcal{A}_2^{(1)}$, $n(\Lambda) = 1$, $\Lambda[\cdot, 1] = (\lambda_1, \lambda_2, \lambda_3)^T\in  \Delta_2(0,1)$, thus
\[\begin{aligned}
K^\ast (\Lambda) 
& = K^{\ast (1)} ((\lambda_1, \lambda_2, \lambda_3)^T) \\
& = K((\lambda_1, \lambda_2, \lambda_3)^T;\mathcal{I}_2) = \argmin_{\I=(I_1, I_2)\in \mathcal{I}_2} \{ \lambda_1 \rho_1(Y_1) +  \lambda_2 \rho_2(Y_2) + \lambda_3 \rho_3(Y_3) \}.
\end{aligned}\]
\item For any $\Lambda \in \mathcal{A}_2^{(2)}$, $n(\Lambda) = 2$, $\Lambda[\cdot, 1] =(\lambda_1,\lambda_2,0)^T $, with $\lambda_1,\lambda_2\in (0,1)$ and $\lambda_1+\lambda_2=1$, $\Lambda[\cdot,2] = (0,0,1)^T$, thus
\[
K^{\ast(1)} ( (\lambda_1, \lambda_2, 0)^T ) = K((\lambda_1, \lambda_2, 0)^T;\mathcal{I}_2) = \argmin_{\I=(I_1, I_2)\in \mathcal{I}_2} \{ \lambda_1 \rho_1(Y_1) +  \lambda_2 \rho_2(Y_2) \} \, ,
\]
\[
\begin{aligned}
K^\ast (\Lambda) 
& = K^{\ast (2)} (\Lambda[\cdot, 1:2]) \\
& = K((0, 0, 1)^T; K^{\ast(1)} ( (\lambda_1, \lambda_2, 0)^T )) = \argmin_{\I=(I_1, I_2)\in K^{\ast(1)} ( (\lambda_1, \lambda_2, 0)^T )}   \rho_3(Y_3).
\end{aligned}\]
For $\Lambda \in \mathcal{A}_2^{(k)} $ with $k=3,4$, the illustration is similar and is therefore omitted.
\item For any $\Lambda \in \mathcal{A}_2^{(5)}$, $n(\Lambda) = 2$, $\Lambda[\cdot,1] = (1,0,0)^T$, $\Lambda[\cdot, 2] = (0,\lambda_2, \lambda_3)^T $, with $\lambda_2,\lambda_3\in (0,1)$ and $\lambda_2+\lambda_3=1$, thus
\[
K^{\ast(1)} ( (1, 0, 0)^T ) = K((1,0,0)^T;\mathcal{I}_2) = \argmin_{\I = (I_1, I_2)\in \mathcal{I}_2}  \rho_1(Y_1)   \,, 
\]
\[\begin{aligned}
K^\ast (\Lambda) 
& = K^{\ast (2)} (\Lambda[\cdot, 1:2]) \\
& = K((0, \lambda_2, \lambda_3)^T; K^{\ast(1)} ( (1, 0, 0)^T )) = \argmin_{\I = (I_1, I_2)\in K^{\ast(1)} ( (1, 0, 0)^T )} \{ \lambda_2 \rho_2(Y_2) +  \lambda_3 \rho_3(Y_3) \}\, .
\end{aligned}\]
For $\Lambda \in \mathcal{A}_2^{(k)}$ with $k=6,7$, the illustration is similar and is therefore omitted.
\item For any $\Lambda\in \mathcal{A}_2^{(8)}$, $n(\Lambda) = 3$, $\Lambda[\cdot,1] = (1,0,0)^T$, $\Lambda[\cdot, 2] = (0,1,0)^T$, $\Lambda[\cdot,3]=(0,0,1)^T$, 
thus
\[
K^{\ast(1)} ( (1, 0, 0)^T ) = K((1,0,0)^T;\mathcal{I}_2) = \argmin_{\I=(I_1, I_2)\in \mathcal{I}_2}  \rho_1(Y_1)  \, ,  
\]
\[
K^{\ast (2)} (\Lambda[\cdot, 1:2]) 
 = K((0, 1, 0)^T; K^{\ast(1)} ( (1, 0, 0)^T )) = \argmin_{\I = (I_1, I_2)\in K^{\ast(1)} ( (1, 0, 0)^T )} \rho_2(Y_2)\, ,
\]
\[
\begin{aligned}
K^\ast (\Lambda) 
& = K^{\ast (3)} (\Lambda[\cdot, 1:3]) \\
& = K((0, 0, 1)^T; K^{\ast(2)} (\Lambda[\cdot, 1:2])) = \argmin_{\I=(I_1, I_2)\in K^{\ast(2)} ( \Lambda[\cdot, 1:2] )}  \rho_3(Y_3) \, .
\end{aligned}
\]
For $\Lambda \in \mathcal{A}_2^{(k)}$ with $k=9,10,\dots,13$, the illustration is similar and is therefore omitted.
\end{itemize}

\section{Dimension for Set of Sequential Matrices}\label{ap:dimension_An}
Let $\mathrm{d}(\mathcal{A}_{-1})={\mathrm{d}(\mathcal{A}_0)=1}$. Let $n\in\mathbb{N}$. The dimension of $\mathcal{A}_n$ is calculated recursively:
\begin{equation*}
\mathrm{d}(\mathcal{A}_{n}) = \sum_{k=1}^{n+1} C_{n+1, k} \cdot \mathrm{d}(\mathcal{A}_{n-k})\, ,
\end{equation*}
where, for $k=1,2,\dots, n+1$, $ C_{n+1, k}= \binom{n+1}{k}$ is the binomial coefficient. The recursion is intuitive: among the $n+1$ agents, there are $C_{n+1, k}$ combinations for picking $k$ agents forming their objectives $\lambda_{i_1}\rho_{i_1}(Y_{i_1}) +\cdots +\lambda_{i_k}\rho_{i_k}(Y_{i_k}) $, for $i_1,i_2,\dots,i_k\in\mathcal{N}_n$, with $0<\lambda_{i_l}<1$ and $\sum_{l=1}^k \lambda_{i_l} = 1$; for each combination of $k$ agents being picked, the remaining $n+1-k$ agents arise $\mathrm{d}(\mathcal{A}_{n-k})$ distinct classes of their ordered sequences.
\section{Pareto Optimality in Euclidean Space}\label{ap:POsets}
Let $C\subseteq\mathbb{R}^n$. The convex hull and affine hull of $C$ are respectively defined below:
\[
\overline{C}=\mathrm{conv}(C) = \left\lbrace  \sum_{j=1}^m \theta_j\mathbf{y}_j: m\in\mathbb{N},\, \mathbf{y}_j\in C, \, \theta_j \geq 0,\, \sum_{j=1}^m\theta_j = 1 \right\rbrace\, ,
\]
\[
\mathrm{aff}(C)=\left\lbrace \sum_{j=1}^m \theta_j \mathbf{y}_j: m\in \mathbb{N}\,, \mathbf{y}_j\in C\,, \theta_j\in \mathbb{R}\,, \sum_{j=1}^m\theta_j = 1\right\rbrace \, .\]
The relative interior of $C$ is denoted by $\mathrm{ri}(C)$, which is defined as the interior of $C$ within the affine hull $\mathrm{aff}(C)$. If $C$ is $n$-dimensional, then $\mathrm{aff}(C) = \mathbb{R}^n$, and $\mathrm{ri}(C)=\mathrm{int}(C)$ which is the interior of $C$ within $\mathbb{R}^n$.

A point $\mathbf{y}^\ast=(y^\ast_1,y^\ast_2,\dots, y^\ast_n)^T\in C$ is Pareto optimal in $C$, if there does not exist another point $\mathbf{y}=(y_1,y_2,\dots, y_n)^T \in C$ such that $y_i\leq y^\ast_i$, for $i=1,2,\dots,n$, with at least one strict inequality; in terms of vectors, there does not exist $\mathbf{y}\in C$ such that $\mathbf{y}\leq \mathbf{y}^\ast$ with $\mathbf{y}\neq \mathbf{y}^\ast$, where the inequalities are component-wise.

Let $\mathbf{y}^\ast$ be a Pareto optimal point in $C$, and consider the set $T_{\mathbf{y}^\ast}= \{\mathbf{y} = (y_1,y_2,\dots , y_n)^T\in\mathbb{R}^n: \mathbf{y}\leq \mathbf{y}^\ast \}$. Since $T_{\mathbf{y}^\ast}$ is $n$-dimensional, $\mathrm{ri}(T_{\mathbf{y}^\ast}) = \mathrm{int}(T_{\mathbf{y}^\ast})$. 
For any $\mathbf{y}\in \mathrm{ri}(T_{\mathbf{y}^\ast}) = \mathrm{int}(T_{\mathbf{y}^\ast}) $, $\mathbf{y}< \mathbf{y}^\ast $; if such $\mathbf{y}\in\mathrm{ri}(C)$, then it contradicts that $\mathbf{y}^\ast$ is Pareto optimal in $C$. Therefore, $\mathrm{ri}(T_{\mathbf{y}^\ast }) \cap \mathrm{ri}(C) =\emptyset$.

Assume that $C$ is convex. By the Separation Theorem (Theorem 11.3 of \cite{rockafellar1997convex}), since $T_{\mathbf{y}^\ast }$ is also convex, by the fact that $\mathrm{ri}(T_{\mathbf{y}^\ast }) \cap \mathrm{ri}(C) =\emptyset$, there exists a vector $\bm{a}=(a_1,a_2,\dots, a_n)^T\in\mathbb{R}^n $, $\bm{a}\neq \mathbf{0}$, such that 
$\sup_{\mathbf{y}\in T_{\mathbf{y}^\ast}} \lbrace \bm{a}^T\mathbf{y}\rbrace\leq\inf_{\mathbf{y}\in C}\lbrace \bm{a}^T\mathbf{y}\rbrace$.

\end{appendices}

\end{document}